\def\isall{1}
\def\isBDVhard{1}

\def\isBDeEhard{1}
\def\islp{1}

\def\islocalratio{1}
\def\iswcolhard{1}
\def\istwhard{1}
\def\isstarforest{1}
\def\istwappx{1}
\def\isstarforesthard{1}
\def\isbde{1}
\def\isgreedy{1}

\newif\ifcomments\commentsfalse
\newif\ifappendix\appendixfalse
\newif\iffull\fulltrue

\documentclass[letterpaper,11pt]{article}
\usepackage[left=1in,top=1in,right=1in,bottom=1in,nohead]{geometry}

\usepackage[T1]{fontenc}
\usepackage{tikz}
\usepackage{algorithm, algpseudocode}
\usepackage{amsmath, amssymb, amsthm, enumerate, hyperref}
\usepackage{color}
\usepackage{boxedminipage}
\usepackage{amsfonts}
\usepackage{graphicx}
\usepackage{xspace}

\makeatletter
\newtheorem*{rep@theorem}{\rep@title}
\newcommand{\newreptheorem}[2]{%
\newenvironment{rep#1}[1]{%
 \def\rep@title{#2 \ref{##1}}%
 \begin{rep@theorem}}%
 {\end{rep@theorem}}}
\makeatother

\newreptheorem{theorem}{Theorem}
\newreptheorem{lemma}{Lemma}

\usepackage{paralist}

\usetikzlibrary{shapes, arrows, automata, positioning, fit, shadows, decorations.pathmorphing}
\usetikzlibrary{decorations.pathreplacing}

\usepackage{xintexpr}

\usepackage{mathrsfs}

\ifcomments
\newcommand{\ali}[1]{{\color{red}/* Ali: #1 */}}
\newcommand{\qq}[1]{{\color{cyan}/* Quanquan: #1 */}}
\newcommand{\tim}[1]{{\color{blue}/* Timothy: #1 */}}
\newcommand{\kyle}[1]{{\color{green}/* Kyle: #1 */}}
\newcommand{\drew}[1]{{\color{pink}/* Drew: #1 */}}
\else
\newcommand{\ali}[1]{\ignorespaces}
\newcommand{\qq}[1]{\ignorespaces}
\newcommand{\tim}[1]{\ignorespaces}
\newcommand{\kyle}[1]{\ignorespaces}
\newcommand{\drew}[1]{\ignorespaces}
\fi

\renewcommand\emph[1]{\textbf{\textit{\boldmath #1}}}

\newtheorem{definition}{Definition}[section]
\newtheorem{theorem}{Theorem}[section]
\newtheorem{lemma}[theorem]{Lemma}
\newtheorem{corollary}[theorem]{Corollary}

\newtheorem{claim}[theorem]{Claim}

\usepackage{color,tabularx,booktabs,xspace}

\let\realbfseries=\bfseries
\def\bfseries{\realbfseries\boldmath}

\let\epsilon=\varepsilon

\def\mC{\mathcal{C}}
\def\mCp{\mathcal{C}_\lambda}
\def\mF{\mathcal{F}}

\def\set#1{\{#1\}}
\newcommand{\arc}[1]{\overrightarrow{#1}}
\newcommand{\eps}{\varepsilon}%
\newcommand{\editnumber}{\ensuremath{k}\xspace}
\newcommand{\editset}{\ensuremath{X}\xspace}

\newcommand{\targetdegree}{\ensuremath{d}\xspace}
\newcommand{\Td}{\targetdegree}
\newcommand{\targetdegeneracy}{\ensuremath{r}\xspace}
\newcommand{\Tr}{\targetdegeneracy}
\newcommand{\targettreewidth}{\ensuremath{w}\xspace}
\newcommand{\Tw}{\targettreewidth}
\newcommand{\targettreedepth}{\ensuremath{p}\xspace}
\newcommand{\Tp}{\targettreedepth}
\newcommand{\targetcliquenumber}{\ensuremath{b}\xspace}
\newcommand{\Tcn}{\targetcliquenumber}
\newcommand{\weakcoloringpath}{\ensuremath{c}\xspace}
\newcommand{\Tc}{\weakcoloringpath}
\newcommand{\targetweakcoloring}{\ensuremath{t}\xspace}
\newcommand{\Tt}{\targetweakcoloring}

\DeclareMathOperator*{\inDeg}{\ensuremath{\deg^-}}%
\DeclareMathOperator*{\outDeg}{\ensuremath{\deg^+}}%
\DeclareMathOperator*{\degener}{degen}%
\newcommand{\core}[2]{\ensuremath{\operatorname{core}_{#1}(#2)}\xspace}

\def\Comment#1{\textsl{$\langle\!\langle$#1\/$\rangle\!\rangle$}}
\def\mypar#1{\smallskip\noindent\textbf{#1}}
\def\sP{\mathcal{P}}
\newcommand{\lpMinEdge}{\Problem{DegenEdgeEdit-LP}}
\newcommand{\lpMinVertex}{\Problem{DegenVertexEdit-LP}}

\newcommand{\optsol}[2]{\opt_{#1}(#2)}
\newcommand{\tw}{\mathsf{tw}}
\newcommand{\pw}{\mathsf{pw}}
\newcommand{\sep}{\mathsf{sep}}

\newcommand{\etal}{et al.}

\newcommand{\Problem}[1]{\textsc{#1}\xspace}

\newcommand{\mcU}{\ensuremath{\mathcal{U}}}
\newcommand{\mcF}{\ensuremath{\mathcal{F}}}

\newcommand{\Editfull}{\Problem{$(\mathcal{C}, \psi)$-Edit}}
\newcommand{\Editfullp}{\Problem{$(\mathcal{C}_\lambda, \psi)$-Edit}}
\newcommand{\kPsieditsG}{\psi_k(\psi_{k-1}( \cdots \psi_2(\psi_1(G))\cdots))}

\DeclareMathOperator*{\optop}{OPT}
\newcommand{\opt}{\ensuremath{\optop}}
\DeclareMathOperator*{\costop}{Cost}
\newcommand{\cost}{{\ensuremath{\costop}}}

\newcommand{\yOpt}{\ensuremath{y^{*}}}
\newcommand{\yCanon}{\ensuremath{\widehat{y}}}
\newcommand{\ycan}{\yCanon}
\newcommand{\yopt}{\yOpt}
\newcommand{\yA}{\ensuremath{y}}
\newcommand{\yB}{\ensuremath{y'}}
\newcommand{\xA}{\ensuremath{x}}

\newcommand{\WLOG}{WLOG\xspace}

\newcommand{\scdg}{\ensuremath{\mathtt{scdg}}\xspace} 
\newcommand{\scte}{\ensuremath{\mathtt{T}_e}\xspace} 
\newcommand{\sG}{\ensuremath{{\mathtt G}}\xspace}

\newcommand{\neighb}{\ensuremath{N}}  

\newcommand{\splitg}{\ensuremath{\mathtt{split}}\xspace}

\newcommand{\gDS}{\ensuremath{{\mathtt D}}\xspace}

\newcommand{\scg}{\ensuremath{\mathtt{scg}}\xspace} 
\newcommand{\weg}[1]{\texttt{EG}_{#1}}
\newcommand{\wsg}[1]{\texttt{SG}_{#1}}
\newcommand{\wcol}[2]{\ensuremath{\operatorname{wcol}_{#1}(#2)}\xspace}
\newcommand{\wscore}[3]{\ensuremath{\operatorname{wscore}_{#1}(#2, #3)}\xspace}
\newcommand{\wreach}[4]{\ensuremath{\operatorname{wreach}_{#1}(#2, #3, #4)}\xspace}

\newcommand{\floor}[1]{\left \lfloor #1 \right \rfloor}
\newcommand{\ceil}[1]{\left \lceil #1 \right \rceil}

\newcommand{\ADS}{\Problem{ADS}}
\newcommand{\ADSfull}{\Problem{Annotated Dominating Set}}

\newcommand{\ADSmaybefull}{\Problem{Annotated (\DSradius-)Dominating Set}}

\newcommand{\ALDS}{\Problem{ADS}}
\newcommand{\ALDSfull}{\Problem{Annotated \DSradius-Dominating Set}}

\newcommand{\CDS}{\Problem{CDS}}
\newcommand{\CDSfull}{\Problem{Connected Dominating Set}}

\newcommand{\ACDS}{\Problem{ACDS}}
\newcommand{\ACDSfull}{\Problem{Annotated Connected Dominating Set}}

\newcommand{\MC}{\Problem{MC}}
\newcommand{\MCfull}{\Problem{Max-Cut}}

\newcommand{\VC}{\Problem{VC}}
\newcommand{\VCfull}{\Problem{Vertex Cover}}

\newcommand{\FVS}{\Problem{FVS}}
\newcommand{\FVSfull}{\Problem{Feedback Vertex Set}}

\newcommand{\MMM}{\Problem{MMM}}
\newcommand{\MMMfull}{\Problem{Minimum Maximal Matching}}

\newcommand{\CN}{\Problem{CRN}}
\newcommand{\CNfull}{\Problem{Chromatic Number}}

\newcommand{\uSC}{\Problem{$k$-uSC}}
\newcommand{\uSCfull}{\Problem{$k$-uniform Set Cover}}
\newcommand{\xSC}[1]{\Problem{$#1$-uSC}}

\newcommand{\SC}{\Problem{SC}}
\newcommand{\SCfull}{\Problem{Set Cover}}

\newcommand{\setbound}{\ensuremath{b}}

\newcommand{\DSradius}{\ensuremath{\ell}\xspace}
\newcommand{\ISradius}{\DSradius}

\newcommand{\EDSmaybefull}{\Problem{Edge (\DSradius-)Dominating Set}}
\newcommand{\EDSmaybe}{\Problem{EDS}}

\newcommand{\EDSfullfixed}{\Problem{Edge Dominating Set}}

\newcommand{\EDS}{\Problem{EDS}}

\newcommand{\DS}{\Problem{DS}}
\newcommand{\DSfull}{\Problem{Dominating Set}}

\newcommand{\LDS}{\Problem{\DSradius-DS}}
\newcommand{\LDSfull}{\Problem{\DSradius-Dominating Set}}

\newcommand{\ELDS}{\Problem{\DSradius-EDS}}
\newcommand{\ELDSfull}{\Problem{Edge \DSradius-Dominating Set}}

\newcommand{\DSmaybe}{\Problem{DS}}
\newcommand{\DSmaybefull}{\Problem{(\DSradius-)Dominating Set}}

\newcommand{\DSB}{\Problem{DSB}}
\newcommand{\DSBfull}{\Problem{Minimum Dominating Set-$B$}}

\newcommand{\IS}{\Problem{IS}}
\newcommand{\ISfull}{\Problem{Independent Set}}

\newcommand{\LIS}{\Problem{\ISradius-IS}}
\newcommand{\LISfull}{\Problem{\ISradius-Independent Set}}

\newcommand{\PFDfull}{\Problem{Planar $\mathcal{F}$-Deletion}}

\newcommand{\bHVD}{\Problem{$H$-Vertex-Deletion}}

\newcommand{\bBDDfull}{\Problem{$\targetdegree$-Bounded-Degree Deletion}}

\newcommand{\bBDDV}{\Problem{$\targetdegree$-BDD-V}}
\newcommand{\bBDDVfull}{\Problem{$\targetdegree$-Bounded-Degree Vertex Deletion}}

\newcommand{\bBDDE}{\Problem{$\targetdegree$-BDD-E}}
\newcommand{\bBDDEfull}{\Problem{$\targetdegree$-Bounded-Degree Edge Deletion}}

\newcommand{\bbDEV}{\Problem{$(\beta \targetdegeneracy)$-DE-V}}
\newcommand{\bbDEVfull}{\Problem{$(\beta \targetdegeneracy)$-Degenerate Vertex Deletion}}

\newcommand{\bDEV}{\Problem{$\targetdegeneracy$-DE-V}}
\newcommand{\bDEVfull}{\Problem{$\targetdegeneracy$-Degenerate Vertex Deletion}}

\newcommand{\bDEE}{\Problem{$\targetdegeneracy$-DE-E}}
\newcommand{\bDEEfull}{\Problem{$\targetdegeneracy$-Degenerate Edge Deletion}}

\newcommand{\bDE}{\Problem{$\targetdegeneracy$-DE}}
\newcommand{\bDEfull}{\Problem{$\targetdegeneracy$-Degenerate Deletion}}

\newcommand{\bWCN}{\Problem{\targetweakcoloring-BWE-\Tc}}
\newcommand{\bWCNfull}{\Problem{\targetweakcoloring-Bounded Weak $\Tc$-Coloring Number Editing}}

\newcommand{\bWCNV}{\Problem{\targetweakcoloring-BWE-V-\Tc}}
\newcommand{\bWCNVfull}{\Problem{\targetweakcoloring-Bounded Weak $\Tc$-Coloring Number Vertex-Deletion}}

\newcommand{\bWCNE}{\Problem{\targetweakcoloring-BWE-E-\Tc}}

\newcommand{\bWCNC}{\Problem{\targetweakcoloring-BWE-C-\Tc}}

\newcommand{\bTWV}{\Problem{$\targettreewidth$-TW-V}}
\newcommand{\bTWVfull}{\Problem{$\targettreewidth$-Treewidth Vertex Deletion}}

\newcommand{\TWV}{\Problem{TW-V}}

\newcommand{\bPWV}{\Problem{$\targettreewidth$-PW-V}}

\newcommand{\bTWE}{\Problem{$\targettreewidth$-TW-E}}

\newcommand{\bPWE}{\Problem{$\targettreewidth$-PW-E}}

\newcommand{\bCNV}{\Problem{$\targetcliquenumber$-CN-V}}
\newcommand{\bCNVfull}{\Problem{$\targetcliquenumber$-Clique Number Vertex Deletion}}

\newcommand{\CNV}{\Problem{CN-V}}

\newcommand{\SFVfull}{\Problem{Star Forest Vertex Deletion}}
\newcommand{\SFV}{\Problem{SF-V}}
\newcommand{\SFEfull}{\Problem{Star Forest Edge Deletion}}
\newcommand{\SFE}{\Problem{SF-E}}

\newenvironment{tightcenter}
 {\parskip=0pt\par\nopagebreak\centering}
 {\par\noindent\ignorespacesafterend}

\usepackage{tikz}
\usetikzlibrary{calc}

\usepackage{xspace}
\usepackage{xargs}
\usepackage{xifthen}
\usepackage{amsfonts}
\usepackage{amsmath}
\usepackage{framed}

\usepackage{ctable}
\newlength{\RoundedBoxWidth}
\newsavebox{\GrayRoundedBox}
\newenvironment{GrayBox}[1]%
   {\setlength{\RoundedBoxWidth}{\linewidth-4.5ex}
    \def\boxheading{#1}
    \begin{lrbox}{\GrayRoundedBox}
       \begin{minipage}{\RoundedBoxWidth}%
   }{%
       \end{minipage}
    \end{lrbox}%
    \begin{tightcenter}%
    \begin{tikzpicture}%
       \node(Text)[draw=black!20,fill=white,rounded corners,%
             inner sep=2ex,text width=\RoundedBoxWidth]%
             {\usebox{\GrayRoundedBox}};
        \coordinate(x) at (current bounding box.north west);
        \node [draw=white,rectangle,inner sep=3pt,anchor=north west,fill=white]
        at ($(x)+(10.5pt,.75em)$) {\boxheading};
    \end{tikzpicture}
    \end{tightcenter}\vspace{0pt}%
    \ignorespacesafterend
}

\newenvironment{problem}[2][]{\noindent\ignorespaces%
                                \FrameSep=6pt%
                                \parindent=0pt%
                \vspace*{-.5em}
                \ifthenelse{\isempty{#1}}{%
                  \begin{GrayBox}{\textsc{#2}}%
                }{%
                  \begin{GrayBox}{\textsc{#2} parametrised by~{#1}}%
                }
                \newcommand\Prob{Problem:}%
                \newcommand\Input{Input:}%
                \newcommand\Objective{Objective:}%
                \begin{tabular*}{\columnwidth}{@{\hspace{.5em}} >{\itshape} p{1.4cm} p{0.85\columnwidth} @{}}%
            }{
                \end{tabular*}%
                \end{GrayBox}%
                \vspace*{-.5em}
                \ignorespacesafterend
            }

\newenvironment{probbox}[1]{\noindent\ignorespaces%
                                \FrameSep=6pt%
                                \parindent=0pt%
                \vspace*{-.5em}
                \begin{GrayBox}{\textsc{#1}}%
            }{
                \end{GrayBox}%
                \vspace*{-.5em}
                \ignorespacesafterend
            }

{\makeatletter
 \gdef\xxxmark{%
   \expandafter\ifx\csname @mpargs\endcsname\relax 
     \expandafter\ifx\csname @captype\endcsname\relax 
       \marginpar{xxx}
     \else
       xxx 
     \fi
   \else
     xxx 
   \fi}
 \gdef\xxx{\@ifnextchar[\xxx@lab\xxx@nolab}
 \long\gdef\xxx@lab[#1]#2{\textbf{[\xxxmark #2 ---{\sc #1}]}}
 \long\gdef\xxx@nolab#1{\textbf{[\xxxmark #1]}}
 \ifcomments\else
   \long\gdef\xxx@lab[#1]#2{}\long\gdef\xxx@nolab#1{}%
 \fi
}

\graphicspath{{./figures/}}

\title{\textbf{Structural Rounding:} Approximation Algorithms for Graphs Near an Algorithmically Tractable Class}

\usepackage{authblk}

\author[1]{Erik~D.~Demaine}
\author[2]{Timothy~D.~Goodrich}
\author[2]{Kyle~Kloster}
\author[2]{Brian~Lavallee}
\author[1]{Quanquan~C.~Liu}
\author[2]{Blair~D.~Sullivan}
\author[1]{Ali~Vakilian}
\author[2]{Andrew van der Poel}

\affil[1]{MIT, Cambridge, MA, USA\\
  \protect\url{{edemaine,quanquan,vakilian}@mit.edu}}
\affil[2]{NC State University, Raleigh, NC, USA\\
  \protect\url{{tdgoodri,kakloste,blavall,blair_sullivan,ajvande4}@ncsu.edu}}

\date{}

\usepackage{tocloft}

\newcommand{\stoptocwriting}{%
  \addtocontents{toc}{\protect\setcounter{tocdepth}{-5}}}
\newcommand{\resumetocwriting}{%
  \addtocontents{toc}{\protect\setcounter{tocdepth}{\arabic{tocdepth}}}}
\AtBeginDocument{\addtocontents{toc}{\protect\thispagestyle{empty}}}

\begin{document}

\maketitle


\begin{abstract}

We develop a new framework for generalizing approximation algorithms from
the structural graph algorithm literature so that they apply to graphs
somewhat close to that class
(a scenario we expect is common when working with real-world networks)
while still guaranteeing approximation ratios.
The idea is to \emph{edit} a given graph via vertex- or edge-deletions
to put the graph into an algorithmically tractable class,
apply known approximation algorithms for that class, and then \emph{lift}
the solution to apply to the original graph.
We give a general characterization of when an optimization problem is
amenable to this approach, and show that it includes many
well-studied graph problems, such as \ISfull, \VCfull, \FVSfull,
\MMMfull, \CNfull, \DSmaybefull, \EDSmaybefull, and \CDSfull.

To enable this framework, we develop new editing algorithms that find
the approximately-fewest edits required to bring a given graph into
one of a few important graph classes
(in some cases also approximating the target parameter of the family).
For bounded degeneracy, we obtain a bicriteria $(4,4)$-approximation which also extends to a smoother bicriteria trade-off.
For bounded treewidth, we obtain a bicriteria
$(O(\log^{1.5} n), O(\sqrt{\log w}))$-approximation, and for bounded pathwidth,
we obtain a bicriteria
$(O(\log^{1.5} n), O(\sqrt{\log w} \cdot \log n))$-approximation.
For treedepth~$2$ (related to bounded expansion),
we obtain a $4$-approximation.
We also prove complementary hardness-of-approximation results
assuming $\mathrm{P} \neq \mathrm{NP}$:
in particular, these problems are all log-factor
inapproximable, except the last which is not approximable below some constant factor ($2$ assuming UGC).
\end{abstract}


\thispagestyle{empty}
\setcounter{page}0
\tableofcontents
\thispagestyle{empty}
\setcounter{page}0
\clearpage

\renewcommand{\cfttoctitlefont}{\hfill\Large\itshape}

\ifdefined\isall
  \section{Introduction}\label{sec:introduction}
\stoptocwriting

Network science has empirically established that real-world networks
(social, biological, computer, etc.)\ exhibit significant sparse structure.
Theoretical computer science has shown that graphs with certain structural
properties enable significantly better approximation algorithms for hard problems.
Unfortunately, the experimentally observed structures and the theoretically
required structures are generally not the same: mathematical graph classes
are rigidly defined, while real-world data is noisy and full of exceptions.
This paper provides a framework for extending approximation guarantees
from existing rigid classes to broader, more flexible graph families
that are more likely to include real-world networks.

Specifically, we hypothesize that most real-world networks are in fact
\emph{small perturbations of graphs from a structural class}.
Intuitively, these perturbations may be
exceptions caused by unusual/atypical behavior
(e.g., weak links rarely expressing themselves),
natural variation from an underlying model,
or noise caused by measurement error or uncertainty.
Formally, a graph is \emph{$\gamma$-close} to a structural class $\mC$, where $\gamma \in \mathbb{N}$,
if some $\gamma$ edits (e.g., vertex deletions, edge deletions, or
edge contractions) bring the graph into class~$\mC$.

Our goal is to extend existing approximation algorithms for a structural class
$\mC$ to apply more broadly to graphs $\gamma$-close to $\mC$.
To achieve this goal, we need two algorithmic ingredients:
\begin{enumerate}
\item \textbf{Editing algorithms.}
  Given a graph $G$ that is $\gamma$-close to a structural class $\mC$,
  find a sequence of $f(\gamma)$ edits that edit $G$ into~$\mC$.
  When the structural class is parameterized (e.g., treewidth $\leq w$),
  we may also approximate those parameters.
\item \textbf{Structural rounding algorithms.}
  Develop approximation algorithms for optimization problems on
  graphs $\gamma$-close to a structural class~$\mC$ by
  converting $\rho$-approximate solutions on an edited graph in class~$\mC$ into
  $g(\rho,\gamma)$-approximate solutions on the original graph.
\end{enumerate}

\begin{figure}[b]
  \centering
  \includegraphics[width=0.6\textwidth]{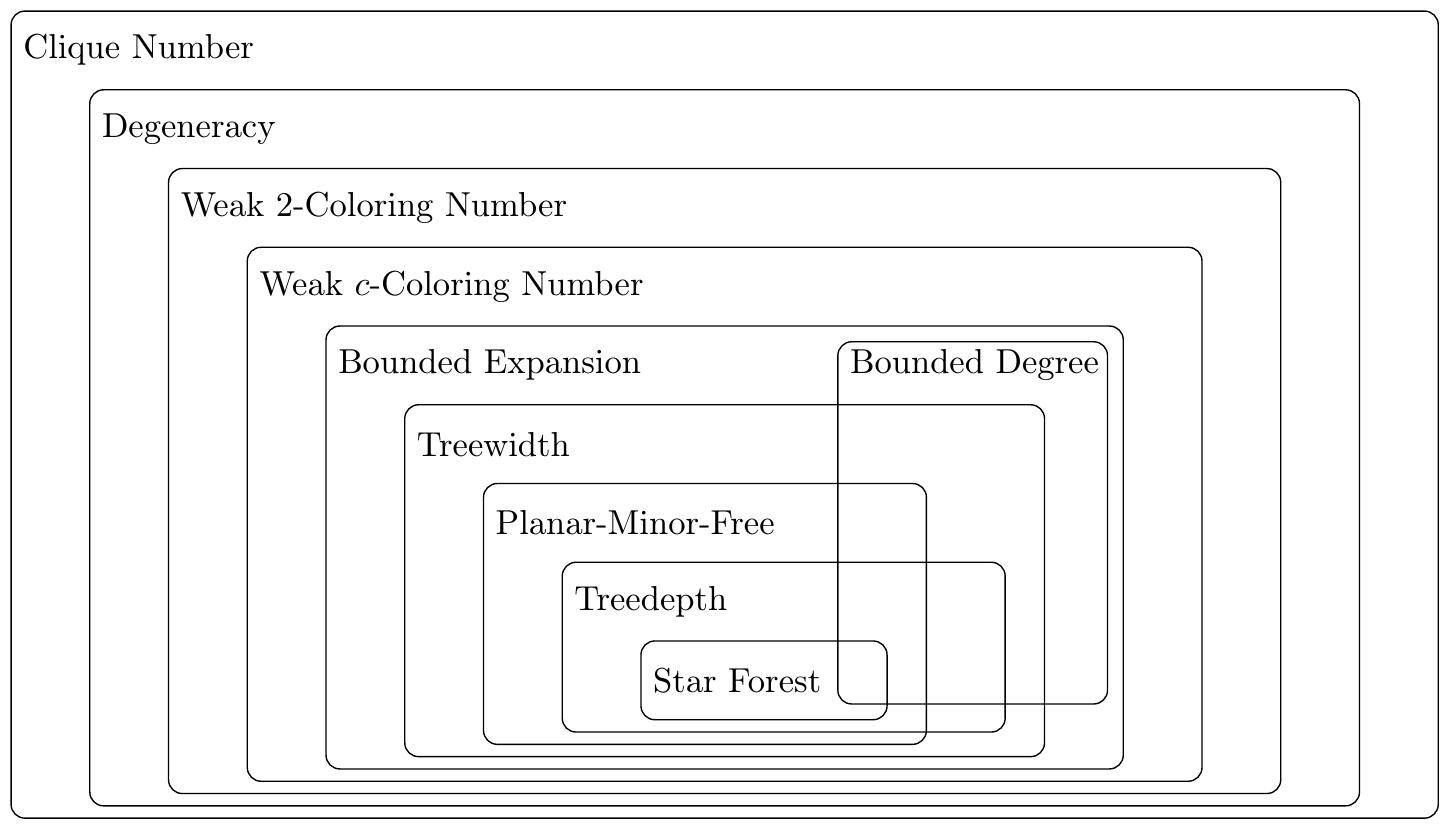}
  \caption{\label{fig:classhierarchy}
    Illustration of hierarchy of structural graph classes used in this paper.
  }
\end{figure}

\subsection{Our Results: Structural Rounding}

In Section~\ref{sec:structural-rounding}, we present a general metatheorem
giving sufficient conditions for an optimization problem to be amenable to
the structural rounding framework.  Specifically, if a problem $\Pi$ has an
approximation algorithm in structural class $\mC$,
the problem and its solutions are ``stable'' under an edit operation,
and there is an $\alpha$-approximate algorithm for editing to $\mC$,
then we get an approximation algorithm for solving $\Pi$ on graphs $\gamma$-close to $\mC$.
The new approximation algorithm incurs an additive error of
$O(\gamma)$, so we preserve PTAS-like $(1+\epsilon)$ approximation factors
provided $\gamma \leq \delta \opt_{\Pi}$ for a suitable constant
$\delta = \delta(\epsilon,\alpha) > 0$.

For example, we obtain $(1+O(\delta \log^{1.5} n))$-approximation algorithms
for \VCfull, \FVSfull, \MMMfull, and \CNfull on graphs
$(\delta \cdot \opt_{\Pi}(G))$-close to having treewidth $\Tw$ via vertex deletions
(generalizing exact algorithms for bounded treewidth graphs);
and we obtain a $(1-4 \delta)/(4\Tr+1)$-approximation algorithm for
\ISfull on graphs $(\delta \cdot \opt_{\Pi}(G))$-close to having degeneracy~$\Tr$
(generalizing a $1/\Tr$-approximation for degeneracy-$\Tr$ graphs).
These results use our new algorithms for editing to treewidth-$\Tw$ and
degeneracy-$\Tr$ graph classes as summarized next.

\subsection{Our Results: Editing}

We develop editing approximation algorithms and/or hardness-of-approximation
results for six well-studied graph classes:
bounded clique number, bounded degeneracy, bounded treewidth and pathwidth, bounded treedepth, bounded weak $\Tc$-coloring number, and bounded degree.
Figure~\ref{fig:classhierarchy} summarizes the relationships among these
classes, and Table~\ref{tab:overview} summarizes our results for each class.

\begin{table}[!h]
\resizebox{\textwidth}{!}{%
\begin{tabular}{ |>{\raggedright}m{23ex}||c|c|  }
 \hline
   & \multicolumn{2}{c|}{\bf Edit Operation $\psi$} \\ \cline{2-3}
 {\bf Graph Family $\mCp$}
   & {\bf Vertex Deletion} & {\bf Edge Deletion} \\
 \hline
 Bounded Degree ($\Td$)  &  \begin{tabular}{@{}c@{}}{\bf \bBDDV} \\ $O(\log \Td)$-approx.\ \cite{ebenlendrapproximation} \\ $(\ln \Td - C\cdot\ln\ln\Td)$-inapprox. \end{tabular}
 							& \begin{tabular}{@{}c@{}}{\bf \bBDDE} \\ Polynomial time \cite{huang17approximate} \end{tabular} \\
 \hline
 Bounded Degeneracy ($\Tr$) & \begin{tabular}{@{}c@{}}{\bf \bDEV} \\ $O(r\log n)$-approx. \\ $\left(\tfrac{4m-\beta \Tr n}{m-\Tr n}, \beta \right)$-approx.\\$\left(\tfrac{1}{\epsilon}, \tfrac{4}{1-2\epsilon}\right)$-approx.\ ($\epsilon < 1/2$)\\ $o(\log (n/\Tr))$-inapprox. \end{tabular}
 						   & \begin{tabular}{@{}c@{}}{\bf \bDEE} \\$O(r\log n)$-approx. \\ -- \\ $\left(\tfrac{1}{\epsilon},\tfrac{4}{1-\epsilon}\right)$-approx.\ ($\epsilon < 1)$\\ $o(\log (n/\Tr))$-inapprox. \end{tabular} \\
 \hline
 Bounded Weak $\Tc$-Coloring Number ($\Tt$)& \begin{tabular}{@{}c@{}}{\bf \bWCNV} \\ -- \\ $o(\Tt)$-inapprox. for $\Tt \in o(\log n)$ \end{tabular}
 										& \begin{tabular}{@{}c@{}}{\bf \bWCNE} \\ -- \\ $o(\Tt)$-inapprox. for $\Tt \in o(\log n)$ \end{tabular} \\
 \hline
 Bounded Treewidth ($\Tw$) & \begin{tabular}{@{}c@{}}{\bf \bTWV} \\ ($O(\log^{1.5}n)$, $O(\sqrt{\log \Tw})$)-\\approx.\\ $o(\log n)$-inapprox. for $\Tw \in \Omega(n^{1/2})$ \end{tabular}
 						 & \begin{tabular}{@{}c@{}}{\bf \bTWE} \\ ($O(\log n \log\log n)$, $O(\log w)$)-\\approx.\ \cite{bansal2017lp}\\ -- \end{tabular} \\
 \hline
 Bounded Pathwidth ($w$) & \begin{tabular}{@{}c@{}}{\bf \bPWV} \\ ($O(\log^{1.5}n)$, $O(\sqrt{\log w} \cdot \log n)$)-\\approx.\\ -- \end{tabular}
 						 & \begin{tabular}{@{}c@{}}{\bf \bPWE} \\ ($O(\log n \log\log n)$, $O(\log w\cdot \log n)$)-\\approx.\ \cite{bansal2017lp}\\ -- \end{tabular} \\
 \hline
 Star Forest & \begin{tabular}{@{}c@{}}{\bf \SFV} \\ $4$-approx.\\ $(2-\epsilon)$-inapprox.\ ({\bf UGC}) \end{tabular}
 			   & \begin{tabular}{@{}c@{}}{\bf \SFE} \\ $3$-approx.\\ APX-complete \end{tabular}\\
 \hline
\end{tabular}}
\caption{\label{tab:overview} Summary of results for \Editfullp problems (including abbreviations and standard parameter notation). ``Approx.''\ denotes a polynomial-time approximation or bicriteria approximation algorithm; ``inapprox.''\ denotes inapproximability assuming $\mathrm{P} \neq \mathrm{NP}$ unless otherwise specified.}
\end{table}

\paragraph{Hardness results.}
We begin by showing that vertex- and edge-deletion to degeneracy $\Tr$ (\bDE) are both $o(\log (n/\Tr) )$-inapproximable.
Furthermore, we prove that \bBDDfull is $(\ln \Td - C\cdot\ln\ln \Td)$-inapproximable for a constant $C>0$; coupled with the $O(\log \Td)$-approximation in~\cite{ebenlendrapproximation}, this establishes that $(\log \Td)$ is nearly tight.
Additionally, we show that editing a graph to have a specified weak $\Tc$-coloring number (a generalization of degeneracy) is $o(\Tt)$-inapproximable.
The problems of vertex editing to treewidth $\Tw$ and clique number $\Tcn$ are each shown to be $o(\log n)$-inapproximable when the target parameter is $\Omega(n^{\delta})$ for a constant $\delta\geq 1/2$.
Each of these results are proven using a strict reduction from \SCfull.

Note that since computing treewidth (similarly, clique number and weak $\Tc$-coloring number) is $\mathrm{NP}$-hard, trivially it is not possible to find any polynomial time algorithm for the problem of editing to treewidth $\Tw$ with finite approximation guarantee.
In particular, this implies that we cannot hope for anything better than bicriteria approximation algorithms in these problems.
Moreover, our hardness results for these problems show that even if the optimal edit set is guaranteed to be large, i.e., polynomial in the size of the input graph, $n$, we cannot achieve an approximation factor better than $\Omega(\log n)$.

Finally, star forest (treedepth 2) is shown to be $(2-\epsilon)$-inapproximable for vertex editing, via an L-reduction from \VCfull assuming the Unique Games Conjecture, and APX-complete for edge editing, via an L-reduction from \DSBfull.

\paragraph{Positive results.}

Complementing our hardness result for \bDE, we present two bicriteria approximation algorithms for \bDEV, one using the local ratio theorem and another using LP-rounding.
While both approximations can be tuned with error values, they yield constant $(4,4)$- and $(6,6)$-approximations for vertex editing, respectively.
Note that the LP-rounding algorithm also gives a $(5,5)$-approximation for \bDEE.
Using vertex separators, we show a $(O(\log^{1.5} n), O(\sqrt{\log w}))$-bicriteria approximation for vertex editing to bounded treewidth and pathwidth.
Finally, using a hitting-set approach to delete forbidden subgraphs, we show 4- and 3-approximations for the vertex and edge deletion variants of editing to star forests.

\subsection{Related Work}

\paragraph{Editing to approximate optimization problems.}

While there is extensive work on editing graphs into a desired graph class
(summarized below), there is little prior work on how editing affects the
quality of approximation algorithms (when applied to the edited graph, but we
desire a solution to the original graph). The most closely related results of
this type are \emph{parameterized approximation} results, meaning that they
run in polynomial time only when the number of edits is very small (constant
or logarithmic input size).
This research direction was initiated by Cai \cite{cai2003parameterized};
see the survey and results of Marx \cite[Section~3.2]{marx2008survey} and
e.g.\ \cite{guo2004structural,marx2006parameterized}.
An example of one such result is a $7 \over 3$-approximation algorithm to
Chromatic Number in graphs that become planar after $\gamma$ vertex edits,
with a running time of $f(\gamma) \cdot O(n^2)$, where $f(\gamma)$ is at least
$2^{2^{2^{2^{\Omega(\gamma)}}}}$ (from the use of Courcelle's Theorem),
limiting its polynomial-time use to when the number of edits satisfies
$\gamma = O(\lg \lg \lg \lg n)$. In contrast, our algorithms allow up to $\delta \opt_{\Pi}$ edits.

Another body of related work is the ``noisy setting'' introduced by Magen and
Moharrami~\cite{magen2009robust}, which imagines that the ``true'' graph lies
in the structural graph class that we want, and any extra edges observed in
the given graph are  ``noise'' and thus can be ignored when solving the
optimization problem. This approach effectively avoids analyzing the effect of
the edge edits on the approximation factor, by asking for a solution to the
edited graph instead of the given graph.
In this simpler model, Magen and Moharrami~\cite{magen2009robust} developed a PTAS for estimating the size of \ISfull (\IS) in graphs that are $\delta n$ edits away from a minor-closed graph family (for sufficiently small values of $\delta$).
Later, Chan and Har-Peled~\cite{chan2012approximation} developed a PTAS that returns a $(1+\eps)$-approximation to \IS in noisy planar graphs.
Recently, Bansal et al.~\cite{bansal2017lp} developed an LP-based approach for noisy minor-closed \IS whose runtime and approximation factor achieve better dependence on~$\delta$.
Moreover, they provide a similar guarantee for noisy Max $k$-CSPs.
Unlike our work, none of these algorithms bound the approximation ratio for
a solution on the original graph.

\paragraph{Editing algorithms.}

Editing graphs into a desired graph class is an active field of research and
has various applications outside of graph theory, including computer vision and
pattern matching \cite{gao2010survey}.
In general, the editing problem is to delete a minimum set $X$ of vertices (or edges) in an input graph $G$ such that the result $G[V \setminus X]$ has a specific property.
Previous work studied this problem from the perspective of identifying the maximum induced subgraph of $G$ that satisfies a desired ``nontrivial, hereditary'' property~\cite{krishnamoorthy1979node,lewis1978complexity,lewis1980node,yannakakis1978node}.
A graph property $\pi$ is nontrivial if and only if infinitely many graphs satisfy $\pi$ and infinitely many do not, and $\pi$ is hereditary if $G$ satisfying $\pi$ implies that every induced subgraph of $G$ satisfies $\pi$.
The vertex-deletion problem for any nontrivial, hereditary property has been shown to be NP-complete~\cite{lewis1980node} and even
requires exponential time to solve, assuming the ETH~\cite{komusiewicz2015}.
Approximation algorithms for such problems have also been studied somewhat~\cite{fujito1998unified,lund1993approximation,okun2003new} in this domain, but
in general this problem requires additional restrictions on the input
graph and/or output graph properties in order to develop fast algorithms
\cite{dabrowski2015editing,drange2015parameterized,drange2015threshold,huffner2015editing,kotrbvcik2016edge,mathieson2010parameterized,mathieson2008parameterized,xiao2016degreeedit}.

Much past work on editing is on parameterized algorithms.
For example, Dabrowski et al.~\cite{dabrowski2015editing} found that editing a graph to have a given degree sequence is W[1]-complete, but if one additionally requires that the
final graph be planar, the problem becomes Fixed Parameter Tractable (FPT).
Mathieson~\cite{mathieson2010parameterized} showed that editing to degeneracy
$d$ is W[P]-hard (even if the original graph has degeneracy $d+1$ or maximum
degree $2d+1$), but suggests that classes which offer a balance between the
overly rigid restrictions of bounded degree and the overly global condition of
bounded degeneracy (e.g., structurally sparse classes such as $H$-minor-free
and bounded expansion~\cite{nesetril2012sparsity}) may still be FPT.
Some positive results on the parameterized complexity of editing to classes can be
found in Drange's 2015 PhD thesis~\cite{drange2015parameterized};
in particular, the results mentioned include parameterized algorithms for a variety of NP-complete editing problems such as editing to threshold and chain graphs~\cite{drange2015threshold}, star forests~\cite{drange2015threshold}, multipartite cluster graphs~\cite{fomin14tight}, and $\mathcal{H}$-free graphs given finite $\mathcal{H}$ and bounded indegree~\cite{drange2016compressing}.

Our approach differs from this prior work in that we focus on approximations of edit distance that are \emph{polynomial-time approximation algorithms}.
There are previous results about approximate edit distance by Fomin \etal~\cite{fomin2012planar} and, in a very recent result regarding approximate edit distance
to bounded treewidth graphs, by Gupta \etal~\cite{gupta2018losing}.
Fomin et al.~\cite{fomin2012planar} provided two types of algorithms for vertex editing to planar $\mathcal{F}$-minor-free graphs: a randomized algorithm that runs in $O( f(\mathcal{F}) \cdot mn)$ time with an approximation constant $c_{\mathcal{F}}$ that depends on $\mathcal{F}$, as well as a fixed-parameter algorithm parameterized by the size of the edit set whose running time thus has an exponential dependence on the size of this edit set.

Gupta et al.~\cite{gupta2018losing} strengthen the results in~\cite{fomin2012planar} but only in the context of \emph{parameterized approximation algorithms}.
Namely, they give a deterministic fixed-parameter algorithm for \PFDfull that runs in $f(\mathcal{F}) \cdot n\log{n} + n^{O(1)}$ time and an $O(\log{k})$-approximation where $k$ is the maximum number of vertices in any planar graph in $\mathcal{F}$;
this implies a fixed-parameter $O(\log{w})$-approximation algorithm with running time $2^{O(w^2 \log{w})} \cdot n \log{n} + n^{O(1)}$ for \bTWV and \bPWV. They also show that \bTWE and \bPWE have parameterized algorithms that give an absolute constant factor approximation but with running times parameterized by $\Tw$ and the maximum degree of the graph~\cite{gupta2018losing}. Finally, they show that when $\mathcal{F}$ is the set of all connected graphs with three vertices, deleting the minimum number of edges to exclude $\mathcal{F}$ as a subgraph, minor, or immersion is APX-hard for bounded degree graphs~\cite{gupta2018losing}.
Again, these running times are weaker than our results, which give bicriteria
approximation algorithms that are polynomial without any parameterization on
the treewidth or pathwidth of the target graphs.

In a similar regime, Bansal et al.~\cite{bansal2017lp} studied \bTWE (which implies an algorithm for \bPWE) and designed an LP-based bicriteria approximation for this problem.
For a slightly different set of problems in which the goal is to exclude a single graph $H$ of size $k$ as a subgraph (\bHVD), there exists a simple $k$-approximation algorithm.
On the hardness side, Guruswami and Lee~\cite{guruswami2017inapproximability} proved that whenever $H$ is $2$-vertex-connected, it is NP-hard to approximate \bHVD within a factor of $(|V(H)|-1-\eps)$ for any $\eps>0$ ($|V(H)|-\eps$ assuming UGC).
Moreover, when $H$ is a star or simple path with $k$ vertices, $O(\log k)$-approximation algorithms with running time $2^{O(k^3\log k)}\cdot n^{O(1)}$ are known~\cite{guruswami2017inapproximability,lee2017partitioning}.

An important special case of the problem of editing graphs into a desired class is the \textit{minimum planarization} problem, in which the target class is planar graphs, and the related application is approximating the well-known \textit{crossing number} problem~\cite{chuzhoy2011graph}.
Refer to~\cite{biedl2016crossing,chekuri2013approximation,chuzhoy2011algorithm,jansen2014near,KawarabayashiS17,kawarabayashi2009planarity,marx2012obtaining,schaefer2013graph}
for the recent developments on minimum planarization and crossing number.

\resumetocwriting

  \section{Techniques}
\stoptocwriting
This section summarizes the main techniques, ideas, and contributions
in the rest of the paper.

\subsection{Structural Rounding Framework}

The main contribution of our structural rounding framework
(Section~\ref{sec:structural-rounding}) is establishing the
right definitions that make for a broadly applicable framework with precise
approximation guarantees.  Our framework supports arbitrary graph edit
operations and both minimization and maximization problems, provided they
jointly satisfy two properties: a combinatorial property called ``stability''
and an algorithmic property called ``structural lifting''.
Roughly, these properties bound the amount of change that $\opt$ can undergo
from each edit operation, but they are also parameterized to enable us to
derive tighter bounds when the problem has additional structure.
With the right definitions in place, the framework is simple:
edit to the target class, apply an existing approximation algorithm, and lift.

The rest of Section~\ref{sec:structural-rounding} shows that this framework applies to many different
graph optimization problems.  In particular, we verify the stability and
structural lifting properties, and combine all the necessary pieces,
including our editing algorithms from Section~\ref{section:positive_results}
and existing approximation algorithms for structural graph classes.
We summarize all of these results in Table~\ref{table:structuralrounding}
and formally define the framework in Section~\ref{sec:structural-rounding-framework}.

\begin{table}
\begin{minipage}{\textwidth}
\centering
\small
\tabcolsep=0.5\tabcolsep
\begin{tabular}{|l | l | c | c | l | l | l |}
 \hline
 \textbf{Problem} & \textbf{Edit type \boldmath $\psi$} & \boldmath $c'$ & \boldmath $c$ & \textbf{Class \boldmath $\mCp$} & \boldmath $\rho(\lambda)$ & \textbf{runtime}\\
 \hline
 \ISfull (\IS) & vertex deletion & 1 & 0 & degeneracy $\Tr$ & ${1\over r+1}$ & polytime\\
 \ADSfull (\ADS) & $\text{vertex}^*$ deletion & 0 & 1 & degeneracy $\Tr$ & $O(r)$ & polytime~\cite{bansal2017tight}\footnote{The approximation algorithm of~\cite{bansal2017tight} is analyzed only for \DS; 
however, it is straightforward to show that the same algorithm achieves $O(\Tr)$-approximation for \ADS as well.} \\
 \ISfull (\IS) & vertex deletion & 1 & 0 & treewidth $\Tw$ & 1 & $O(2^w n)$~\cite{alber01improved}\\
 \ADSfull (\ADS) & $\text{vertex}^*$ deletion & 0 & 1 & treewidth $\Tw$ & 1 & $O(3^{w}n)$ \\
 \ADSmaybefull (\ADS) & $\text{vertex}^*$ deletion & 0 & 1 & treewidth $\Tw$ & 1 & $O((2\ell+1)^{w}n)$~\cite{borradaile_et_al:LIPIcs:2017:6919} \\
 \CDSfull (\CDS) & $\text{vertex}^*$ deletion & 0 & 3 & treewidth $\Tw$ & 1 & $O(n^{\Tw})$\footnote{Our rounding framework needs to solve an annotated version of \CDS which can be solved in $O(n^{\Tw})$ by modifying the $O(\Tw^{\Tw} n)$ dynamic-programming approach of \DS.} \\
 \VCfull (\VC) & vertex deletion & 0 & 1 & treewidth $\Tw$ & 1 & $O(2^w n)$~\cite{alber01improved}\\
 \FVSfull (\FVS) & vertex deletion & 0 & 1 & treewidth $\Tw$ & 1 & $2^{O(w)}n^{O(1)}~\cite{cygan2011solving}$\\
 \MMMfull (\MMM) & vertex deletion & 0 & 1 & treewidth $\Tw$ & 1 & $O(3^{w} n)$\footnote{The same dynamic-programming approach of \DS can be modified to solve \ADS and \MMM in $O(3^{\Tw} n)$.}\\
 \CNfull (\CN) & vertex deletion & 0 & 1 & treewidth $\Tw$ & 1 & $\Tw^{O(\Tw)} n^{O(1)}$\\
 \ISfull (\IS) & edge deletion & 0 & 1 & degeneracy $\Tr$ & ${1\over r+1}$ & polytime\\
 \DSfull (\DS) & edge deletion & 1 & 0 & degeneracy $\Tr$ & $O(\Tr)$ & polytime~\cite{bansal2017tight} \\
 \DSmaybefull (\DSmaybe) & edge deletion & 1 & 0 & treewidth $w$ & 1 & $O((2\ell+1)^w n)$~\cite{borradaile_et_al:LIPIcs:2017:6919}\\
 \EDSmaybefull (\EDSmaybe) & edge deletion & 1 & 1 & treewidth $w$ & 1 & $O((2\ell+1)^w n)$~\cite{borradaile_et_al:LIPIcs:2017:6919}\\
 \MCfull (\MC) & edge deletion & 1 & 0 & treewidth $w$ & 1 & $O(2^w n)$~\cite{downey2013fundamentals} \\
 \hline
\end{tabular}
\caption{Problems for which structural rounding (Theorem~\ref{thm:general-edit-sr})
results in approximation algorithms for graphs near the structural class $\mC$, where
the problem has a $\rho(\lambda)$-approximation algorithm. We also give
the associated stability ($c'$) and lifting ($c$) constants, which are class-independent.  The last column shows the running time of the $\rho(\lambda)$-approximation algorithm for each problem provided an input graph from class $\mCp$. We remark that $\text{vertex}^*$ is used to emphasize the rounding process has to pick the set of annotated vertices in the edited set carefully to achieve the associated stability and lifting constants.} 
\label{table:structuralrounding}
\end{minipage}
\end{table}

\subsection{Editing to Bounded Degeneracy and Degree}

We present two constant-factor bicriteria approximation algorithms for
finding the fewest vertex or edge deletions to reduce the \emph{degeneracy}
to a target threshold~$\Tr$.
The first approach (Section~\ref{section:positive_degeneracy_localratio})
uses the local ratio technique by Bar-Yehuda et al.~\cite{bar2004local} to
establish that good-enough local choices result in a guaranteed approximation.
The second approach (Section~\ref{section:positive_degeneracy_lp})
is based on rounding a linear-programming
relaxation of an integer linear program.

On the lower bound side,
we show $o(\log (n/\Tr))$-approximation is impossible for vertex or edge edits
(Section~\ref{section:negative_bounded_degeneracy_editing})
when we forbid bicriteria approximation, i.e., when we must match the
target degeneracy $\Tr$ exactly.
This result is based on a reduction
from \SCfull.
A similar reduction proves $o(\log \Td)$-inapproximability of
editing to maximum degree~$d$, which proves tightness (up to constant factors)
of a known $O(\log \Td)$-approximation algorithm \cite{ebenlendrapproximation}.
This algorithm is also LP-based, employing LLL-based analysis to show that the
standard randomized rounding approach works.

\subsection{Editing to Bounded Weak $c$-Coloring Number}

The \emph{weak $\Tc$-coloring number} can be used to characterize
\emph{bounded expansion}, and nicely generalizes the notion of bounded
degeneracy, which corresponds to $\Tc=1$.  (Note, however, that larger $c$ values
make for smaller graph classes; see Section~\ref{section:weak-coloring-number}
for details.)

In Section~\ref{section:negative_weak_coloring},
we show that (non-bicriteria) $o(\Tt)$-approximation is NP-hard
for any $\Tc \geq 2$, by adapting our proof of lower bound for bounded degeneracy.
This hardness result applies to vertex and edge deletions, as well as edge contractions.

\subsection{Editing to Bounded Treewidth}

In Section~\ref{section:positive_treewidth},
we present a bicriteria approximation algorithm for
finding the fewest vertex edits to reduce the \emph{treewidth}
to a target threshold~$\Tw$.
Our approach builds on the deep separator structure inherent in treewidth.
We combine ideas from Bodlaender's $O(\log n)$-approximation algorithm for
treewidth with Feige et al.'s $O(\sqrt{\log w})$-approximation algorithm for vertex separators \cite{feige2008improved} (where $w$ is the target treewidth).
In the end, we obtain a bicriteria
$(O(\log^{1.5} n), O(\sqrt{\log w}))$-approximation
that runs in polynomial time on all graphs
(in contrast to many previous treewidth algorithms).
The tree decompositions that we generate are guaranteed to have
$O(\log n)$ height.
As a result, we also show a bicriteria
$(O(\log^{1.5} n), O(\sqrt{\log w} \cdot \log n))$-approximation
result for pathwidth, based on the fact that the \emph{pathwidth} is at most
the width times the height of a tree decomposition.

On the lower bound side (Section~\ref{section:negative_treewidth_and_clique}),
we prove a $o(\log w)$-inapproximability result
by another reduction from \SCfull.  By a small modification,
this lower bound also applies to editing to \emph{bounded clique number}.

\subsection{Editing to Treedepth $2$}

As a warmup to general treedepth results, we study the simple case of
editing to graphs of treedepth $2$, which are \emph{star forests}.
Here we do not need bicriteria approximation, though the problem remains
hard, even to approximate within a constant factor.
Assuming the Unique Games Conjecture (UGC), we show in
Section~\ref{section:negative_starforest_v} that \VCfull reduces to this
editing problem without a loss in $\opt$, implying a factor-$2$ lower bound
for vertex edits.
We also show APX-hardness for edge edits
(Section~\ref{section:negative_starforest_e}).
Our $4-$approximation (Section~\ref{section:positive_starforest})
is based on a reduction to \Problem{Hitting Set}.
\resumetocwriting

\section{Preliminaries}\label{sec:preliminaries}
\stoptocwriting

This section defines several standard notions and graph classes,
and is probably best used as a reference.  The one exception is
Section \ref{section:editing-relatedwork}, which formally defines
the graph-class editing problem \Editfullp introduced in this paper.

\textbf{Graph notation.}
We consider finite, loopless, simple graphs. Unless otherwise specified, we assume that graphs are undirected and unweighted.
We denote a graph by $G = (V,E)$, and set $n = |V|$, $m = |E|$.
Given $G = (V,E)$ and two vertices $u,v \in V$ we denote edges by $e(u,v)$ or $(u,v)$. We write
$N(v) = \{u \,|\, (u,v) \in E\}$ for the set of neighbors of a vertex $v$;
the degree of $v$ is $\deg(v) = |N(v)|$. In digraphs, in-neighbors and out-neighbors of a vertex $v$
are defined using edges of the form $(u,v)$ and $(v,u)$, respectively,
and we denote in- and out-degree by $\inDeg(v), \outDeg(v)$, respectively.
For the maximum degree of $G$ we use $\Delta(G)$, or just $\Delta$ if context is clear.
The clique number of $G$, denoted $\omega(G)$, is the size of the largest clique in $G$.
Given some subset $E'$ of the edges in $G$, we define $G[E']$ to be the subgraph of $G$ induced on the edge set $E'$.
Note that if every edge adjacent to some vertex $v$ is in $E \setminus E'$, then $v$ does not appear in the vertex set of $G[E']$.

\iffull
\subsection{Structural Graph Classes}\label{section:graphclasses}

In this section, we provide the necessary definitions for several structural graph classes
(illustrated in Figure~\ref{fig:classhierarchy}).

\subsubsection{Degeneracy, cores, and shells}
\begin{definition}\label{def:degeneracy}
    A graph $G$ is \emph{$r$-degenerate} if
    every subgraph contains a vertex of degree at most $r$;
    the \emph{degeneracy} is the smallest $r \in \mathbb{N}$ so that
    $G$ is $r$-degenerate, and write $\degener(G) = r$.
\end{definition}

Much of the literature on degeneracy is in the context of the more refined notion of \emph{$k$-cores}.

\begin{definition}\label{def:k-core}\label{def:k-shell}
  For a graph $G$ and positive integer $r$, the \emph{$r$-core} of $G$, \core{r}{G}, is the maximal subgraph of $G$ with minimum degree $r$.
  The \emph{$r$-shell} of $G$ is $\core{r}{G} \setminus \core{r+1}{G}$.
\end{definition}

\begin{lemma}[\cite{chrobak1991planar,lick1970degenerate,matula1983smallest}]\label{lem:degen-properties}
Given a graph $G=(V,E)$, the following are equivalent:
    \begin{enumerate}
      \itemsep0em
        \item{The degeneracy of $G$ is at most $r$.}
        \item{The $(r+1)$-core of $G$ is empty.}
        \item{There exists an ordering $v_1, \dots, v_n$ of $V$ so the degree of $v_j$ in $G[\{v_j, \dots, v_n\}]$ is at most $r$.}\looseness=-1
    \end{enumerate}
\end{lemma}
\begin{lemma}\label{lem:degen-bounded-outdeg-orientation}
Given a graph $G=(V,E)$, if there exists an orientation of the edges in $E$ so that $\outDeg(u)\leq r$ for all $u\in V$, then the degeneracy of $G$ is at most $2r$.
\end{lemma}
\begin{proof}
It is straightforward to verify that in any induced subgraph $H$ of $G$, the same orientation of edges ensures that out-degree of each vertex $v\in V[H]$, $\outDeg_H(v) \leq r$. This in particular implies that $H$ contains a vertex of degree at most ${|E[H| \over |V[H]|} \leq 2r$. Hence, $\degener(G)\leq 2r$.
\end{proof}

It immediately follows that a graph $G$ has degeneracy $r$ if and only if $r$ is the largest number such that the $r$-core of $G$ is non-empty.
We note that having bounded degeneracy immediately implies bounded clique and chromatic numbers
$\omega(G), \chi(G) \leq \degener(G) + 1$ (the latter follows
from a greedy coloring using the ordering from Lemma~\ref{lem:degen-properties}).

\subsubsection{Weak coloring number}
\label{section:weak-coloring-number}
The \emph{weak $\Tc$-coloring number} was introduced along with the $\Tc$-coloring number by Kierstead and Yang in~\cite{kierstead2003orderings}, and it generalizes the notion of degeneracy in the following sense. As described in Lemma~\ref{lem:degen-properties}, the degeneracy of a graph can be understood as a worst-case bound on the forward degree of a vertex given an optimal ordering of the vertices. The weak $\Tc$-coloring number bounds the number of vertices $v$ reachable from $u$ via a path of length at most $\Tc$ consisting of vertices that occur earlier than $v$ in an ordering of the vertices.  In fact, the weak $1$-coloring number and degeneracy are equivalent notions.

\begin{definition} (Section 2~\cite{kierstead2003orderings}) \label{definition:wcol}
Let $G$ be a finite, simple graph, let $L: V(G) \rightarrow \mathbb{N}$ be an injective function defining an ordering on the vertices of $G$, and let $\Pi(G)$ be the set of all possible such orderings.
A vertex $v$ is \emph{weakly $\Tc$-reachable} from $u$ with respect to $L$ if there exists a $uv$-path $P$ such that $|P| \leq \Tc$ and for all $w \in P$, $L(w) \leq L(v)$;
we use $\wreach{\Tc}{G}{L}{u}$ to denote the set of all such vertices.
Let $\wscore{\Tc}{G}{L}$ be the $\max_{u \in V(G)} |\wreach{\Tc}{G}{L}{u}|$. The
\emph{weak $\Tc$-coloring number} of $G$ is defined as $\wcol{\Tc}{G} = \min_{L \in \Pi(G)} \wscore{\Tc}{G}{L}$.
\end{definition}

The weak $\Tc$-coloring numbers provide useful characterizations for several structural graph classes:

\begin{lemma}[\cite{nesetril2012sparsity,nesetril2011nowhere,zhu2009colouring}] \label{lemma:wcn-characterizations}
The weak $\Tc$-coloring numbers characterize each following class $\mC$:

\begin{enumerate} \itemsep0em
    \item $\mC$ is nowhere dense $\iff$ $\lim\limits_{\Tc \rightarrow \infty} \limsup\limits_{G \in \mC} \frac{\log(\wcol{\Tc}{G})}{\log|G|} = 0$.
    \item $\mC$ has bounded expansion $\iff$ $\exists \textrm{ a function } f, \forall G \in \mC, \forall \Tc \in \mathbb{N}: \wcol{\Tc}{G} \leq f(\Tc)$.
    \item $\mC$ has treedepth bounded by $k$ $\iff$ $\forall G \in \mC, \forall \Tc \in \mathbb{N}: \wcol{\Tc}{G} \leq k$.
\end{enumerate}
\end{lemma}

Additionally, the treewidth of a graph provides an upper bound on its weak coloring numbers.  If a graph $G$ has treewidth $k$, then $\wcol{\Tc}{G} \leq \binom{k + c}{k}$, and there is an infinite family of graphs such that this bound is tight~\cite{grohe2016colouring}.

\subsubsection{Treewidth, pathwidth, and treedepth}
Perhaps the most heavily studied structural graph class is that of \emph{bounded treewidth}; in this subsection
we provide the necessary definitions for treewidth, pathwidth, and treedepth. Bounded treedepth is a stronger structural property than bounded pathwidth, and
intuitively measures how ``shallow'' a tree the graph can be embedded in when edges can only occur between
ancestor-descendent pairs.

\begin{definition}[\cite{robertson1986graph}]\label{definition:treewidth}
    Given a graph $G$, a \emph{tree decomposition} of $G$ consists of a collection $\mathcal{Y}$ of subsets (called \emph{bags}) of vertices in $V(G)$ together with a tree $T = (\mathcal{Y}, \mathcal{E})$ whose nodes $\mathcal{Y}$ correspond to bags which satisfy the following properties:
    \begin{enumerate}
      \itemsep0em
       \item Every $v \in V(G)$ is contained in a bag $B \in \mathcal{Y}$ (i.e.\ $\bigcup_{B \in \mathcal{Y}} B = V$).
       \item For all edges $(u, v) \in E(G)$ there is a bag $B \in \mathcal{Y}$ that contains both endpoints $u,v$.
       \item For each $v \in V(G)$, the set of bags containing $v$ form a connected subtree of $T$ (i.e.\ $\left\{B| v \in B, B \in \mathcal{Y}\right\}$ forms a subtree of $T$).
    \end{enumerate}
    The \emph{width} of a tree decomposition is $\max_{B \in \mathcal{Y}} |B| - 1$, and the \emph{treewidth} of a graph $G$, denoted $\tw(G)$, is the minimum width of any tree decomposition of $G$.
\end{definition}

All graphs that exclude a simple fixed planar minor $H$ have bounded treewidth,
indeed, treewidth $|V(H)|^{O(1)}$ \cite{Chekuri-Chuzhoy-2016}.
Thus, every planar-$H$-minor-free graph class is a subclass of some
bounded treewidth graph class.

\begin{definition}[\cite{robertson1986graph}]\label{definition:pathwidth}
  A \emph{path decomposition} is a tree decomposition in which the tree $T$ is a path. The \emph{pathwidth} of $G$, $\pw(G)$, is the minimum width of any path decomposition of $G$.
\end{definition}

\begin{definition}[\cite{nesetril2006treedepth}] \label{definition:treedepth}
A \emph{treedepth decomposition} of a graph $G$ is an injective mapping $\psi: V(G) \rightarrow V(F)$ to a rooted forest $F$ such that for each edge $(u,v) \in E(G)$, $\psi(u)$ is either an ancestor or a descendant of $\psi(v)$ in $F$.  The \emph{depth} of a treedepth decomposition is the height of the forest $F$.  The \emph{treedepth} of $G$ is the minimum depth of any treedepth decomposition of $G$.
\end{definition}
\fi

\subsection{Editing Problems}\label{section:editing-relatedwork}

This paper is concerned with algorithms that edit graphs into a desired structural class, while guaranteeing an
approximation ratio on the size of the edit set. Besides its own importance, editing graphs into structural classes plays a key role in our structural rounding framework for approximating optimization problems on graphs that are ``close'' to structural graph classes (see Section~\ref{sec:structural-rounding}).
The basic editing problem is defined as follows relative to an edit operation
$\psi$ such as vertex deletion, edge deletion, or edge contraction:

\begin{problem}{\Editfull}
    \Input & An input graph $G=(V,E)$, family $\mC$ of graphs, edit operation $\psi$\\
    \Prob  & Find $k$ edits $\psi_1, \psi_2, \dots, \psi_k$ such that $\kPsieditsG \in \mC$. \\
    \Objective & Minimize $k$
\end{problem}

The literature has limited examples of approximation algorithms for specific edit operations and graph classes.
Most notably, Fomin et al.~\cite{fomin2012planar} studies \Editfull for vertex deletions into the class of \emph{planar-$H$-minor-free graphs}
(graphs excluding a fixed planar graph~$H$).\footnote{More generally, Fomin et al.~\cite{fomin2012planar} consider editing to the class of graphs excluding a finite family $\mF$ of graphs at least one of which is planar, but as we just want the fewest edits to put the graph in some structural class, we focus on the case $|\mF| = 1$.}

In addition to fixed-parameter algorithms (for when $k$ is small),
they give a $c_H$-approximation algorithm for \Editfull where the constant
$c_H = \Omega\left(2^{2^{|V(H)|^3}}\right)$ is rather large.

Most of the graph classes we consider consist of graphs where some
parameter $\lambda$ (clique number, maximum degree, degeneracy, weak $c$-coloring
number, or treewidth) is bounded.  Thus we can think of the graph class $\mC$
as in fact being a parameterized family~$\mCp$.  (For planar-$H$-minor-free,
$\lambda$~could be $|V(H)|$.)  In addition to approximating just the number of edits,
we can also loosen the graph class we are aiming for, and approximate the
parameter value $\lambda$ for the family~$\mCp$.
Thus we obtain a \emph{bicriteria problem} which can be formalized as follows:

\begin{problem}{\Editfullp}
    \Input & An input graph $G=(V,E)$, parameterized family $\mCp$ of graphs, a target parameter value $\lambda^*$, edit operation $\psi$\\
    \Prob  & Find $k$ edits $\psi_1, \psi_2, \dots, \psi_k$ such that $\kPsieditsG \in \mCp$ where $\lambda \geq \lambda^*$. \\
    \Objective & Minimize $k$.
\end{problem}

\begin{definition} \label{definition:bicriteria-approx}
  An algorithm for \Editfullp is a \emph{(bicriteria) $(\alpha,\beta)$-approximation}
  if it guarantees that the number of edits is at most $\alpha$ times the optimal
  number of edits into~$\mCp$, and that $\lambda \leq \beta \cdot \lambda^*$.
\end{definition}

See Table~\ref{tab:overview} for a complete list of the problems considered, along with their abbreviations.

\iffull
\subsection{Hardness and Reductions}\label{sec:prelims-reductions}
One of our contributions in this paper is providing hardness of approximation for
several important instances of \Editfull defined in Section~\ref{section:editing-relatedwork}.
Here, we describe the necessary definitions for approximation-preserving reductions as
well as known approximability hardness results for several key problems.

\subsubsection{Approximation preserving reductions}
A classic tool in proving approximation hardness is the \emph{L-reduction}, which linearly preserves
approximability features~\cite{papadimitriou1991optimization}, and implies PTAS reductions.

\begin{definition}\label{definition:lreduction}
Let $A$ and $B$ be minimization problems with cost functions $\cost_A$ and $\cost_B$, respectively.
An \emph{L-reduction} is a pair of functions $f$ and $g$ such that:
\begin{enumerate}
  \itemsep0em
\item $f$ and $g$ are polynomial time computable,
\item for an instance $\xA$ of $A$, $f(\xA)$ is an instance of $B$,
\item for a feasible solution $\yB$ of $B$, $g(\yB)$ is a feasible solution of $A$,
\item there exists a constant $c_1$ such that
\begin{equation}
\textstyle{\optsol{B}{f(\xA)} \leq c_1 \optsol{A}{\xA}},
\end{equation}
\item and there exists a constant $c_2$ such that
\begin{equation}
\textstyle{\cost_A(g(\yB)) - \opt_A(\xA) \leq c_2\left(\cost_B(\yB) - \opt_B(f(\xA))\right)}.
\end{equation}
\end{enumerate}
\end{definition}

\noindent In many cases, we will establish a stronger form of reduction known as a \emph{strict reduction}, which implies an $L$-reduction~\cite{ko1982computational, orponen1987approximation}.

\begin{definition}\label{def:strict-reduction}
Let $A$ and $B$ be minimization problems with cost functions $\cost_A$ and $\cost_B$, respectively.
A \emph{strict reduction} is a pair of functions $f$ and $g$ such that:
\begin{enumerate}
   \itemsep0em
\item $f$ and $g$ are polynomial time computable,
\item for an instance $x$ of $A$, $f(\xA)$ is an instance of $B$,
\item for a feasible solution $y'$ of $B$, $g(\yB)$ is a feasible solution of $A$,
\item and it holds that
\begin{align}
\frac{\cost_A(g(\yB))}{\opt_A(\xA)} \leq \frac{\cost_B(\yB)}{\opt_B(f(\xA))}.
\end{align}
\end{enumerate}
\end{definition}

We note that to prove a strict reduction, it suffices to demonstrate that $\opt_A(\xA) = \opt_B(f(\xA))$ and $\cost_A(g(\yB)) \leq \cost_B(\yB)$.

\subsubsection{Hard problems}
As mentioned earlier, our approximation hardness results for the instances of \Editfull studied in this paper are via reductions from \SCfull, \VCfull and \DSBfull (\Problem{Minimum Dominating Set} in graphs of maximum degree $B$).
We now formally define each of these, and state the
associated hardness of approximation results used.


\begin{problem}{\SCfull (\SC)}
    \Input & A universe $\mcU$ of elements and a collection $\mcF$ of subsets of the universe. \\
    \Prob  & Find a minimum size subset $\editset \subseteq \mcF$ that covers $\mcU$: $\bigcup_{S\in\editset}S = \mcU$.
\end{problem}

\begin{theorem}[\cite{dinur2014analytical,feige1998threshold,lund1994hardness,moshkovitz2012projection,moshkovitz2010two}]\label{lem:set-cover-general}
It is NP-hard to approximate $\SCfull(\mcU,\mcF)$ within a factor of $(1-\epsilon)\ln( |\mcU| )$ for any $\epsilon > 0$.
Moreover, this holds for instances where $|\mcF| \leq \textrm{poly}(|\mcU|)$.
\end{theorem}
We remark that this result is tight, due to an $(\ln |\mcU|)$-approximation algorithm for \SC~\cite{johnson1974approximation}.

\begin{theorem}[\cite{trevisan2001nonapproximability}]\label{lem-set-cover-sparse}
  There exists a constant $C > 0$ so that it is NP-hard to approximate $\SCfull(\mcU,\mcF)$ within a factor of
  $\left(\ln \Delta - C\ln\ln \Delta \right)$,
  where $\Delta = \max_{S\in \mcF} |S|$.
  Moreover, in the hard instances, $\Delta\geq f_{\max}$ where $f_{\max}$ is the maximum frequency of an element of $\mcU$ in $\mcF$.
\end{theorem}


\begin{problem}{\uSCfull (\uSC)}
    \Input & A universe $\mcU$ of elements and a collection $\mcF$ of subsets of the universe such that every element of $\mcU$ is contained in exactly $k$ sets in $\mcF$.\\
    \Prob  & Find a minimum size subset $\editset \subseteq \mcF$ that covers $\mcU$: $\bigcup_{S\in\editset}S = \mcU$.
\end{problem}

\begin{theorem}[Theorem 1.1~\cite{dinur2005new}]\label{lem-set-cover-frequency}
  For any constant $k \geq 3$,
  it is NP-hard to approximate $\uSC$ within a factor of $(k-1-\epsilon)$ for any $\epsilon >0$.
\end{theorem}

Note that $k$ is assumed to be constant with respect to $|\mcU|$ in this result.
However, the same paper provides a slightly weaker hardness result when $k$ is super-constant with respect to $|\mcU|$.

\begin{theorem}[Theorem 6.2~\cite{dinur2005new}]\label{lem-set-cover-frequency-superconstant}
  There exists a constant $b>0$ so that
  there is no polynomial time algorithm for approximating $\uSC$ within a factor of $(\floor{k/2}-0.01)$ when $4 \leq k \leq (\log |\mcU|)^{1/b}$, unless NP $\subseteq$ DTIME$(n^{O(\log\log n)})$. This holds for instances where $|\mcF| \leq |\mcU|$.

\end{theorem}


\begin{problem}{\VCfull (\VC)}
\Input & A graph $G=(V,E)$.\\
\Prob  & Find a minimum size set of vertices $\editset \subseteq V$ s.t. $G[V\setminus\editset]$ has no edge.
\end{problem}

\begin{theorem}[\cite{dinur2005hardness,khot2008vertex}]\label{lem:vertex-cover}
It is NP-hard to approximate $\VCfull$ within a factor of $1.3606$. Moreover, assuming UGC, \VC has no $(2 - \epsilon)$ approximation for $\epsilon > 0$.
\end{theorem}

\begin{problem}{\DSBfull (\DSB)}
\Input & An undirected graph $G = (V, E)$ with maximum degree at most $B$.\\
\Prob & Find a minimum size set of vertices $C \subseteq V$ such that every vertex in $V$ is either in $C$ or is adjacent to a vertex in $C$.
\end{problem}

\begin{theorem}[\cite{trevisan2001nonapproximability}]\label{thm:ds-hard-approx}
There are constants $C > 0$ and $B_0 \geq 3$ so that for every $B \geq B_0$ it is NP-hard to approximate \DSBfull within a factor of $\ln{B} - C \ln{\ln{B}}$.
\end{theorem}

The best known constants $C$ for small $B$ in Theorem~\ref{thm:ds-hard-approx} are given in~\cite{chlebik2008approximation}.

\subsection{Optimization Problems}\label{section:opt-problems}

We conclude our preliminaries with formal definitions of several additional optimization problems for which
we give new approximation algorithms via structural rounding in Section~\ref{sec:structural-rounding}.

\begin{problem}{\LDSfull (\LDS)}
\Input & An undirected graph $G = (V, E)$ and a positive integer \DSradius.\\
\Prob & Find a minimum size set of vertices $C \subseteq V$ s.t.
every vertex in $V$ is either in $C$ or is connected by a path of length at most \DSradius to a vertex in $C$.
\end{problem}

\begin{problem}{\ELDSfull (\ELDS)}
\Input & An undirected graph $G = (V, E)$ and a positive integer \DSradius.\\
\Prob & Find a minimum size set of edges $C \subseteq E$ s.t.
every edge in $E$ is either in $C$ or is connected by a path of length at most \DSradius to an edge in $C$.
\end{problem}

When $\DSradius = 1$, these are \DSfull (\DS) and \EDSfullfixed (\EDS). 

\begin{problem}{\ADSmaybefull (\ADS)}
\Input & An undirected graph $G = (V, E)$, a subset of vertices $B\subseteq V$ and a positive integer.\\
\Prob & Find a minimum size set of vertices $C \subseteq V$ s.t.
every vertex in $B$ is either in $C$ or is connected by a path of length at most $\ell$ to a vertex in $C$.
\end{problem}

Note that when $B = V$, \ADSmaybefull becomes \DSmaybefull\footnote{The \ADSfull problem has also been studied in the literature as \emph{subset dominating set problem} in~\cite{guha1996approximation,har2017approximation}.}.

\begin{problem}{\LISfull (\LIS)}
\Input & A graph $G=(V,E)$.\\
\Prob  & Find a maximum size set of vertices $X \subseteq V$ s.t. no two vertices in $X$ are connected by a path of length $\leq \ISradius$.
\end{problem}

When $\ISradius = 1$, we call this \ISfull (\IS).

\begin{problem}{\FVSfull (\FVS)}
\Input & A graph $G=(V,E)$.\\
\Prob  & Find a minimum size set of vertices $X \subseteq V$ s.t. $G\setminus X$ has no cycles.
\end{problem}

\begin{problem}{\MMMfull (\MMM)}
\Input & A graph $G=(V,E)$.\\
\Prob  & Find a minimum size set of edges $X \subseteq E$ s.t. $X$ is a maximal matching.
\end{problem}

\begin{problem}{\CNfull (\CN)}
\Input & A graph $G=(V,E)$.\\
\Prob  & Find a minimum size coloring of $G$ s.t. adjacent vertices are different colors.
\end{problem}

\begin{problem}{\MCfull (\MC)}
\Input & A graph $G=(V,E)$.\\
\Prob  & Find a partition of the nodes of $G$ into sets $S$ and $V\setminus S$ such that the number
of edges from $S$ to $V\setminus S$ is greatest.
\end{problem}
\fi

\resumetocwriting

  \section{Structural Rounding}\label{sec:structural-rounding}

In this section, we show how approximation algorithms for a structural graph
class can be extended to graphs that are near that class, provided we can
find a certificate of being near the class.  These results thus motivate our
results in later sections about editing to structural graph classes.
Our general approach, which we call \emph{structural rounding}, is to apply
existing approximation algorithms on the edited (``rounded'') graph in the
class, then ``lift'' that solution to solve the original graph,
while bounding the loss in solution quality throughout.

\subsection{General Framework}\label{sec:structural-rounding-framework}

First we define our notion of ``closeness'' in terms of a general family
$\psi$ of allowable graph edit operations (e.g., vertex deletion, edge
deletion, edge contraction):

\begin{definition}\label{definition:gamma-close}
  A graph $G'$ is \emph{$\gamma$-editable} from a graph $G$
  under edit operation $\psi$
  if there is a sequence of $k \leq \gamma$ edits
  $\psi_1, \psi_2, \dots, \psi_k$ of type $\psi$ such that
  $G' = \kPsieditsG$.
  A graph $G$ is \emph{$\gamma$-close} to a graph class $\mC$
  under $\psi$ if some $G' \in \mC$ is
  $\gamma$-editable from $G$ under~$\psi$.

\end{definition}

To transform an approximation algorithm for a graph class $\mC$
into an approximation algorithm for graphs $\gamma$-close to $\mC$,
we will need two properties relating the optimization problem
and the type of edits:%
\footnote{These conditions are related to, but significantly generalize,
  the ``separation property'' from the bidimensionality framework for
  PTASs \cite{demaine2005bidimensionality}.}

\begin{definition} \label{definition:stable}
A graph minimization (resp.\ maximization) problem $\Pi$ is \emph{stable}
under an edit operation $\psi$ with constant $c'$
if $\optsol{\Pi}{G'} \leq \optsol{\Pi}{G} + c' \gamma$ (resp.\ $\optsol{\Pi}{G'} \geq \optsol{\Pi}{G} - c' \gamma$)
for any graph $G'$ that is $\gamma$-editable from $G$ under~$\psi$.
In the special case where $c'=0$, we call $\Pi$ \emph{closed} under~$\psi$.
When $\psi$ is vertex deletion, closure is equivalent to
the graph class defined by $\optsol{\Pi}{G} \leq \lambda$ (resp.\ $\optsol{\Pi}{G} \geq \lambda$)
being \emph{hereditary}; we also call $\Pi$ \emph{hereditary}.
\end{definition}

\begin{definition}\label{definition:struct-lift-c-psi}
A minimization (resp.\ maximization) problem $\Pi$ can be \emph{structurally
lifted} with respect to an edit operation $\psi$ with constant $c$ if,
given any graph $G'$ that is $\gamma$-editable from $G$ under~$\psi$,
and given the corresponding edit sequence $\psi_1, \psi_2, \dots, \psi_k$
with $k \leq \gamma$,
a solution $S'$ for $G'$ can be converted in polynomial time to
a solution $S$ for $G$ such that
$\cost_{\Pi}(S) \leq \cost_{\Pi}(S') + c\cdot k$
(resp.\ $\cost_\Pi(S) \geq \cost_\Pi(S') - c \cdot k$).
\end{definition}

Now we can state the main result of structural rounding:

\begin{theorem}[Structural Rounding Approximation]\label{thm:general-edit-sr}
Let $\Pi$ be a minimization (resp.\ maximization) problem that is stable
under the edit operation $\psi$ with constant $c'$ and
that can be structurally lifted with respect to $\psi$ with constant~$c$.
If $\Pi$ has a polynomial-time $\rho(\lambda)$-approximation algorithm in the graph class $\mCp$,
and \Editfullp has a polynomial-time $(\alpha,\beta)$-approximation algorithm,
then there is a polynomial-time
$((1 + c' \alpha \delta) \cdot \rho(\beta \lambda) + c \alpha \delta)$-approximation
(resp.\ $((1 - c' \alpha \delta) \cdot \rho(\beta \lambda) - c \alpha \delta)$-approximation)
algorithm for $\Pi$
on any graph that is $(\delta \cdot \optsol{\Pi}{G})$-close
to the class~$\mCp$.
\end{theorem}

\iffull
\begin{proof}
We write $\opt(G)$ for $\optsol{\Pi}{G}$.
Let $G$ be a graph that is $(\delta \cdot \opt(G))$-close to the class~$\mCp$.
By Definition~\ref{definition:bicriteria-approx}, the polynomial-time
$(\alpha,\beta)$-approximation algorithm finds
edit operations $\psi_1, \psi_2, \dots, \psi_k$
where $k \leq \alpha \delta \cdot \opt(G)$
such that $G' = \kPsieditsG \in \mC_{\beta \lambda}$.%
\footnote{We assume that $C_i \subseteq C_j$ for $i \leq j$,
  or equivalently, that $\rho(\lambda)$ is monotonically increasing in~$\lambda$.}
Let $\rho = \rho(\beta \lambda)$ be the approximation factor we can attain
on the graph $G' \in \mC_{\beta \lambda}$.

First we prove the case when $\Pi$ is a minimization problem.
Because $\Pi$ has a $\rho$-approximation in $\mC_{\beta \lambda}$
(where $\rho > 1$), we can obtain a solution $S'$ with cost at most
$\rho \cdot \opt(G')$ in polynomial time.
Applying structural lifting (Definition~\ref{definition:struct-lift-c-psi}), we can use $S'$ to obtain a solution
$S$ for $G$ with $\cost(S) \leq \cost(S') + c k \leq \cost(S') + c \alpha \delta \cdot \opt(G)$ in polynomial time.
Because $\Pi$ is stable under $\psi$ with constant $c'$,
\begin{align*}
\opt(G') &\leq \opt(G) + c' k
\leq \opt(G) + c' \alpha \delta \cdot \opt(G)
= (1 + c' \alpha \delta) \opt(G),
\end{align*}
and we have
\begin{align*}
\cost(S) &\leq \rho \cdot \opt(G') + c \alpha \delta \cdot \opt(G)\\
&\leq \rho (1 + c' \alpha \delta) \opt(G) + c \alpha \delta \cdot \opt(G)\\
&= (\rho + \rho c' \alpha \delta + c \alpha \delta) \opt(G),
\end{align*}
proving that we have a polynomial time $(\rho + (c + c' \rho) \alpha \delta)$-approximation algorithm as required.

Next we prove the case when $\Pi$ is a maximization problem.
Because $\Pi$ has a $\rho$-approximation in $\mC$ (where $\rho < 1$),
we can obtain a solution $S'$ with cost at least
$\rho \cdot \opt(G')$ in polynomial time.
Applying structural lifting (Definition~\ref{definition:struct-lift-c-psi}), we can use $S'$ to obtain a solution
$S$ for $G$ with $\cost(S) \geq \cost(S') - c k \geq \cost(S') - c \alpha \delta \cdot \opt(G)$ in polynomial time.
Because $\Pi$ is stable under $\psi$ with constant~$c'$,
\begin{align*}
\opt(G') &\geq \opt(G) - c' k
\geq \opt(G) - c' \alpha \delta \cdot \opt(G)
= (1 - c' \alpha \delta) \opt(G),
\end{align*}
and we have
\begin{align*}
\cost(S) &\geq \rho \cdot \opt(G') - c \alpha \delta \cdot \opt(G) \\
&\geq \rho (1 - c' \alpha \delta) \opt(G) - c \alpha \delta \cdot \opt(G) \\
&= (\rho - (c + c' \rho) \alpha \delta) \opt(G),
\end{align*}
proving that we have a polynomial-time $(\rho - (c + c'\rho) \alpha \delta)$-approximation algorithm as required.
Note that this approximation is meaningful only when $\rho > (c + c'\rho) \alpha \delta$.
\end{proof}
\fi

To apply Theorem~\ref{thm:general-edit-sr}, we need four ingredients:
(a)~a proof that the problem of interest is stable under some edit operation (Definition~\ref{definition:stable});
(b)~a polynomial-time $(\alpha,\beta)$-approximation algorithm for editing
under this operation (Definition~\ref{definition:bicriteria-approx});
(c)~a structural lifting algorithm
(Definition~\ref{definition:struct-lift-c-psi}); and
(d)~an approximation algorithm for the target class~$\mC$.

In the remainder of this section, we show how this framework applies to many
problems and graph classes, as summarized in Table~\ref{table:structuralrounding}
on page~\pageref{table:structuralrounding}.
Most of our approximation algorithms depend on
our editing algorithms described in Section~\ref{section:positive_results}.
We present the problems ordered by edit type, as listed in
Table~\ref{table:structuralrounding}.

\paragraph{Structural rounding for annotated problems.}
We refer to graph optimization problems where the input consists of both a graph and subset of annotated vertices/edges as \emph{annotated} problems (see \ADSfull in Section~\ref{section:opt-problems}). Hence, in our rounding framework, we have to carefully choose the set of annotated vertices/edges in the edited graph to guarantee small {\em lifting} and {\em stability} constants. To emphasize the difference compared to ``standard'' structural rounding, we denote the edit operations as $\text{vertex}^*$ and $\text{edge}^*$ in the annotated cases.
Moreover, we show that we can further leverage the flexibility of annotated rounding to solve {\em non-annotated} problems that cannot normally be solved via structural rounding. In Section~\ref{sec:annotated-vertex-deletion}, we consider applications of annotated rounding for both annotated problems such as \ADSfull and non-annotated problems such as \CDSfull.
\iffull
\subsection{Vertex Deletions}
\fi

For each problem, we show stability and structural liftability,
and use these to conclude approximation algorithms.
\iffull
Because \IS is the only maximization problem we first consider in this section, we consider it separately.
\fi

\iffull
\begin{lemma}\label{lemma:IS-stable}
  \ISfull is stable under vertex deletion with constant $c'=1$.
\end{lemma}

\begin{proof}
  Given a graph $G$ and any set $X \subseteq V(G)$ with $|X| \leq \gamma$, let $G' = G[V\setminus X]$.
  For any independent set $Y \subset V(G)$, $Y' = Y\setminus X$ is also an independent set in $G'$ with size $|Y'| \geq |Y|-|X|$, which is bounded below by $|Y| - \gamma$.
  In particular, for $Y$ optimal in $G$ we have
  $|Y'| \geq \opt(G) - \gamma$, and so $\opt(G') \geq \opt(G) - \gamma$.
\end{proof}

\begin{lemma}\label{lemma:IS-stable-lift}
  \ISfull can be structurally lifted with respect to vertex deletion with constant $c=0$.
\end{lemma}

\begin{proof}
  An independent set in $G' = G \setminus X$ is also an independent set in~$G$.
  Thus, a solution $S'$ for $G'$ yields a solution $S$ for $G$ such that $\cost_\IS(S') = \cost_\IS(S)$.
\end{proof}

\begin{corollary}\label{cor:is-vertex-edit-approximation}
    For graphs $(\delta \cdot \opt(G))$-close to a graph class $\mC_\lambda$ via vertex deletions, \ISfull has the following approximations.  For degeneracy $\Tr$, \IS has a $(1 - 4 \delta)/(4\Tr+1)$-approximation, for treewidth~$\Tw$ such that $w\sqrt{\log w} = O(\log n)$, \IS has a $(1 - O(\delta \log^{1.5} n))$-approximation, and for planar-$H$-minor-free, \IS has a $(1 - c_H \delta)$-approximation.
\end{corollary}

\begin{proof}
  We apply Theorem~\ref{thm:general-edit-sr} using
  stability with $c'=1$ (Lemma~\ref{lemma:IS-stable})
  and structural lifting with $c=0$ (Lemma~\ref{lemma:IS-stable-lift}).
  The independent-set approximation algorithm and the
  editing approximation algorithm depend on the class~$\mC_\lambda$.

  For degeneracy $\Tr$, we use our $(4,4)$-approximate editing algorithm
  (Section~\ref{section:positive_degeneracy_LP_vertex})
  and a simple $1/(\Tr+1)$-approximation algorithm for independent set:
  the $\Tr$-degeneracy ordering on the vertices of a graph gives a canonical
  $(\Tr+1)$-coloring, and the pigeonhole principle guarantees an independent
  set of size at least $|V|/(\Tr+1)$, which is at least $1/(\Tr+1)$ times the
  maximum independent set.
  Thus $\alpha=\beta=4$ and $\rho(\beta \Tr) = 1/(\beta \Tr + 1)$,
  resulting in an approximation factor of $(1-4 \delta)/(4 \Tr+1)$.

  For treewidth $\Tw$ such that $w\sqrt{\log w} = O(\log n)$, we use our
  $(O(\log^{1.5} n), O(\sqrt{\log \Tw}))$-approximate editing algorithm
  (Section~\ref{section:positive_treewidth}) and an exact algorithm for
  independent set \cite{bodlaender1988dynamic,alber01improved} given a tree decomposition of width $O(\log n)$ of the edited graph.
  Thus $\alpha=O(\log^{1.5} n)$ and $\rho = 1$,
  resulting in an approximation factor of $1 - O(\log^{1.5} n) \delta$.

  For planar-$H$-minor-free, we use Fomin's $c_H$-approximate editing algorithm
  \cite{fomin2012planar} and the same exact algorithm for \IS in bounded treewidth
  (as any planar-$H$-minor-free graph has bounded treewidth \cite{Chekuri-Chuzhoy-2016}).
  Thus $\alpha=c_H$ and $\rho = 1$,
  resulting in an approximation factor of $1 - c_H \delta$.
\end{proof}

\begin{lemma}\label{lemma:VD-hereditary}
The problems \VCfull, \FVSfull, \MMMfull, and \CNfull are hereditary (closed under vertex deletion).
\end{lemma}

\begin{proof}
Let $G$ be a graph, and $G' = G\setminus X$ where $X \subseteq V(G)$.
Any vertex cover in $G$ remains
a cover in $G'$ because $E(G') \subseteq E(G)$, so \VC is hereditary.

Let $S$ bs a feedback vertex set in $G$ and $S' = S \setminus X$.
For \FVS, we observe that removing vertices can only decrease the number of cycles in the graph.
Deleting a vertex in $S$ breaks
all cycles it is a part of and, thus, the cycles no longer need to be covered by a vertex in the feedback vertex set of $G'$.
Deleting a vertex not in $S$ can only decrease the
number of cycles, and, thus, all cycles in $G'$ are still covered by $S'$.
Hence, \FVS is hereditary.

For \MMM, deleting vertices with adjacent edges not in the matching only decreases the number of edges; thus, the original matching
is a still a matching in the edited graph. Deleting vertices adjacent to an edge in the matching means that at most one edge in the matching
per deleted vertex is deleted. For each edge in the matching with one of its two endpoints deleted, at most one additional edge (an edge adjacent
to its other endpoint) needs to be added to maintain the maximal matching. Thus, the size of the maximal matching does not increase
and \MMM is hereditary.

\CN is trivially hereditary because deleting vertices can only decrease the number of colors necessary  to color the graph.
\end{proof}

\begin{lemma}\label{lemma:VD-sr}
The problems \VCfull, \FVSfull, \MMMfull, and \CNfull can be structurally lifted with respect to vertex
deletion with constant $c = 1$.
\end{lemma}

\begin{proof}
    Let $G$ be a graph, and $G' = G\setminus X$ where $X \subseteq V(G)$.
    Let $S'$ be a solution to optimization problem $\Pi$ on $G'$. We will
    show that $S \subseteq S' \cup X$ is a valid solution to $\Pi$ on $G$ for each $\Pi$ listed in the Lemma.

    Given a solution $S'$ to \VC for the graph $G'$, the only edges not covered by $S'$ in $G'$ are edges
    between $X$ and $G'$ and between two vertices in $X$. Both sets of such edges are covered by $X$. Thus,
    $S = S' \cup X$ is a valid cover for $G$.

    Given a solution $S'$ to \FVS for the graph $G'$, the only cycles not covered by $S'$ in $G'$ are
    cycles that include a vertex in $X$. Thus, $S = S' \cup X$ is a valid feedback vertex set for $G$ since $X$ covers
    all newly introduced cycles in $G$.

    Given a solution $S'$ to \MMM for the graph $G'$, the only edges not in the matching and not adjacent to
    edges in the matching are edges between $X$ and $G'$ and edges between two vertices in $X$. Thus, any additional
    edges added to the maximal matching will come from $X$, and $S \subseteq S' \cup X$ (by picking edges to add to the maximal matching
    greedily for example) is a valid solution.

    Given a solution $S'$ to \CN for the graph $G'$, the only vertices that could violate the coloring of the
    graph $G'$ are vertices in $X$. Making each vertex in $X$ a different color from each other as well as the colors in
    $G'$ creates a valid coloring of $G$.
    Thus, $S = S' \cup X$ is a valid coloring.
\end{proof}

\begin{corollary}
  The problems \VCfull, and \FVSfull have
  $(1 + O(\delta \log^{1.5} n))$-approximations for graphs
  $(\delta \cdot \opt(G))$-close to treewidth~$w$ via vertex deletions where $w\sqrt{log w} = O(\log n)$;
  and $(1 + c_H \delta)$-approximations for graphs
  $(\delta \cdot \opt(G))$-close to planar-$H$-minor-free via vertex deletions.
\end{corollary}

\begin{proof}

  We apply Theorem~\ref{thm:general-edit-sr} using
  stability with constant $c'=0$ (Lemma~\ref{lemma:VD-hereditary})
  and structural lifting with constant $c=1$ (Lemma~\ref{lemma:VD-sr}).

  For treewidth $\Tw$, we use our
  $(O(\log^{1.5} n), O(\sqrt{\log \Tw}))$-approximate editing algorithm
  (Section~\ref{section:positive_treewidth}) and an exact polynomial-time algorithm for
  the problem of interest \cite{bodlaender1988dynamic,alber01improved,cygan2011solving} given the tree-decomposition of width $O(w\sqrt{\log w})$ of the edited graph.
  Thus $\alpha=O(\log^{1.5} n)$ and $c = 1$,
  resulting in an approximation factor of $1 + O(\log^{1.5} n) \delta$.
  Note that since the edited graph has treewidth $O(w\sqrt{\log w}) = O(\log n)$, the exact algorithm runs in polynomial-time.
  %
  For planar-$H$-minor-free graphs, we use Fomin's $c_H$-approximate editing algorithm
  \cite{fomin2012planar} and the same exact algorithm for bounded treewidth
  (as any planar-$H$-minor-free graph has bounded treewidth \cite{Chekuri-Chuzhoy-2016}).
  Thus $\alpha=c_H$ and $c = 1$,
  resulting in an approximation factor of $1 + c_H \delta$.
\end{proof}

\begin{corollary}\label{cor:vertex-deletion-approx}
  The problems \MMMfull, and \CNfull have
  $(1 + O(\delta \log^{1.5} n))$-approximations for graphs
  $(\delta \cdot \opt(G))$-close to treewidth~$w$ via vertex deletions where $w\log^{1.5} w = O(\log n)$;
  and $(1 + c_H \delta)$-approximations for graphs
  $(\delta \cdot \opt(G))$-close to planar-$H$-minor-free via vertex deletions.
\end{corollary}

\begin{proof}

  We apply Theorem~\ref{thm:general-edit-sr} using
  stability with constant $c'=0$ (Lemma~\ref{lemma:VD-hereditary})
  and structural lifting with constant $c=1$ (Lemma~\ref{lemma:VD-sr}).

  For treewidth $\Tw$, we use our
  $(O(\log^{1.5} n), O(\sqrt{\log \Tw}))$-approximate editing algorithm
  (Section~\ref{section:positive_treewidth}) and an exact algorithm for
  the problem of interest \cite{bodlaender1988dynamic} given a tree-decomposition of width $O(w\sqrt{\log w})$ of the edited graph.
  Thus $\alpha=O(\log^{1.5} n)$ and $c = 1$,
  resulting in an approximation factor of $1 + O(\log^{1.5} n) \delta$.
  Note that since the edited graph has treewidth $O(w\log^{1.5} w) = O(\log n)$, the exact algorithm runs in polynomial-time.
  %
  For planar-$H$-minor-free graphs, we use Fomin's $c_H$-approximate editing algorithm
  \cite{fomin2012planar} and the same exact algorithm for bounded treewidth
  (as any planar-$H$-minor-free graph has bounded treewidth \cite{Chekuri-Chuzhoy-2016}).
  Thus $\alpha=c_H$ and $c = 1$,
  resulting in an approximation factor of $1 + c_H \delta$.
\end{proof}
\fi

\iffull
\subsection{Edge Deletions}
\fi

\begin{theorem}\label{thm:edge-edits-approx}
For graphs $(\delta \cdot \opt(G))$-close to degeneracy~$\Tr$ via edge deletions:
\begin{itemize}
\item \ISfull has a $(1/(3\Tr+1) - 3\delta)$-approximation.
\item \DSfull has an $O((1+ \delta) \Tr)$-approximation.
\end{itemize}

\noindent For graphs $(\delta \cdot \opt(G))$-close to treewidth~$w$ via edge deletions:

\begin{itemize}
  \item \DSmaybefull and \EDSmaybefull have
  $(1 + O(\delta \log n\log\log n))$-approximations when $w\log w = O(\log_{\ell} n)$.
  \item \MCfull has a
  $(1 - O(\delta \log n\log\log n))$-approximation when $w\log w = O(\log n)$.
 \end{itemize}
\end{theorem}

\iffull
We now consider the edit operation of edge deletion.
For each problem, we show stability and structural liftability,
and use these to conclude approximation algorithms.

\begin{lemma}\label{lemma:IS-edge-stable}
  For $\ell \geq 1$,
  ($\ell$-)\ISfull is stable under edge deletion with constant $c'= 0$.
\end{lemma}
\begin{proof}
  Given $G$ and any set $X \subseteq E(G)$ with $|X| \leq \gamma$, let $G' = G[E\setminus X]$.
  For any ($\ell$-)independent set $Y \subseteq V(G)$, $Y' = Y$ is also an ($\ell$-)independent set in $G'$.
  Then $\opt(G') \geq |Y'| = |Y|$, and so for optimal $Y$, $\opt(G') \geq \opt(G)$.
\end{proof}

\begin{lemma}\label{lemma:IS-edge-lift}
  For $\ell \geq 1$, ($\ell$-)\ISfull can be structurally lifted with respect to edge deletion with constant $c=1$.
\end{lemma}
\begin{proof}
  Given a graph $G$ and $X \subseteq E(G)$, let $G' = G[E\setminus X]$.
  Let $Y' \subseteq V(G')$ be an ($\ell$-)independent set in $G'$,
  and consider the same vertex set $Y'$ in $G$.
  Assume that the edit set is a single edge, $X = \{(u,v)\}$.
  We claim there exists a subset of $Y'$ with size at least $|Y'| - 1$
  which is still an ($\ell$-)independent set in $G$.

  For convenience, we let $d(\cdot, \cdot) := d_G(\cdot, \cdot)$ for the remainder of this proof.  Suppose there are four distinct nodes $a,b,f,g \in Y'$ such that
  $d(a,b)\leq \ell$ and $d(f,g) \leq \ell$ in $G$.
  Since these nodes are in $Y'$, we know $d_{G'}(a,b), d_{G'}(f,g) \geq \ell+1$,
  hence, any shortest path from $a$ to $b$ in $G$ must use the edge $(u,v)$ in order to have length $\leq \ell$.
  \WLOG we can assume the $a$-$b$ path goes from $a$ to $u$ to $v$ to $b$, and so
  $d(a,u) + 1 + d(v,b) \leq \ell$.
  Similarly we can assume the $f$-$g$ path goes from $f$ to $v$ to $u$ to $g$, and so
  $d(f,v) + 1 + d(u,g) \leq \ell$.
  We now argue that the shortest paths in $G$ from $a$ to $u$, $v$ to $b$, $f$ to $v$, and $u$ to $g$ do not use the edge $(u,v)$ and are therefore also paths in $G'$. Suppose
  not and consider \WLOG the case when a shortest path from $a$ to $u$ contains $(u,v)$.
  Then concatenating the subpath from $a$ to $v$ with
  a shortest path from $v$ to $b$ gives an $a,b$-path of length $d(a,u) - 1 + d(v,b) < \ell$, which does not use the edge $(u,v)$ (and is thus a path in $G'$, contradicting $(\ell-)$independence of $Y'$).

  Now consider the paths ($a$ to $u$ to $g$) and ($f$ to $v$ to $b$).
  Let $\ell_A = d(a,u) + d(u,g)$ and $\ell_F = d(f,v) + d(v,b)$,
  and note that
  $\ell_A+\ell_F = d(a,u) + d(v,b) + d(f,v) + d(u,g),$
  which is $\leq 2\ell - 2$.
  So at least one of $\ell_A$ or $\ell_F$ must be $\leq \ell-1$, a contradiction.

  Three cases remain: (1) $Y'$ contains exactly two vertices connected by a path of length $\leq \ell$ in $G$; (2) $Y'$ contains three distinct vertices pair-wise connected by paths of length $\leq \ell$ in $G$; or (3) $Y'$ contains one vertex, $a$, connected to two or more other vertices of $Y'$ by paths of length $\leq \ell$ in $G$.
  In the first case, $Y'$ contains $a,b$ with $d(a,b) \leq \ell$; then removing either endpoint from $Y'$ yields an $(\ell-)$independent set of size $|Y'|-1$ in $G$.

  We now show the second case cannot occur. Suppose that $d(b,c), d(a,b), d(a,c) \leq \ell$ for $a,b,c \in Y'$. Note that each vertex is within distance $\ell/2$ of at least one of the vertices $u$ or $v$. By the pigeonhole principle, some two of $a,b,c$ must be within $\ell/2$ of the same endpoint of $(u,v)$; say vertices $a$ and $b$ are within $\ell/2$ of $u$ \WLOG; this implies $d_{G'}(a,b) \leq \ell$, a contradiction.

  Finally, in the third case, $Y'$ contains a node $a$ and a subset $S$ so that $|S| \geq 2$,
  $d(a,s) \leq \ell$ for all $s \in S$ and $d(s_1, s_2) > \ell$ for all $s_1 \neq s_2$ in $S$. Further, we know no other pair of nodes in $Y'$ is at distance at most $\ell$ in $G$ (since then we would have two disjoint pairs at distance at most $\ell$, a case we already handled). In this setting, $Y'\setminus \{a\}$ is an $(\ell-)$independent set of size $|Y'|-1$ in $G$.
  This proves that adding a single edge to $G'$ will reduce the size of the ($\ell$-)independent set $Y'$ by no more than one, so by induction the lemma holds.
\end{proof}

\begin{corollary} \label{cor:is-structural-rounding}
  \ISfull has a $(1/(3\Tr+1) - 3\delta)$-approximation for graphs
  $(\delta \cdot \opt(G))$-close to degeneracy~$\Tr$ via edge deletions.
\end{corollary}

\begin{proof}
  We apply Theorem~\ref{thm:general-edit-sr} using
  stability with constant $c'=0$ (Lemma~\ref{lemma:IS-stable})
  and structural lifting with constant $c=1$ (Lemma~\ref{lemma:IS-stable-lift}).
  We use our $(3,3)$-approximate editing algorithm
  (Corollary~\ref{cor:edge-edit})
  and the $1/(\Tr+1)$-approximation algorithm for independent set described in in the proof of Corollary~\ref{cor:is-vertex-edit-approximation}.
  Thus $\alpha=\beta=3$ and $\rho(\beta \Tr) = 1/(\beta \Tr + 1)$,
  resulting in an approximation factor of $1/(3 \Tr+1) - 3 \delta$.
\end{proof}

Note that Corollary~\ref{cor:is-structural-rounding} only applies to \IS and not \LIS.

\begin{lemma}\label{lemma:ED-stable}
The problems \DSmaybefull and \EDSmaybefull are stable under edge deletion with constant $c' = 1$.
\end{lemma}

\begin{proof}
Given $G$ and any set $X \subseteq E(G)$ with $|X| \leq \gamma$, let $G' = G[E\setminus X]$, and let $Y$ be a minimum (\DSradius-)dominating set on $G$. Each vertex $v$ may be (\DSradius-)dominated by multiple vertices on multiple paths, which we refer to as $v$'s \textit{dominating paths}.

Consider all vertices for which a specific edge $(u,v)$ is on all of their dominating paths in $G$. We refer to each of these vertices as $(u,v)$-dependent. Note that if we traverse all dominating paths from each $(u,v)$-dependent vertex, $(u,v)$ is traversed in the same direction each time. Assume \WLOG $(u,v)$ is traversed with $u$ before $v$, implying $u$ is not $(u,v)$-dependent but $v$ may be. Now if $(u,v)$ is deleted, then $Y \cup {\{v\}}$ is a (\DSradius-)dominating set on the new graph. Therefore for each edge $(u,v)$ in $X$ we must add at most one vertex to the (\DSradius-)dominating set.  Thus if $Y'$ is a minimum (\DSradius-)dominating set on $G'$ then $|Y'| \leq |Y| + \gamma$ and \DSmaybe is stable under edge deletion with constant $c' = 1$.

Now let $Z$ be a minimum edge (\DSradius-)dominating set on $G$.
The proof for \EDSmaybe follows similarly as in the above case when a deleted edge $(u,v)$ is not in $Z$ (though an edge incident to $v$ would be picked to become part of the dominating set instead of $v$ itself). However if $(u,v)$ is in the minimum edge (\DSradius-)dominating set then it is possible that there are edges which are strictly $(u,v)$-dependent through only $u$ or $v$ and no single edge is within distance \DSradius~ of both. In this case we add an edge adjacent to $u$ and an edge adjacent to $v$ to $Z$, which also increases $Z$'s size by one with the deletion of $(u,v)$.
Thus if $Z'$ is a minimum edge (\DSradius-)dominating set on $G'$ then $|Z'| \leq |Z| + \gamma$ and
 \EDSmaybefull is stable under edge deletion with constant $c'=1$.
\end{proof}

\begin{lemma}\label{lemma:ED-sl}
 \DSmaybefull and \EDSmaybefull can be structurally lifted with respect to edge deletion with constants $c = 0$ and $c = 1$ respectively.
\end{lemma}

\begin{proof}

Given $G$ and any set $X \subseteq E(G)$ with $|X| \leq \gamma$, let $G' = G[E\setminus X]$.
A (\DSradius-)dominating set in $G'$ is also a (\DSradius-)dominating set in $G$.
Therefore, a solution $S'$ in $G'$ yields a solution $S$ in $G$ such that $\cost_{\DSmaybe}(S') = \cost_{\DSmaybe}(S)$.

An edge (\DSradius-)dominating set $Y'$ in $G'$ may not be an edge (\DSradius-)dominating set in $G$, as there may be edges in $X$ which are not (\DSradius-)dominated by $Y'$. However  $Y' \cup X$ is an edge (\DSradius-)dominating set in $G$ and $|Y' \cup X| \leq |Y'| + |X|$.
\end{proof}

\begin{corollary}\label{cor:edge-deletion-approx-degen}
  \DSfull has an $O((1+ \delta) \Tr)$-approximation for graphs
  $(\delta \cdot \opt(G))$-close to degeneracy~$r$ via edge deletions.
\end{corollary}

\begin{proof}
  We apply Theorem~\ref{thm:general-edit-sr} using
  stability with constant $c'=1$ (Lemma~\ref{lemma:ED-stable})
  and structural lifting with constant $c=0$ (Lemma~\ref{lemma:ED-sl}).
  We use our $(3,3)$-approximate editing algorithm
  (Section~\ref{section:positive_degeneracy_LP_edge})
  and a known $O(\Tr^2)$-approximation algorithm for \DS
  \cite{lenzen2010minimum}.
  Thus $\alpha=\beta=3$ and $\rho(\beta \Tr) = \beta^2 \Tr^2$,
  resulting in an approximation factor of $9 (1+3 \delta) \Tr^2$.
\end{proof}

\begin{corollary}\label{cor:edge-deletion-approx-treewidth}
  \DSmaybefull and \EDSmaybefull have
  $(1 + O(\delta \log n\log\log n))$-approximations for graphs
  $(\delta \cdot \opt(G))$-close to treewidth~$w$ via edge deletions where $w\log w = O(\log_{\ell} n)$.
\end{corollary}

\begin{proof}
  We apply Theorem~\ref{thm:general-edit-sr} using
  stability with constant $c'=1$ (Lemma~\ref{lemma:ED-stable})
  and structural lifting with constant $c=0$ for \DSmaybe and constant $c=1$ for \EDSmaybe (Lemma~\ref{lemma:ED-sl}).
  For treewidth $\Tw$, we use the
  $(O(\log n\log\log n), O(\log \Tw))$-approximate editing algorithm
  of Bansal \etal~\cite{bansal2017lp}
  and an exact algorithm for \DSmaybe and \EDSmaybe~\cite{borradaile_et_al:LIPIcs:2017:6919} given a tree-decomposition of width $O(w\log w)$ of the edited graph.

  Thus $\alpha=O(\log n\log\log n)$ and $c' = 1$ for \DS and $c'=c=1$ for \EDS,
  resulting in an approximation factor of $1 + O(\log n\log\log n) \delta$.
  Note that since the edited graph has treewidth $O(w\log w) = O(\log_{\ell} n)$, the exact algorithm runs in polynomial-time.
\end{proof}

\begin{lemma}\label{lemma:MC-stable}
The problem \MCfull is stable under edge deletion with constant $c' = 1$.
\end{lemma}
\begin{proof}
Given $G$ and any set $X \subseteq E(G)$ with $|X| \leq \gamma$, let $G' = G[E\setminus X]$, and let $Y$ be a maximum cut in $G$. Then, $Y': = Y\setminus X$ is a cut in $G'$ of size at least $|Y| - |X|$; hence, $c' =1$.
\end{proof}

\begin{lemma}\label{lemma:MC-sl}
 \MCfull can be structurally lifted with respect to edge deletion with constant $c = 0$.
\end{lemma}
\begin{proof}
Given $G$ and any set $X \subseteq E(G)$ with $|X| \leq \gamma$, let $G' = G[E\setminus X]$.
A cut $Y\subseteq E(G')$ is trivially a valid cut in $G$ and consequently $c=0$.
\end{proof}

\begin{corollary}\label{cor:mc-approx}
  \MCfull has
  $(1 - O(\delta \log n\log\log n))$-approximations for graphs
  $(\delta \cdot \opt(G))$-close to treewidth~$w$ via edge deletions where $w\log w = O(\log n)$.
\end{corollary}

\begin{proof}
  We apply Theorem~\ref{thm:general-edit-sr} using
  stability with constant $c'=1$ (Lemma~\ref{lemma:MC-stable})
  and structural lifting with constant $c=0$ for \MC (Lemma~\ref{lemma:MC-sl}).
  For treewidth $\Tw$, we use the
  $(O(\log n\log\log n), O(\log \Tw))$-approximate editing algorithm
  of Bansal~\etal~\cite{bansal2017lp}
  and an exact algorithm for \MC given a tree-decomposition of width $O(w\log w)$ of the edited graph.

  Thus $\alpha=O(\log n\log\log n)$ and $c' = 1$ for \MC ,
  resulting in an approximation factor of $1 + O(\log n\log\log  n) \delta$.
  Note that since the edited graph has treewidth $O(w\log w) = O(\log n)$, the exact algorithm runs in polynomial-time.
\end{proof}

\subsection{Vertex Deletion for Annotated Problems ($\text{Vertex}^*$ Deletion)}\label{sec:annotated-vertex-deletion}
In this section, we show that several important variants of {\em annotated} \DSfull (\ALDS) (which include their non-annotated variants as special cases) are closed under a relaxed version of vertex deletion, denoted by $\text{vertex}^*$ deletion, which is sufficient to apply the  structural rounding framework. Given an instance of \ALDSfull with input graph $G=(V,E)$ and a subset of vertices $B$, the resulting \ALDS instance $(G', B')$ after deleting the set $X\subset V$ is defined as follows: $G' = (V\setminus X, E[V\setminus X])$ and $B' = B\setminus N_{\ell}[X]$ where $N_{\ell}[X]$ denotes the set of all vertices at distance at most $\ell$ from $X$ in $G$.

\begin{lemma}\label{lem:annotated-dom-set-stable}
For $\ell\geq 1$, \ALDSfull is stable under $\text{vertex}^*$ deletion with $c' =0$.
\end{lemma}
\begin{proof}
Note that \ALDSfull with $B = V$ reduces to \LDSfull and in particular \LDSfull is stable under $\text{vertex}^*$ deletion with constant $c' = 0$.

Let $(G', B')$ denote the \ALDS instance after performing $\text{vertex}^*$ deletion with edit set $X$; $G' = (V\setminus X, E[V\setminus X])$ and $B' = B\setminus N_{\ell}[X]$ where $N_{\ell}[X]$ denotes the set of all vertices at distance at most $\ell$ from $X$ in $G$.
Moreover, let $\opt(G, B)$ denote an optimal solution of $\ALDS(G, B)$. We show that $\opt(G,B)\setminus X$ is a feasible solution of $\ALDS(G', B')$. 
Since $X$ $\ell$-dominates $N_{\ell}[X]$, the set $B\setminus N_{\ell}[X]$ is $\ell$-dominated by $\opt(G,S) \setminus X$; hence, $\opt(G,S)\setminus X$ is a feasible solution of $\ALDS(G', S')$. Thus $|\opt(G',S')| \leq |\opt(G,S)\setminus X| \leq |\opt(G,S)|$.
\end{proof}

\begin{lemma}\label{lem:annotated-dom-set-lifting}
For $\ell\geq 1$, \ALDSfull can be structurally lifted with respect to $\text{vertex}^*$ deletion with constant $c = 1$.
\end{lemma}
\begin{proof}
Note that \ALDSfull with $B = V$ reduces to \LDSfull and in particular \LDSfull can be structurally lifted with respect to $\text{vertex}^*$ deletion with constant $c=1$.

Let $(G', B') = ((V\setminus X, E[V\setminus X]), B\setminus N_{\ell}[X])$ denote the \ALDS instance after performing $\text{vertex}^*$ deletion with edit set $X$ on $\ALDS(G,B)$ and let $\opt(G',B')$ denote an optimal solution of $\ALDS(G', B')$ instance. Since the set $X$ $\ell$-dominates $N_{\ell}[X]$, $\opt(G',B')\cup X$ $\ell$-dominates $B' \cup N_{\ell}[X] =B$. Hence, $|\opt(G,B)| \leq |\opt(G',B')| + |X|$.
\end{proof}

\begin{corollary}\label{cor:annotated-dom-set-vertex-degen}
\ADSfull has an $O(\Tr+\delta)$-approximation for graphs $(\delta\cdot \opt(G))$-close to degeneracy $\Tr$ via vertex deletion.
\end{corollary}
\begin{proof}
  We apply Theorem~\ref{thm:general-edit-sr} using
  stability with constant $c'=0$ (Lemma~\ref{lem:annotated-dom-set-stable})
  and structural lifting with constant $c=1$ (Lemma~\ref{lem:annotated-dom-set-lifting}).

  We use a $(O(1), O(1))$-approximate editing algorithm (Section~\ref{section:positive_degeneracy_localratio}/~\ref{section:positive_degeneracy_lp}) and $O(\Tr)$-approximation algorithm for the problem of interest~\cite{bansal2017tight} in $\Tr$-degenerate graphs. Note that although the algorithm of~\cite{bansal2017tight} is for \DSfull, it can easily be modified to work for the annotated variant. Thus, $\alpha = O(1)$ and $c=1$, resulting in an $O(r+\delta)$-approximation algorithm.
\end{proof}

\begin{corollary}\label{cor:annotated-dom-set-vertex-treewidth}
\ADSfull has an $O(1 + O(\delta \log^{1.5} n))$-approximation for graphs $(\delta\cdot \opt(G))$-close to treewidth $\Tw$ via vertex deletion where $w\sqrt{\log w} = O(\log_{\ell} n)$.
\end{corollary}
\begin{proof}
  We apply Theorem~\ref{thm:general-edit-sr} using
  stability with constant $c'=0$ (Lemma~\ref{lem:annotated-dom-set-stable})
  and structural lifting with constant $c=1$ (Lemma~\ref{lem:annotated-dom-set-lifting}).

  We use our $(O(\log^{1.5} n), O(\sqrt{\log \Tw}))$-approximate editing algorithm
  (Section~\ref{section:positive_treewidth}) and an exact polynomial-time algorithm for
  the problem of interest \cite{borradaile_et_al:LIPIcs:2017:6919} given the tree-decomposition of width $O(w\sqrt{\log w})$ of the edited graph. Note that the algorithm of~\cite{borradaile_et_al:LIPIcs:2017:6919} is presented for \LDS; however, by slightly modifying the dynamic programming approach it works for the annotated version as well.
  Thus $\alpha=O(\log^{1.5} n)$ and $c = 1$,
  resulting in an approximation factor of $(1 + O(\log^{1.5} n) \delta)$.
  Moreover, since the edited graph has treewidth $O(w\sqrt{\log w}) = O(\log n)$, the exact algorithm runs in polynomial-time.
\end{proof}

\paragraph{Smarter $\text{Vertex}^*$ Deletion.}
The idea of applying edit operations on annotated problems can also be used for non-annotated problems.
More precisely,  for several optimization problems that fail to satisfy the required conditions of the standard structural rounding under vertex deletion, we can still apply our structural rounding framework with a more careful choice of the subproblem that we need to solve on the edited graph. An exemplary problem in this category is \CDSfull (\CDS). Note that \CDSfull is not stable under vertex deletion and the standard structural rounding framework fails to work for this problem. Besides the stability issue, it is also non-trivial how to handle the connectivity constraint under vertex or edge deletions.
However, in what follows we show that if we instead solve a {\em slightly different problem} (i.e. annotated variant of \CDSfull) on the edited graph, then we can guarantee an improved approximation factor for \CDS on the graphs close to a structural class.

Let $G=(V,E)$ be an input graph that is $(\delta\cdot\opt(G))$-close to the class $\mC$ and let $X\subset V$ be a set of vertices so that $G\setminus X \in \mC$.
For a subset of vertices $X$, $\neighb_G(X)$ is defined to be the set of all neighbors of $X$ excluding the set $X$ itself; $\neighb_{G}(X) := \set{u \;|\; uv\in E(G), v\in X \text{ and } u\notin X}$\footnote{We drop the $G$ in $\neighb_G$ when it is clear from the context.}.
Let $G'= G[V\setminus X]$ be the resulting graph after removing the edit set $X$. The problem that we have to solve on $G'$ is an {\em annotated} variant of \CDS which is defined as follows:

\begin{problem}{\ACDSfull}
\Input & An undirected graph $G = (V, E)$, a subset of vertices $B\subset V$ and $\ell$ vertex-disjoint cliques $K_1 = (V_1, E_1), \cdots, K_\ell = (V_{\ell},E_\ell)$ where for each $i$, $V_i\subset V$.\\
\Prob & Find a minimum size set of vertices $S \subseteq V$ s.t. $S$ dominates all vertices in $B$ and $S$ induces a connected subgraph in $G\cup (\bigcup_{i\in [\ell]} K_i)$.
\end{problem}

To specify the instance of \ACDSfull that we need to solve on the edited graph $G'$, we construct an auxiliary graph $\bar{G} = (\neighb_G(X), \bar{E})$ as follows: $uv\in \bar{E}$ if there exists a $uv$-path in $G$ whose intermediate vertices are all in $X$.

First, we show that \CDS is stable under $\text{vertex}^*$ deletion with constant $c'=0$: the size of an optimal solution of $\ACDS(G', B', K_1,\cdots, K_\ell)$ is not more than the size of an optimal solution of $\CDS(G)$ where $\set{K_1,\cdots, K_{\ell}}$ are the connected components of $\bar{G}$. Note that due to the transitivity of connectivity for each $i\in [\ell]$, $K_i$ is a clique.

\begin{lemma}\label{lem:cds-stable}
\CDSfull is stable under $\text{vertex}^*$ deletion with $c' =0$.
\end{lemma}
\begin{proof}
Let $\opt$ be an optimal solution of $\CDS(G)$. Here, we show that $\opt\setminus X$ is a feasible solution of $\ACDS(G' = G[V\setminus X], B' = V\setminus \neighb_G(X), K_1, \cdots, K_\ell)$ where $K_1, \cdots, K_\ell$ are connected the components of $\bar{G}$ as constructed above. This in particular implies that
\[\opt(G', B', K_1, \cdots, K_\ell) \leq |\opt\setminus X| \leq |\opt| = \opt(G).\]

Since $\opt$ dominates $V$, it is straightforward to verify that $\opt\setminus X$ dominates $B'$ in $G'$. Next, we show that $\opt\setminus X$ is connected in $G'$ when for each $i$, all edges between the vertices of $K_i$ are added to $G'$. Suppose that there exists a pair of vertices $u,v \in \opt\setminus X$  that are not connected in $G'$. However, since $\opt$ is connected, there exists a $uv$-path $P_{uv}$ in $\opt$.
If $P_{uv}$ does not contain any vertices in $X$, then $P_{uv}$ is contained in $\opt\setminus X$ as well and it is a contradiction.
Now consider all occurrences of the vertices of $X$ in $P_{uv}$. We show that each of them can be replaced by an edge in one of the $K_i$s: for each subpath $v_0, x_1, \cdots, x_{q}, v_1$ of $P_{uv}$ where $x_i\in X$ for all $i\in [q]$ and $v_0,v_1\in \neighb_G(X)$, $v_0v_1$ belongs to the same connected component of $\bar{G}$ . Hence, given $P_{uv}$, we can construct a path $P'_{uv}$ in $G'\cup (\bigcup_{i\in [\ell]} K_i)$. Thus, $\opt\setminus X$ is a feasible solution of $\ACDS(G', B', K_1, \cdots, K_\ell)$.
\end{proof}

Next, we show that a solution of the \ACDSfull  instance we solve on the edited graph can be structurally lifted to a solution for \CDSfull on the original graph with constant $c=3$.
\begin{lemma}\label{lem:cds-lifting}
\CDSfull can be structurally lifted under $\text{vertex}^*$ deletion with constant $c = 3$.
\end{lemma}
\begin{proof}
Let $\opt$ be an optimal solution of $\ACDS(G' = G[V\setminus X], B' = V\setminus \neighb_G(X), K_1, \cdots, K_\ell)$ where $K_1, \cdots, K_\ell$ are the connected components of $\bar{G}$ as constructed above.
Here, we show that $\opt\cup X \cup Y$ is a feasible solution of $\CDS(G)$ where $Y$ is a subset of $V\setminus X$ such that $|Y| \leq 2|X|$.
First, it is easy to see that since $X$ dominates $\neighb_G(X) \cup X$ in $G$, $\opt\cup X$ is a dominating set of $G$.
Next, we show that in polynomial time we can find a subset of vertices $Y$ of size at most $2|X|$ such that $\opt\cup X\cup Y$ is a connected dominating set in $G$.

Note that if the subgraph induced by the vertex set $\opt$ on $G'\cup (\bigcup_{i\in [\ell]}K_i)$ contains an edge $uv$ which is not in $E(G')$, the edge can be replaced by a $uv$-path in $G$ whose intermediate vertices are all in $X$. Hence, we can replace all such edges in $\opt$ by including a subset of vertices $X' \subseteq X$ and the set $\opt \cup X'$ remains connected in $G$. At this point, if $X = X'$, we are done: $\opt\cup X$ is a connected dominating set in $G$. Suppose this is not the case and let $X_1:= X\setminus X'$ and $Y_1:= \neighb_{G}(X_1)\setminus \neighb_{G}(X')$. Since $G$ is connected, there exists a path from $X_1$ to $\opt \cup X'$. Moreover, we claim that there exists a path of length at most $4$ from $X_1$ to $\opt\cup X'$. Recall that $\opt\cup X'$ dominates $V\setminus (X_1 \cup Y_1)$. Hence, the shortest path from of $X_1$ to $\opt\cup X'$ has length at most $4$.
We add the vertices on the shortest path which are in $X \setminus X'$ to $X'$ and the vertices in $V \setminus (X \cup \opt \cup Y)$ to $Y$, and update the sets $X_1$ and $Y_1$ accordingly.
Thus we reduce the size of $X_1$ and as we repeat this process it eventually becomes zero. At this point $X = X'$ and $\opt \cup X \cup Y$ is a connected dominating set in $G$. Since, we pick up at most three vertices per each $x\in X_1$ and at least one is in $X$, the set $X \cup Y$ has size at most $3|X|$.
\end{proof}

\begin{corollary}\label{cor:conected-dom-set-vertex-treewidth}
\CDSfull has $O(1 + O(\delta \log^{1.5} n))$-approximation for graphs $(\delta\cdot \opt(G))$-close to treewidth $\Tw$ via vertex deletion where $w$ is a fixed constant.
\end{corollary}
\begin{proof}
  We apply Theorem~\ref{thm:general-edit-sr} using
  stability with constant $c'=0$ (Lemma~\ref{lem:cds-stable})
  and structural lifting with constant $c=3$ (Lemma~\ref{lem:cds-lifting}).

  We use our $(O(\log^{1.5} n), O(\sqrt{\log \Tw}))$-approximate editing algorithm
  (Section~\ref{section:positive_treewidth}) and an exact polynomial-time algorithm for
  \ACDS given the tree-decomposition of width $O(w\sqrt{\log w})$ of the edited graph.
  The FPT algorithm modifies the $\Tw^{O(\Tw)}\cdot n^{O(1)}$ dynamic-programming approach
  of \DS such that it incorporates the annotated sets and cliques $K_1, \cdots, K_\ell$ which then
  runs in $(\Tw + \ell)^{O(\Tw)}\cdot n^{O(1)} = n^{O(\Tw)}$.
  Thus $\Tw = O(1)$, $\alpha=O(\log^{1.5} n)$ and $c = 3$, resulting in an algorithm that runs in
  polynomial time and constructs a $(1 + O(\log^{1.5} n) \delta)$-approximate solution.
\end{proof}
\fi


Although we do not present any editing algorithms for edge contractions,
we point out that such an editing algorithm would enable our framework to
apply to additional problems such as \Problem{(Weighted) TSP Tour}
(which is closed under edge contractions and
can be structurally lifted with constant $c=2$ \cite{demaine2011contraction}),
and to apply more efficiently to other problems such as
\DSfull (reducing $c'$ from $1$ to $0$).

  \section{Editing Algorithms}
\label{section:positive_results}

\ifdefined\islocalratio
    \subsection{Degeneracy: Density-Based Bicriteria Approximation}
    \label{section:positive_degeneracy_localratio}
    
\noindent In this section we prove the following:
\begin{theorem}
\label{theorem:local_ratio_degeneracy}
\bDEV has a
$\left( \tfrac{4m - \beta \targetdegeneracy n}{m - \targetdegeneracy n}, \beta\right)$-approximation algorithm.
\end{theorem}

Observe that this yields a $(4,4)$-approximation when $\beta = 4$. The algorithm is defined in Algorithm \ref{algorithm:local_ratio_degeneracy},
and the analysis is based on the local ratio theorem from Bar-Yehuda et al.~\cite{bar2004local}.

\subsubsection{Analysis overview and the local ratio theorem}
Fundamentally, the local ratio theorem \cite{bar2004local} is machinery for
showing that ``good enough'' local choices accumulate into a global approximation
bound.
This bookkeeping is done by maintaining \textit{weight vectors} that encode the choices made.
The local ratio theorem applies to optimization problems of the following form: given a weight vector $w \in \mathbb{R}^n$
and a set of feasibility constraints $\mathcal{C}$, find a solution vector $x \in \mathbb{R}^n$ satisfying
the constraints $\mathcal{C}$ and minimizing $w^T x$ (for
maximization problems see~\cite{bar2004local}). We say a solution $x$ to such a problem is
\emph{$\alpha$-approximate with respect to $w$} if $w^T x \leq \alpha \cdot \min\limits_{z \in \mathcal{C}} (w^T z)$.

\begin{theorem}[Local Ratio Theorem \cite{bar2004local}] \label{theorem:local-ratio-theorem}
Let $\mathcal{C}$ be a set of feasibility constraints on vectors in
$\mathbb{R}^n$. Let $w, w_1, w_2 \in \mathbb{R}^n$ be such that
$w = w_1 + w_2$. Let $x \in \mathbb{R}^n$ be a feasible solution (with respect
to $\mathcal{C}$) that is $\alpha$-approximate with respect to $w_1$, and with
respect to $w_2$. Then $x$ is $\alpha$-approximate with respect to $w$ as well.
\end{theorem}

In our case, an instance of \bbDEVfull (abbreviated \bbDEV) is represented with ($G$, $w$, $\targetdegeneracy$, $\beta$), where $G$ is the graph, $w$ is a weight vector on the vertices (where $w$ is the all-ones vector, $\overrightarrow{1}$, when $G$ is unweighted), $\targetdegeneracy$ is our target degeneracy, and $\beta$ is a multiplicative error on the target degeneracy. Our bicriteria approximation algorithm will yield an edit set to a $(\beta \targetdegeneracy)$-degenerate graph, using at most $\alpha \cdot \optsol{\bbDEV}{G, w, \targetdegeneracy, \beta}$ edits. This (weighted) cost function is encoded as an input vector of vertex weights $w$, which is evaluated with an indicator function $\mathcal{I}_X$ on a feasible solution $X$, such that the objective is to minimize $w^T \mathcal{I}_X$. Note that while the local ratio theorem can allow all feasible solutions, we require \emph{minimal} feasible solutions for stronger structural guarantees.

\begin{algorithm}
\caption{Approximation for \bDEVfull} \label{algorithm:local_ratio_degeneracy}
\begin{algorithmic}[1]
\Procedure{LocalRatioRecursion}{Graph $G$, weights $w$, target degeneracy
\targetdegeneracy, error $\beta$}
\If{$V(G) = \emptyset$}
\State \textbf{return} $\emptyset$.
\ElsIf{$\exists~v \in V(G)$ where $\deg_G(v) \leq \beta \targetdegeneracy$}
\State \textbf{return} \Call{LocalRatioRecursion}{$G \setminus \{v\}$, $w$,
\targetdegeneracy, $\beta$}
\ElsIf{$\exists~v \in V(G)$ where $w(v) = 0$}
\State $X \gets$ \Call{LocalRatioRecursion}{$G \setminus \{v\}$, $w$,
\targetdegeneracy, $\beta$}
\If{$G \setminus X$ has degeneracy $\beta \targetdegeneracy$}
\State \textbf{return} $X$.
\Else
\State \textbf{return} \Call{MinimalSolution}{$G$, $X \cup \{v\}$,
\targetdegeneracy, $\beta$}.
\EndIf
\Else
\State Let $\epsilon := \min_{v \in V(G)} \frac{w(v)}{\deg_G(v)}$.
\State Define $w_1(u) := \epsilon \cdot \deg_G(u)$ for all $u \in V$.
\State Define $w_2 := w - w_1$.
\State \textbf{return} \Call{LocalRatioRecursion}{$G$, $w_2$,
\targetdegeneracy, $\beta$}.
\EndIf
\EndProcedure
\end{algorithmic}
\end{algorithm}

To utilize the local ratio theorem, our strategy is to define a recursive function that decomposes the weight vector into $w = w_1 + w_2$ and then recurses on $(G, w_2, \targetdegeneracy, \beta)$. By showing that the choices made in this recursive function lead to an $(\alpha, \beta)$-approximation for the instances $(G, w_1, \targetdegeneracy, \beta)$ and $(G, w_2, \targetdegeneracy, \beta)$, by the local ratio theorem, these choices also sum to an $(\alpha, \beta)$-approximation for $(G, w, \targetdegeneracy, \beta)$.

As outlined in \cite[Section 5.2]{bar2004local}, the standard algorithm
template for this recursive method handles the following cases: if a zero-cost minimal solution can be found, output this optimal
solution, else if the problem contains a zero-cost element, do a problem
size reduction, and otherwise do a weight decomposition.

\begin{algorithm}
\caption{Subroutine for guaranteeing minimal solutions}
\label{algorithm:local_ratio_minimality}
\begin{algorithmic}[1]
\Procedure{MinimalSolution}{Graph $G$, edit set $X$, target degeneracy
\targetdegeneracy, error $\beta$}
\For{vertex $v \in V(G)$}
\If{$G \setminus (X \setminus \{v\})$ has degeneracy $\beta \targetdegeneracy$}
\State \textbf{return} \Call{MinimalSolution}{$G$, $X \setminus \{v\}$,
\targetdegeneracy, $\beta$}.
\EndIf
\EndFor
\State \textbf{return} $X$.
\EndProcedure
\end{algorithmic}
\end{algorithm}

Algorithm \ref{algorithm:local_ratio_degeneracy} follows this structure: Lines
2-3 are the first case, Lines 4-12 are the second case, and Lines 13-18 are the
third case. The first two cases are typically straightforward, and the crucial
step is the weight decomposition of $w = w_1 + w_2$. Note that the first case guarantees that all vertices in $G$ have degree at least $\beta \targetdegeneracy + 1$ before a weight decomposition is executed, so we may assume \WLOG that the original input graph also has minimum degree $\beta \targetdegeneracy + 1$.

In the following subsections we show that
Algorithm~\ref{algorithm:local_ratio_degeneracy} returns a minimal, feasible
solution (Lemma~\ref{theorem:local_part_1}), that the algorithm returns an $(\alpha, \beta)$-approximate solution with respect to $w_1$ (Theorem~\ref{theorem:local_part_2}), and finally that the algorithm returns an $(\alpha, \beta)$-approximate solution with respect to $w$ (Theorem~\ref{theorem:local_ratio_degeneracy}).

\subsubsection{$\left( \frac{4m - \beta \targetdegeneracy n}{m - \targetdegeneracy n}, \beta\right)$-approximation for vertex deletion}

\begin{lemma}\label{theorem:local_part_1}
Algorithm \ref{algorithm:local_ratio_degeneracy} returns minimal, feasible
solutions for \bbDEV.
\end{lemma}
\begin{proof}
We proceed by induction on the number of recursive calls. In the base case,
only Lines 2-3 will execute, and the empty set is trivially a minimal, feasible
solution. In the inductive step, we show feasibility by constructing the
degeneracy ordering. We consider each of the three branching cases not covered
by the base case:

\begin{itemize}
\item Lines 4-5: Given an instance $(G, w, \targetdegeneracy, \beta)$, if a vertex $v$ has degree at most $\beta \targetdegeneracy$, add $v$ to the degeneracy ordering and remove it from the graph. By the induction hypothesis, the algorithm will return a minimal, feasible solution $X_{\beta \targetdegeneracy}$ for $(G-\{v\}, w, \targetdegeneracy, \beta)$. By definition, $v$ has at most
$\beta \targetdegeneracy$ neighbors later in the ordering (e.g. neighbors in
$G - \{v\}$), so the returned $X_{\beta \targetdegeneracy}$ is still a feasible, minimal solution.

\item Lines 6-12: Given an instance $(G, w, \targetdegeneracy, \beta)$, if a vertex $v$ has weight 0, remove $v$ from the graph. By the induction hypothesis, the algorithm will return a minimal, feasible solution
$X_{\beta \targetdegeneracy}$ for $(G-\{v\}, w, \targetdegeneracy, \beta)$. If $X_{\beta \targetdegeneracy}$ is a feasible solution
on the instance $(G, w, \targetdegeneracy, \beta)$, then $X_{\beta \targetdegeneracy}$ will be returned
as the minimal, feasible solution for this instance. Otherwise the solution
$X_{\beta \targetdegeneracy} \cup \{v\}$ is feasible, and can be made minimal with a straightforward greedy subroutine (Algorithm~\ref{algorithm:local_ratio_minimality}).

\item Lines 13-18: In this case, no modifications are made to the graph, therefore the recursive call's minimal, feasible solution $X_{\beta \targetdegeneracy}$ remains both minimal and
feasible.
\end{itemize}

In all cases, a minimal, feasible solution is returned.
\end{proof}

We now show that a minimal, feasible solution is $(\alpha, \beta)$-approximate with respect to the instance defined by weight function $w_1$:

\begin{theorem}
\label{theorem:local_part_2}
Any minimal, feasible solution $X_{\beta \targetdegeneracy}$ is a $\left( \tfrac{4m-\beta rn}{m-\targetdegeneracy n}, \beta \right)$-approximation to the instance $(G, w_1, \targetdegeneracy, \beta)$.
\end{theorem}

Given a minimal, feasible solution $X_{\beta \targetdegeneracy}$, note that $w_1^T \mathcal{I}_{X_{\beta \targetdegeneracy}} = \epsilon \sum_{v \in X_{\beta \targetdegeneracy}}\deg_G(v)$. Therefore it suffices to show that $b \leq \sum_{v \in X_\targetdegeneracy}\deg_G(v)$ and $\sum_{v \in X_{\beta \targetdegeneracy}}\deg_G(v) \leq \alpha b$, for some bound $b$, any minimal, feasible edit set $X_\targetdegeneracy$ to degeneracy $\targetdegeneracy$, and any minimal, feasible edit set $X_{\beta \targetdegeneracy}$ to degeneracy $\beta \targetdegeneracy$. We prove these two bounds for $b = m-\targetdegeneracy n$ in Lemmas~\ref{lemma:degen_lower_bound}
and~\ref{lemma:degen_upper_bound}, respectively.

\begin{lemma}
\label{lemma:degen_lower_bound}
For any minimal feasible solution $X_\targetdegeneracy$ for editing to degeneracy $\targetdegeneracy$,
\begin{align*}
m - \targetdegeneracy n \leq \sum_{v \in X_\targetdegeneracy} \deg_G(v).
\end{align*}
\end{lemma}

\begin{proof}
Since $G \setminus X_\targetdegeneracy$ has degeneracy $\targetdegeneracy$, it has at most $\targetdegeneracy n$ edges, so at least $m - \targetdegeneracy n$ edges were deleted. Each deleted edge had at least one endpoint in $X_\targetdegeneracy$, therefore $m - \targetdegeneracy n \leq \sum_{v \in X_\targetdegeneracy} \deg_G(v)$.
\end{proof}

Before proving the upper bound, we define some notation. Let $X_{\beta \targetdegeneracy}$ be a minimal, feasible solution to \bbDEV and let $Y = V(G) \setminus X_{\beta \targetdegeneracy}$ be the vertices in the $(\beta \targetdegeneracy)$-degenerate graph. Denote by $m_X$, $m_Y$, and $m_{XY}$ the number of edges with both endpoints in $X_{\beta \targetdegeneracy}$, both endpoints in $Y$, and one endpoint in each set, respectively. We begin by bounding $m_{XY}$:

\begin{lemma}\label{lemma:m_bound}
For any $X_{\beta \targetdegeneracy}$, it holds that $m_{XY} \leq 2m_Y + 2m_{XY} - \beta \targetdegeneracy |Y|$.
\end{lemma}

\begin{proof}
Recall that we may assume \WLOG that every vertex in $G$ has degree at least $\beta \targetdegeneracy + 1$. Therefore $\beta \targetdegeneracy |Y| \leq \sum_{v \in Y} \deg_G(v) \leq 2m_Y + m_{XY}$, and so $m_{XY} \leq 2m_Y + 2m_{XY} - \beta \targetdegeneracy |Y|$.
\end{proof}

\begin{corollary}
\label{corollary:m_bound}
For any $X_{\beta \targetdegeneracy}$, it holds that $-\beta \targetdegeneracy |X_{\beta \targetdegeneracy}| \geq -2m_Y - 2m_{XY} + \beta \targetdegeneracy |Y|$.
\end{corollary}

\begin{proof}
Because $X_{\beta \targetdegeneracy}$ is minimal, every vertex in $X_{\beta \targetdegeneracy}$ will induce a $(\beta \targetdegeneracy + 1)$-core with vertices in $Y$ if not removed. Therefore each such vertex has at least $(\beta \targetdegeneracy + 1)$-neighbors in $Y$, and
$\beta \targetdegeneracy|X_{\beta \targetdegeneracy}| \leq m_{XY}$. Substituting into Lemma \ref{lemma:m_bound}, we find that
$-\beta \targetdegeneracy |X_{\beta \targetdegeneracy}| \geq -2m_Y - m_{XY} + \beta \targetdegeneracy |Y|$.
\end{proof}

\noindent We now prove the upper bound:
\begin{lemma}\label{lemma:degen_upper_bound}
For any minimal, feasible solution $X_{\beta \targetdegeneracy}$ to \bbDEV,
\begin{align*}
\sum_{v \in X_{\beta \targetdegeneracy}} \deg_G(v) \leq 4m - \beta \targetdegeneracy n.
\end{align*}
\end{lemma}

\begin{proof}
By using substitutions from Lemmas \ref{lemma:m_bound} and Corollary \ref{corollary:m_bound}, we know that
\begin{align*}
\sum_{v \in X_{\beta \targetdegeneracy}} \deg_G(v) &= 2m_{X} + m_{XY}\\
&\leq 2m_X + 2m_Y + 2m_{XY} - \beta \targetdegeneracy |Y|\\
&= 2m - \beta \targetdegeneracy |Y|\\
&= 2m + 2m_Y + 2m_{XY} - 2m_Y - 2m_{XY} + \beta \targetdegeneracy |Y| - 2 \beta \targetdegeneracy |Y|\\
&\leq 2m + 2m_Y + 2m_{XY} - \beta \targetdegeneracy |X_{\beta \targetdegeneracy}| - 2 \beta \targetdegeneracy |Y|\\
&\leq 4m - \beta \targetdegeneracy n.
\end{align*}
\end{proof}

\begin{proof}[Proof of Theorem~\ref{theorem:local_part_2}]
Let $X_{\beta \targetdegeneracy}$ be any minimal, feasible solution for editing to a graph of degeneracy $\beta \targetdegeneracy$. By definition of $w_1$ in Algorithm~\ref{algorithm:local_ratio_degeneracy}, it holds that $w_1^T \mathcal{I}_{X_{\beta \targetdegeneracy}} = \epsilon \sum_{v \in X_{\beta \targetdegeneracy}}\deg_G(v)$, and because $\epsilon$ is a constant computed independently of the optimal solution, it suffices to show that $\sum_{v \in X_{\beta \targetdegeneracy}}\deg_G(v)$ has an $\alpha$-approximation.

By Lemma~\ref{lemma:degen_lower_bound}, any minimal, feasible edit set to a degeneracy-$r$ graph has a degree sum of at least $m - \targetdegeneracy n$. If an edit set is allowed to leave a degeneracy-($\beta \targetdegeneracy$) graph, then by Lemma~\ref{lemma:degen_upper_bound}, at most $4m - \beta \targetdegeneracy n$ degrees are added to the degree sum of $X_{\beta \targetdegeneracy}$. Therefore $X_{\beta \targetdegeneracy}$ is $\left( \tfrac{4m-\beta \targetdegeneracy n}{m-\targetdegeneracy n}, \beta \right)$
-approximate with respect to $(G, w_1, \targetdegeneracy, \beta)$.
\end{proof}

We now prove the main result stated at the beginning of this section,
Theorem~\ref{theorem:local_ratio_degeneracy}.

\begin{proof}
For clarity, let
$\alpha := \left( \tfrac{4m-\beta \targetdegeneracy n}{m - \targetdegeneracy n}\right)$;
we prove that Algorithm \ref{algorithm:local_ratio_degeneracy} is an
$(\alpha, \beta)$-approximation. We proceed by induction on the number of recursive
calls to Algorithm \ref{algorithm:local_ratio_degeneracy}. In the base case
(Lines 2-3), the solution returned is the empty set, which is trivially
optimal. In the induction step, we examine the three recursive calls:

\begin{itemize}
\item Lines 4-5: Given an instance $(G, w, \targetdegeneracy, \beta)$, if a
vertex $v$ has degree at most $\beta \targetdegeneracy$, add $v$ to the
degeneracy ordering and remove it from the graph. By the induction hypothesis,
the algorithm will return an $(\alpha, \beta)$-approximate solution $X_{\beta \targetdegeneracy}$ for
$(G-\{v\}, w, \targetdegeneracy, \beta)$. Since $v$ will not be added to $X_{\beta \targetdegeneracy}$,
then $X_{\beta \targetdegeneracy}$ is also an $(\alpha, \beta)$-approximation for $(G, w, \targetdegeneracy, \beta)$.

\item Line 6-12: Given an instance $(G, w, \targetdegeneracy, \beta)$, if a vertex
$v$ has weight 0, remove $v$ from the graph. By the induction hypothesis, the
algorithm returns an $(\alpha, \beta)$-approximate solution $X_{\beta \targetdegeneracy}$ for
$(G-\{v\}, w, \targetdegeneracy, \beta)$. Regardless of whether $v$ is added to
$X_{\beta \targetdegeneracy}$ or not, it contributes exactly zero to the cost of the solution, therefore an $(\alpha, \beta)$-approximation is returned.

\item Line 13-18: In this case, the weight vector is decomposed into $w_1$ and
$w_2 = w - w_1$. By induction, the algorithm will return an $(\alpha, \beta)$-approximate solution $X_{\beta \targetdegeneracy}$ for $(G, w-w_1, \targetdegeneracy, \beta)$. By Theorem
\ref{theorem:local_part_2}, $w_1^T \mathcal{I}_{X_{\beta \targetdegeneracy}}$ is also $(\alpha, \beta)$-approximate. Therefore, by Theorem~\ref{theorem:local-ratio-theorem}, $w^T \mathcal{I}_{X_{\beta \targetdegeneracy}}$ must be $(\alpha, \beta)$-approximate.
\end{itemize}
\end{proof}

\fi

\ifdefined\islp
    \subsection{Degeneracy: LP-based Bicriteria Approximation}
    \label{section:positive_degeneracy_lp}
    
In this section, we design a bicriteria approximation for the problem of minimizing the number of required edits (edge/vertex deletions) to the family of  $\Tr$-degenerate graphs. Consider an instance of \iffull
$\bDEEfull(G, \Tr)$
\fi
and let $\opt$ denote an optimal solution.
The algorithm we describe here works even when the input graph is {\em weighted} (both vertices and edges are weighted) and the goal is to minimize the total weight of the edit set.

Our approach is based on Lemma~\ref{lem:degen-bounded-outdeg-orientation} which we restate here for convenience.
\begin{lemma}\label{thm:k-degenerate}
A graph $G=(V, E)$ is $(2\Tr)$-degenerate if there exists an orientation of the edges in $E$ such that the out-degree of each vertex $v$ is at most $r$.
\end{lemma}

\subsubsection{$(6,6)$-approximation for vertex deletion}
\label{section:positive_degeneracy_LP_vertex}

In what follows we formulate an LP-relaxation for the problem of minimizing the number of required vertex deletions to the family of $\Tr$-degenerate graphs. For each edge $uv\in E$, $x_{uv}$ variable denotes the orientation of $uv$; $x_{\arc{uv}}=1$, $x_{\arc{vu}} =0$ if $uv$ is oriented from $u$ to $v$ and $x_{\arc{vu}}=1$, $x_{\arc{uv}} =0$ otherwise. Moreover, for each vertex $v\in V$ we define $y_{v}$ to denote whether $v$ is part of the edit set $\editset$ ($y_{v}=1$ if $v \in \editset$ and zero otherwise).

\begin{center}
\begin{minipage}{0.75\linewidth}
\begin{probbox}{\lpMinVertex}
	\emph{Input:} \hskip1em $G=(V,E), w, \Tr$
\begin{align*}
	\text{Minimize} & 	&\sum_{v \in V} y_{v} w_{v} &						 &\\
	\text{s.t.}         & 	&x_{\arc{vu}} + x_{\arc{uv}} 			&\geq 1-y_{u}-y_{v}  &\forall uv \in E\\
						    &  &\sum_{u \in N(v)}x_{\arc{vu}} 		 	&\leq \Tr 			&\forall v \in V\\
							&  &x_{\arc{uv}} 								&\geq 0			&\forall uv\in E
\end{align*}
\end{probbox}
\end{minipage}
\end{center}

The first set of constraints in the LP-relaxation \lpMinVertex guarantees that for each edge $uv$ whose none of its endpoints is in $X$, it is oriented either from $v$ to $u$ or from $u$ to $v$. The third set of the constraints ensure that for all $v\in V$, $\outDeg(v)\leq \Tr$. Note that if $v\in \editset$ and thus $y_{v}=1$, then \WLOG we can assume that both $x_{\arc{uv}}$ and $x_{\arc{vu}}$ are set to zero.

\begin{lemma}\label{lem:lp-relaxation-vertex}
$\lpMinVertex(G, w, \Tr)$ is a valid LP-relaxation of $\bDEV(G,w)$.
\end{lemma}
\begin{proof}
Let $X$ be a feasible edit set of $\bDEV(G,w)$. Let $D$ be an $\Tr$-degenerate ordering of $V\setminus X$. We define vectors $(x,y)$ corresponding to $X$ as follows: for each $v\in V$, $y(v) = 0$ if $v\in X$ and zero otherwise. Moreover, $x_{\arc{uv}} =1$ if $u,v \in V\setminus X$ and $u$ comes before $v$ in the ordering $D$; otherwise, $x_{\arc{uv}}$ is set to zero.  

Next, we show that the constructed solution $(x,y)$ satisfies all constraints in \lpMinVertex. Since $x$ only obtains non-negative values, for the first set of constraints we can only consider the set of survived edges after removing set $X$, $E[V\setminus X]$. For these edges, since one of $u$ and $v$ comes first in $D$, exactly one of $x_{\arc{uv}}, x_{\arc{vu}}$ is one and the constraint is satisfied. Lastly, since $D$ is an $\Tr$-degenerate ordering of $V\setminus X$, for each vertex $v\in V\setminus X$, the out-degree is at most $\Tr$. Moreover, for each $v\in X$, the LHS in the second set of constraints is zero. 
\end{proof}
\iffull
Similarly to our approach for \bDEEfull, first
\fi we find an optimal solution $(x,y)$ of \lpMinVertex in polynomial time.
\subparagraph{Rounding scheme.} We prove that the following rounding scheme of \lpMinVertex gives a $({1\over \eps},{4\over 1-2\eps})$-bicriteria approximation for \bDEV.
\begin{align}\label{rule:vertex_round_first}
\hat{y}_{v} =
\left\{
	\begin{array}{ll}
		1  & \mbox{if } y_{v} \geq \eps, \\
		0  & \mbox{otherwise.}
	\end{array}
\right.
\end{align}

\begin{align}\label{rule:vertex_round_second}
\hat{x}_{uv} =
\left\{
	\begin{array}{ll}
		1  & \mbox{if } x_{uv} \geq (1 - 2\eps)/2, \\
		0  & \mbox{otherwise.}
	\end{array}
\right.
\end{align}

\begin{lemma}\label{lem:vertex-edit}
If $(x, y)$ is an optimal solution to \lpMinVertex, then $(\hat{x}, \hat{y})$ as given by Equations \ref{rule:vertex_round_first} and \ref{rule:vertex_round_second} is
 an integral $({1\over \eps},{2\over 1-2\eps})$-bicriteria approximate solution of $\lpMinVertex(G,w,r)$.
\end{lemma}
\begin{proof}
First we show that $(\hat{x}, \hat{y})$ satisfies the first set of constraints: for each $uv \in E$, $\hat{x}_{\arc{vu}} + \hat{x}_{\arc{uv}}\geq 1-\hat{y}_{v} - \hat{y}_{u}$. Note that if either $\hat{y}_v$ or $\hat{y}_u$ is one then the constraint trivially holds. Hence, we assume that both $\hat{y}_v$ and $\hat{y}_u$ are zero.
By Equation~\eqref{rule:vertex_round_first}, this implies that both $y_v$ and $y_u$ have value less than $\eps$. Hence, by feasibility of $(x,y)$,
\begin{align*}
x_{\arc{vu}} + x_{\arc{uv}}\geq 1- y_v - y_u \geq 1-2\eps,
\end{align*}
and in particular, $\max(x_{\arc{vu}}, x_{\arc{vu}}) \geq {1 - 2\eps \over 2}$. Then, by Equation~\eqref{rule:vertex_round_second}, $\max(\hat{x}_{\arc{uv}}, \hat{x}_{\arc{vu}}) = 1$ and the constraint is satisfied: $\hat{x}_{uv} + \hat{x}_{vu} \geq 1 \geq 1-\hat{y}_v-\hat{y}_u$. Note that if both of $\hat{x}_{uv}$ and $\hat{x}_{vu}$ are set to one, we can arbitrarily set one of them to zero.

Moreover, since for each arc $\arc{uv}$, $\hat{x}_{\arc{uv}} \leq \frac{2}{1-2\eps} \cdot  x_{\arc{uv}}$, for each $v\in V$:
\begin{align*}
\sum_{u\in N(v)} \hat{x}_{\arc{vu}} ~{\leq}~ {2\over 1-2\eps}\cdot \sum_{u\in N(v)} x_{\arc{vu}} ~{\leq}~ {2\Tr\over 1-2\eps},
\end{align*}
where the first inequality follows from Equation~\eqref{rule:vertex_round_second} and the second from the feasibility of $(x,z)$.

Finally, since for each $v\in V$, $\hat{y}_v \leq y_v/\eps$, the cost of the rounded solution $(\hat{x},\hat{y})$ is at most
\begin{align*}
\sum_{u\in V} \hat{y}_{v} w_v~ {\leq} ~\frac{1}{\eps}\cdot \sum_{u\in V} y_{v} w_v~ {\leq} ~{\textstyle{\optsol{\bDEV}{G, w, \targetdegeneracy}}}/\eps,
\end{align*}
where the first inequality follows from Equation~\eqref{rule:vertex_round_first} and the second directly from the optimality of $(x,y)$.
Hence, $(\hat{x}, \hat{y})$ is an integral $({1\over \eps},{2\over 1-2\eps})$-bicriteria approximate solution of $\lpMinVertex(G, w,\Tr)$.
\end{proof}
Note that, the integral solution $(\hat{x}, \hat{y})$ specifies an edit set $X:= \set{v\in V | \hat{y}(v) =1}$ and orientation of edges $D:=\set{\arc{uv} | x_{\arc{uv} = 1}}$ such that for each $v\in V\setminus X$, $\outDeg(v) \leq {2r \over 1-2\eps}$. Hence, together with Lemma~\ref{thm:k-degenerate}, we have the following result.
\begin{corollary}\label{cor:vertex-edit}
There exists a $({1\over \eps}, {4\over 1-2\eps})$-bicriteria approximation for $\bDEV$. 

In particular, by setting $\eps = {1/6}$, there exists a $(6,6)$-bicriteria approximation algorithm for the \bDEVfull problem.
\end{corollary}

\iffull
\subsection{$(5,5)$-approximation for edge deletion}
\label{section:positive_degeneracy_LP_edge}
In what follows we formulate an LP-relaxation for the problem of minimizing the number of required edge edits (deletions) to the family of $\Tr$-degenerate graphs. For each edge $uv\in E$, $x$ variables denote the orientation of $uv$; $x_{\arc{uv}}=1$, $x_{\arc{vu}} =0$ if $uv$ is oriented from $u$ to $v$ and $x_{\arc{vu}}=1$, $x_{\arc{uv}} =0$ if $e$ is oriented from $v$ to $u$. Moreover, for each $uv$ we define $z_{uv}$ to denote whether the edge $uv$ is part of the edit set $\editset$ ($z_{uv}=1$ if the edge $uv \in \editset$ and zero otherwise).

\begin{center}
\begin{minipage}{0.75\linewidth}
\begin{probbox}{\lpMinEdge}
	\emph{Input:} \hskip1em $G=(V,E), w, \Tr$
\begin{align*}
	\text{Minimize} &  & \sum_{uv \in E} z_{uv} w_{uv} & &\\
	\text{s.t.}     &  & x_{\arc{vu}} + x_{\arc{uv}}   &\geq 1-z_{uv}  &\forall uv \in E\\
						      &  & \sum_{u \in N(v)}x_{\arc{vu}} 		 &\leq \Tr 			&\forall v \in V\\
						 	    &  & x_{\arc{uv}} 								 &\geq 0			&\forall uv\in V\times V
\end{align*}
\end{probbox}
\end{minipage}
\end{center}

The first set of constraints in the LP-relaxation \lpMinEdge guarantee that for each edge $uv\notin \editset$, it is oriented either from $v$ to $u$ or from $u$ to $v$. The second set of the constraints ensure that for all $v\in V$, $\outDeg(v)\leq \Tr$. Note that if an edge $uv\in \editset$ and thus $z_{uv}=1$, then \WLOG we can assume that both $x_{\arc{uv}}$ and $x_{\arc{vu}}$ are set to zero.
\begin{lemma}\label{lem:lp-relaxation-edge}
$\lpMinEdge(G, w, \Tr)$ is a valid LP-relaxation of $\bDEE(G,w)$.
\end{lemma}

Next, we propose a {\em two-phase} rounding scheme for the \lpMinEdge.


\mypar{First phase.} Let $(x, z)$ be an optimal solution of \lpMinEdge. Note that since the \lpMinEdge has polynomial size, we can find its optimal solution efficiently.
Consider the following {\em semi-integral} solution $(x, \hat{z})$ of \lpMinEdge:

\begin{align}\label{rule:round_first}
\hat{z}_{uv} =
\left\{
	\begin{array}{ll}
		1  & \mbox{if } z_{uv} \geq \eps, \\
		0  & \mbox{otherwise.}
	\end{array}
\right.
\end{align}

\begin{claim}\label{clm:first-round}
$({x \over 1- \eps}, \hat{z})$ as given by Equation~\eqref{rule:round_first} is a $({1\over \eps},{1\over 1-\eps})$-bicriteria approximate solution of $\lpMinEdge(G, w, \Tr)$.
\end{claim}
\begin{proof}
First, we show that $({1\over 1- \eps}x,\hat{z})$ satisfies the first set of constraints. For each edge $uv$,
\begin{align*}
{x_{\arc{uv}}\over 1- \eps} + {x_{\arc{vu}}\over 1- \eps} = {1\over 1- \eps}(x_{\arc{uv}} + x_{\arc{vu}}) \geq {1\over 1- \eps}(1-z_{uv}) \geq 1-\hat{z}_{uv},
\end{align*}
where the first inequality follows from the feasibility of $(x,z)$ and the second inequality follows from Equation~\eqref{rule:round_first}. Moreover, it is straightforward to check that as we multiply each $x_{vu}$ by a factor of $1/(1-\eps)$, the second set of constraints are off by the same factor; that is, $\forall v\in V, \sum_{u \in V}x_{\arc{vu}}/(1-\eps)\leq \Tr/(1-\eps)$. Finally, since for each edge ${uv}$, $\hat{z}_{uv} \leq z_{uv}/\eps$, the cost of the edit set increases by at most a factor of $1/\eps$; that is, $\sum_{uv\in E} \hat{z}_{uv} w_{uv} \leq \frac{1}{\eps}\sum_{uv\in E} z_{uv} w_{uv}$.
\end{proof}

\mypar{Second phase.} Next, we prune the fractional solution further to get an {\em integral} approximate {\em nearly feasible} solution of \lpMinEdge. Let $\hat{x}$ denote the orientation of the surviving edges (edges $uv$ such that $\hat{z}_{uv}=0$) given by:
\begin{align}\label{rule:round-second}
\hat{x}_{\arc{uv}} =
\left\{
	\begin{array}{ll}
		1  & \mbox{if } x_{\arc{uv}} \geq (1-\eps)/2, \\
		0  & \mbox{otherwise.}
	\end{array}
\right.
\end{align}

We say an orientation is \emph{valid} if each surviving edge $(u,v)$ is oriented from $u$ to $v$ or $v$ to $u$.

\begin{lemma}\label{lem:valid-orient}
$\hat{x}$ as given by Equation~\eqref{rule:round-second} is a valid orientation of the set of surviving edges.
\end{lemma}
\begin{proof}
We need to show that for each $uv\in E$ with $\hat{z}_{uv} =0$ at least one of $\hat{x}_{\arc{uv}}$ or $\hat{x}_{\arc{vu}}$ is one. Note that if both are one, we can arbitrarily set one of them to zero.

For an edge $uv$, by Equation~\eqref{rule:round_first}, $\hat{z}_{uv} = 0$ iff $z_{uv}\leq \eps$. Then, using the fact that $(x,z)$ is a feasible solution of \lpMinEdge, $x_{\arc{uv}} + x_{\arc{vu}} \geq 1-z_{uv} \geq 1-\eps$. Hence, $\max(x_{\arc{uv}}, x_{\arc{vu}})\geq (1-\eps)/2$ which implies that $\max(\hat{x}_{\arc{uv}}, \hat{x}_{\arc{vu}}) =1$. Hence, for any surviving edge $uv$, at least one of $\hat{x}_{\arc{uv}}$ or $\hat{x}_{\arc{vu}}$ will be set to one.
\end{proof}
\begin{lemma}\label{lem:edge-edit}
$(\hat{x}, \hat{z})$ as given by Equations \ref{rule:round_first} and \ref{rule:round-second} is an integral $({1\over \eps},{2\over 1-\eps})$-bicriteria approximate solution of $\lpMinEdge(G,w,r)$.
\end{lemma}
\begin{proof}
As we showed in Lemma~\ref{lem:valid-orient}, $\hat{x}$ is a valid orientation of the surviving edges with respect to $\hat{z}$. Moreover, by Equation~\eqref{rule:round-second}, for each $uv \in E$, $\hat{x}_{\arc{uv}} \leq 2 x_{\arc{uv}} / (1-\eps) $. Hence, for each vertex $v\in V$, $\outDeg(v) \leq 2\Tr/(1-\eps)$. Finally, as we proved in Claim~\ref{clm:first-round}, the total weight of the edit set defined by $\hat{z}$ is at most $\frac{1}{\eps}$ times the total weight of the optimal solution $(x,z)$.
\end{proof}

Hence, together with Lemma~\ref{lem:degen-bounded-outdeg-orientation}, we have the following result.
\begin{corollary}\label{cor:edge-edit}
There exists a $({1\over \eps},{4\over 1-\eps})$-bicriteria approximation algorithm for \bDEE.

In particular, by setting $\eps = 1/5$, there exists a $(5,5)$-bicriteria approximation algorithm for the \bDEEfull problem.
\end{corollary}
\fi

We note that our approach also works in the general setting
when both vertices and edges are weighted, and we consider an edit operation
which includes both vertex and edge deletion.

\iffull
\subsubsection{Integrality gap of \lpMinEdge and \lpMinVertex}
A natural open question is if we can obtain ``purely multiplicative'' approximation guarantees for $\bDEE$ and $\bDEV$ via LP-based approaches. In this section, we show that the existing LP-relaxation of editing to bounded degeneracy cannot achieve $o(n)$-approximation. These results are particularly important because they show that the best we can hope for are bicriteria approximations.

\begin{theorem}\label{thm:integrality-gap-edge}
The integrality gap of \lpMinEdge is $\Omega(n)$.
\end{theorem}
\begin{proof}
Consider an instance of $\bDEE(G)$ where $G$ is an unweighted complete graph of size $2n$ and $\Tr = n-2$. First, we show that $\lpMinEdge(G, r)$ admits a fractional solution of cost/size $O(n)$ and then we show that the size of any feasible edit set of $\bDEE(G)$ is $\Omega(n^2)$.

Consider the following fractional solution of $\lpMinEdge(G, r)$: for all $uv\in V\times V$ and $u\neq v$, $x_{\arc{uv}} = 1/2 - 1/n$ and for all edges $uv\in E$, $z_{uv} = 2/n$. Note that $x$ and $z$ satisfy the first set of constraints in $\lpMinEdge(G, \Tr)$:
\begin{align*}
\forall uv\in E, \quad x_{\arc{uv}} + x_{\arc{vu}} = 1- 2/n = 1-z_{uv}.
\end{align*}
Moreover, $x$ satisfies the second set of the constraints in $\lpMinEdge(G, \Tr)$
\begin{align*}
\forall v \in V, \quad \sum_{u\in V} x_{\arc{vu}} = (2n-1)(1/2 -1/n) < n-2.
\end{align*}
Finally, $\cost(x, z) = \sum_{uv\in E} z_{uv} = n(2n-1) \cdot (2/n) = 4n-2$ which implies that the cost of an optimal solution of $\lpMinEdge(G, \Tr)$ is $O(n)$.

Next, we show that any integral solution of $\bDEE(G, \Tr)$ has size $\Omega(n^2)$. Let $\editset$ be a solution of $\bDEE(G, \Tr)$. Then, there exits an ordering of the vertices in $G$, $v_1, \cdots, v_{2n}$ such that $\deg(v_i)$ in $G[v_i,\dots,v_{2n}]$ is at most $r\leq n-2$. This implies that for $i\leq n-2$,
$|\delta(v_i) \cap \editset| \geq n+2-i$, where $\delta(v)$ denotes the set of edges incident to a vertex
$v$. Thus,
\begin{align*}
|\editset| \geq {1\over 2}\sum_{i\leq n-2} |\delta(v_i) \cap \editset| \geq {1\over 2}\sum_{i\leq n-2} n+2-i \geq {1\over 2} (n^2 -4 - {(n-2)(n-3) \over 2}) \geq {n^2 / 4}.
\end{align*}
Hence, the integrality gap of \lpMinEdge is $\Omega(n)$.
\end{proof}

\begin{theorem}\label{thm:integrality-gap-vertex}
The integrality gap of \lpMinVertex is $\Omega(n)$.
\end{theorem}
\begin{proof}
Consider an instance of $\bDEV(G)$ where $G$ is an unweighted complete graph of size $2n$ and $\Tr = n-2$. First, we show that $\lpMinVertex(G, r)$ admits a constant size fractional solution and then we show that the size of any feasible edit set of $\bDEE(G)$ is $\Omega(n)$.

Consider the following fractional solution of $\lpMinVertex(G, r)$: for each $uv\in V^2$ and $u\neq v$, $x_{\arc{uv}} = 1/2 - 1/n$ and for each vertex $v\in V$, $z_{v} = 1/n$. First, we show that $x$ and $z$ satisfy the first set of constraints in $\lpMinVertex(G, \Tr)$:
\begin{align*}
\forall uv\in E, \quad x_{\arc{uv}} + x_{\arc{vu}} = 1- 2/n = 1-z_{u} - z_{v}.
\end{align*}
Moreover, $x$ satisfies the second set of the constraints in $\lpMinVertex(G, \Tr)$
\begin{align*}
\forall v \in V, \quad \sum_{u\in V} x_{\arc{vu}} = (2n-1)(1/2 -1/n) < n-2.
\end{align*}
Finally, $\cost(x, z) = \sum_{uv\in E} z_{uv} = 2n \cdot (1/n) = 2$ which implies that the cost of an optimal solution of $\lpMinVertex(G, \Tr)$ is at most $2$.

Next, we show that any integral solution of $\bDEV(G, \Tr)$ has size $\Omega(n)$.
Let $\editset$ be a solution of $\bDEE(G, \Tr)$. Since $G\setminus \editset$ is a complete graph of size $2n-|\editset|$, in order to get degeneracy $n-2$, $|\editset| \geq n+2$.

Hence, the integrality gap of \lpMinVertex is $\Omega(n)$.
\end{proof}
\fi

\fi

\ifdefined\isgreedy
    \subsection{Degeneracy: $O(\log n)$ Greedy Approximation}
    \label{section:positive_degeneracy_one}
    
In this section, we give a polytime $O(\log n)$-approximation for reducing the degeneracy of a graph by one using either vertex deletions or edge deletions.
More specifically, given a graph $G = (V, E)$ with degeneracy $\targetdegeneracy$, we produce an edit set $X$ such that $G' = G \setminus X$ has degeneracy $\targetdegeneracy - 1$ and $|X|$ is at most $O(\log |V|)$ times the size of an optimal edit set.
Note that this complements an $O(\log \frac{n}{\targetdegeneracy})$-approximation hardness result for the same problem.

In general, the algorithm works by computing a vertex ordering and greedily choosing an edit to perform based on that ordering.
In our algorithm, we use the \emph{min-degree ordering} of a graph.
The \emph{min-degree ordering} is computed via the classic greedy algorithm given by Matula and Beck~\cite{matula1983smallest} that computes the degeneracy of the graph by repeatedly removing a minimum degree vertex from the graph.
The degeneracy of $G$, $\degener(G)$, is the maximum degree of a vertex when it is removed.
In the following proofs, we make use of the observation that given a min-degree ordering $L$ of the vertices in $G = (V, E)$ and assuming the edges are oriented from smaller to larger indices in $L$, $\outDeg(u) \leq \degener(G)$\footnote{For notational reminders for $\outDeg(u)$, please refer to Section~\ref{sec:preliminaries}.} for any $u \in L$.

The first ordering $L_0$ is constructed by taking a min-degree ordering on the vertices of $G$ where ties may be broken arbitrarily.
Using $L_0$, an edit is greedily chosen to be added to $X$.
Each subsequent ordering $L_i$ is constructed by taking a min-degree ordering on the vertices of $G \setminus X$ where ties are broken based on $L_{i - 1}$.
Specifically, if the vertices $u$ and $v$ have equal degree at the time of removal in the process of computing $L_i$, then $L_i(u) < L_i(v)$ if and only if $L_{i - 1}(u) < L_{i - 1}(v)$.
The algorithm terminates when the min-degree ordering $L_j$ produces a witness that the degeneracy of $G \setminus X$ is $\targetdegeneracy - 1$.

In order to determine which edit to make at step $i$, the algorithm first computes the forward degree of each vertex $u$ based on the ordering $L_i$ (equivalently, $\outDeg(u)$ when edges are oriented from smaller to larger index in $L_i$).
Each vertex with forward degree $\targetdegeneracy$ is marked, and similarly, each edge that has a marked left endpoint is also marked.
The algorithm selects the edit that \emph{resolves} the largest number of marked edges.
We say that a marked edge is \emph{resolved} if it will not be marked in the subsequent ordering $L_{i + 1}$.

We observe that given an optimal edit set (of size $k$), removing the elements of the set in any order will resolve every marked edge after $k$ rounds (assuming that at most one element from the optimal edit set is removed in each round).
If it does not, then the final ordering $L_k$ must have a vertex with forward degree $\targetdegeneracy$, a contradiction.
Let $m_i$ be the number of marked edges based on the ordering $L_i$.
We show that we can always resolve at least $\frac{m_i}{k}$ marked edges in each round, giving our desired approximation.

\begin{lemma} \label{lemma:unmarked-vertex}
    A vertex that is unmarked in $L_i$ cannot become marked in $L_{j}$ for any $j > i$.
\end{lemma}

\begin{proof}
    For an unmarked vertex $v$ to become marked, its forward degree must increase from $d \leq \targetdegeneracy - 1$ to $\targetdegeneracy$ when going from $L_i$ to $L_j$ for some $j > i$.
    In other words, $\outDeg_{L_i}(v) < \targetdegeneracy$ whereas $\outDeg_{L_j}(v) = \targetdegeneracy$.
    Since edges are not added to $G$, this can only occur if a backward neighbor $u$ of $v$ becomes a forward neighbor.
    Let $\{u, v\}$ be an \emph{inversion} if $L_i(u) > L_i(v)$ but $L_j(u) < L_j(v)$\footnote{Note that $u$ and $v$ do \emph{not} have to be connected by an edge.}. An inversion can occur between neighbors 
    $u'$ and $v'$, in which case $u'$ and $v'$ are connected by an edge.
    We call this a \emph{positive inversion} for $u$ and a \emph{negative inversion} for $v$.
    If the number of positive inversions for $u$ of $u$'s neighbors is greater than the number of negative inversions of $u$'s neighbors between $L_i$ and $L_j$, then $\outDeg_{L_j}(u) > \outDeg_{L_i}(u)$.

    By our previous observation, an unmarked vertex can only become a marked vertex through inversions.
    Let $u$ be a vertex that was unmarked in $L_i$ but becomes marked in $L_j$.
    Let $u$ and $v$ be the first positive inversion for $u$ in $L_i$ (i.e.\ there is not a $w$ such that $u$ and $w$ form a positive inversion for $u$ and $L_i(w) < L_i(v)$).
    Because the algorithm breaks ties when constructing $L_j$ based on $L_{j-1}$, if $u$ and $v$ form a positive inversion for $u$ and $u$ becomes marked in $L_j$, then $\outDeg_{L_j}(u) < \deg_{L_j[i_u, n]}(v)$\footnote{Let $\deg_{L_j[i_u, n]}(v)$ be the degree of $v$ restricted to vertices between indices $i_u$ and $n$ in $L_j$.  Here, $i_u$ is the index of $u$.} and $\outDeg_{L_j}(u) = r$.
    (If, instead, $\outDeg_{L_j}(u) > \deg_{L_j[i_u, n]}(v)$, then $v$ would have been removed first according to $L_j$.)
    Then, either (1) $\deg_{L_j[i_u, n]}(v) \leq \outDeg_{L_i}(v)$ or (2) $\deg_{L_j[i_u, n]}(v) > \outDeg_{L_i}(v)$.

    If (1) occurs, then $\outDeg_{L_j}(u)$ cannot be $\targetdegeneracy$ since this would imply $\outDeg_{L_i}(v) > \targetdegeneracy$, a contradiction.
    However, (2) can only occur through positive inversions of $v$ (in fact, through positive inversions of $v$'s neighbors) since we chose $v$ to be the first positive inversion of $u$. (Hence, $v$ cannot gain additional edges in the range $L_j[i_u, i_v]$ when going from $L_i$ to $L_j$.)
    Let $w$ be the first positive inversion of $v$ in $L_j$. Given that $v$ must have at least one positive inversion with one of its neighbors, $w$ must exist (it can either be the neighbor of $v$ or another node.)
    The same case analysis applies to $v$ and $w$, implying that $w$ has a positive inversion and so on, eventually leading to a contradiction due to a lack of additional vertices to form a positive inversion.
\end{proof}

Using Lemma~\ref{lemma:unmarked-vertex}, we are able to prove a similar statement about marked edges.

\begin{lemma} \label{lemma:unmarked-edge}
    An edge that is unmarked in $L_i$ cannot become marked in $L_{j}$ for any $j > i$.
\end{lemma}

\begin{proof}
    Suppose \WLOG that the edge $e = (u, v)$ is unmarked in $L_i$ and that $L_i(u) < L_i(v)$.
    This implies that $u$ is unmarked in $L_i$.
    By Lemma~\ref{lemma:unmarked-vertex}, $u$ cannot become marked in $L_{j}$.
    Thus, in order for $e$ to become marked, $L_{j}(v) < L_{j}(u)$ and $\outDeg_{L_j}(v) = \targetdegeneracy$.
    Since $u$ is unmarked in $L_j$, we know that $\outDeg_{L_j}(u) \leq \targetdegeneracy - 1$.
    So for $L_{j}(v) < L_{j}(u)$ where $u$ and $v$ are an inversion, the forward degree of $v$ including $u$ must be less than $\targetdegeneracy - 1$.
    Thus, $v$ must be unmarked in $L_{j}$, and so $e$ is also unmarked.
\end{proof}

Lemmas~\ref{lemma:unmarked-vertex} and~\ref{lemma:unmarked-edge} allow us to make a claim about the number of marked edges that any one edit resolves.

\begin{lemma} \label{lemma:resolve-noninc}
    For a given edit $x$, the number of marked edges that it resolves is monotonically non-increasing from $L_0$.
    In other words, for any $i < j$, the number of marked edges $x$ resolves in $L_i$ is at least as many as the number of marked edges $x$ resolves in $L_j$.
\end{lemma}

\begin{proof}
    By Lemma~\ref{lemma:unmarked-edge}, we know that an unmarked edge cannot become marked.
    Thus, for the number of marked edges that an edit resolves in $L_i$ to increase in $L_{i + 1}$, an existing marked edge must become resolvable by making a different edit.
    Note that edits resolve edges by either deleting them or reducing the forward degree of marked vertices.
    Since the back neighbors of an edit $x$ cannot become marked by Lemma~\ref{lemma:unmarked-vertex} and any vertices that form a negative inversion with $x$ via another edit must have forward degree at most $\targetdegeneracy - 1$, it is not possible for $x$ to gain marked edges that are resolvable by deletion.

    Instead, any new resolvable marked edges must be resolved by the deletion of $x$ reducing the degree of a back neighbor by one.
    Note that changes to the set of resolvable edges can only occur if the relative ordering of the neighbors of $x$ changes.
    First, we will consider the case where a forward neighbor $v$ forms an inversion with neighbor $x$.
    After the inversion, the forward degree of $v$ will be one less than the original forward degree of $x$.
    Furthermore, we note that $v$ can now only form an inversion with its last back neighbor $b$ assuming that the forward degree of $b$ is exactly one greater than the forward degree of $v$.
    Thus, the forward edges of $b$ are resolvable by $x$ if they are marked.
    Note that there are exactly $\targetdegeneracy$ of these edges if this is the case.
    However, this implies that the forward degree of $v$ is $\targetdegeneracy - 1$, and so the original forward degree of $x$ must have been $\targetdegeneracy$.
    These $\targetdegeneracy$ edges must have been resolved by the inversion of $v$ and $x$, so $x$ resolves at most the same number of edges as it did originally.
    Note that if multiple forward neighbors form inversions with $x$ simultaneously, only one of them can have forward degree $\targetdegeneracy - 1$.

    Next, we consider the case where a back neighbor $b$ of $x$ forms an inversion with some other neighboring vertex $v$ (potentially of no relation to $x$).
    Again, we rely on the fact that $b$ must have forward degree $\targetdegeneracy - 1$ after the inversion in order to be able to make a second inversion that resolves additional marked edges.
    However, this implies that $v$ originally had forward degree $\targetdegeneracy$, and so when $b$ inverted with $v$, $\targetdegeneracy$ marked edges were resolved.
    Since $b$ can only form an inversion with one of its back neighbors per edit, there are at most $\targetdegeneracy$ new marked edges that could be resolved by editing $x$.
    Thus, $x$ resolves at most the same number of edges as it did originally.

    Finally, we consider the case where $v$ forms an inversion with a back neighbor $b$.
    In order for $b$ to be able to form another inversion with a different marked neighbor, it must have forward degree $\targetdegeneracy - 1$.
    However, it must first form an inversion with $v$ again which has already been considered.
\end{proof}

\begin{theorem} \label{theorem:degen-by-one}
    There exists an $O(\log n)$-approximation for finding the minimum size edit set to reduce the degeneracy of a graph from $\targetdegeneracy$ to $\targetdegeneracy - 1$.
\end{theorem}

\begin{proof}
    Let $m_0$ be the number of marked edges in the first min-degree ordering.
    Since the optimal edit set (of size $k$) resolves every marked edge in $k$ steps, the $k$ edits must on average resolve $\frac{m_0}{k}$ marked edges.
    Thus, the largest number of marked edges resolved by an edit from the optimal edit set must be at least $\frac{m_0}{k}$.
    Fix a sequence of the optimal edit set.
    By Lemma~\ref{lemma:resolve-noninc}, the number of marked edges that this edit resolves at the current step must be at least as large as when it appears in the optimal sequence.
    Thus, there must exist an edit which resolves at least $\frac{m_0}{k}$ marked edges, and so the edit that resolves the most marked edges must resolve at least $\frac{m_0}{k}$ as well.
    After one iteration then, there are $m_1 \leq m_0(1 - \frac{1}{k})$ marked edges remaining.
    Since the edited graph is a subgraph of the original, the optimal sized edit set must still be of size at most $k$ and so the same analysis applies.
    Thus, after $t$ steps, there are at most $m_0(1 - \frac{1}{k})^t$ marked edges remaining.
    If we set $t = k \ln m_0$, then there at most $m_0(1 - \frac{1}{k})^{k \ln m_0} \leq m_0 \cdot \frac{1}{e}^{\ln m_0} = 1$ marked edges remaining.
    Thus, we need at most $k \ln m_0 + 1$ iterations to resolve every marked edge.
    Since we add one edit at each iteration, we produce an edit set of size at most $k \ln m_0 + 1$.
    Because $m_0$ is at most $n^2$, the size of the edit set is $k \ln m_0 + 1 \leq k \ln n^2 + 1 = O(k  \log n)$.
\end{proof}

\begin{corollary} \label{corollary:degen-by-d}
    There exists an $O(d \cdot \log n)$-approximation for finding the minimum size edit set to reduce the degeneracy of a graph from $\targetdegeneracy$ to $\targetdegeneracy - d$.
\end{corollary}

\begin{proof}
    Apply the above algorithm $d$ times.
    Let $G_i$ be the edited graph after $i$ applications.
    Note that $G_i$ has degeneracy $\targetdegeneracy - i$.
    The optimally sized edit set $\opt$ to reduce the degeneracy of $G_i$ from $\targetdegeneracy - i$ to $\targetdegeneracy - i - 1$ is at most the size of the smallest edit set to reduce the degeneracy of $G_0$ from $\targetdegeneracy$ to $\targetdegeneracy - d$.
    Thus, at each iteration, we add at most $O(|\opt|\log n)$ edits to our edit set.
    After $d$ iterations then, we have a graph with degeneracy $\targetdegeneracy - d$ and an edit set of size $O(d|\opt|\log n)$.
\end{proof}

\begin{corollary} \label{corollary:degen-to-r}
    There exists an $O(r \cdot \log n)$-approximation for finding the minimum size edit set to reduce the degeneracy of a graph to $\targetdegeneracy$.
\end{corollary}

\begin{proof}
    Apply an algorithm from Section~\ref{section:positive_degeneracy_localratio} or~\ref{section:positive_degeneracy_lp}.
    This yields a graph with degeneracy $O(\targetdegeneracy)$ and an edit set of size $O(|\opt|)$.
    Apply the algorithm from Corollary~\ref{corollary:degen-by-d} to reduce the degeneracy by the remaining $O(\targetdegeneracy)$ steps.
    The final size of the edit set is $O(|\opt|) + O(\targetdegeneracy|\opt|\log n) = O(\targetdegeneracy|\opt|\log n)$.
\end{proof}

\fi

\ifdefined\istwappx
    \subsection{Treewidth/Pathwidth: Bicriteria Approximation for Vertex Editing}\label{section:positive_treewidth}
    In this section, we design a polynomial-time algorithm that constructs a $(O(\log^{1.5} n), O(\sqrt{\log w}))$-bicriteria approximate solution to \bTWVfull: the size of the edit set is at most $O(\log^{1.5} n)$ times the optimum ($\optsol{\bTWV}{G}$) and the resulting subgraph has treewidth $O(\Tw\sqrt{\log \Tw})$. We also give a $(O(\log^{1.5} n), O(\sqrt{\log w}\cdot \log n))$-bicriteria approximation for editing to
pathwidth $w$.

Our approach relies on known results for {\em vertex $c$-separators},  structures which are used extensively in many other algorithms for finding an approximate tree decomposition.

\begin{definition}\label{def:node-separator}
For a subset of vertices $W$, a set of vertices $S\subseteq V(G)$ is a \emph{vertex $c$-separator} of $W$ in $G$ if each component of $G[V \setminus S]$ contains at most $c|W|$ vertices of $W$.
The minimum size vertex $c$-separator of a graph, denoted $\sep_c(G)$, is the minimum integer $k$ such that for any subset $W \subseteq V$ there exists a vertex $c$-separator of $W$ in $G$ of size $k$.
\end{definition}
The size of a minimum size vertex $c$-separator of a graph is a parameter of interest and has applications in bounding treewidth and finding an approximate tree decomposition. Our algorithms in this section use vertex $\left({3\over 4}\right)$-separators.
\begin{lemma}[{\cite[Section 6.2]{feige2008improved}}]\label{lem:approx-node-sep}
There exist polynomial time algorithms that find a vertex $\left({3\over 4}\right)$-separator of a graph $G$ of size $c_1\cdot\sep_{2/3}(G)\sqrt{\log \sep_{2/3}(G)}$, for a sufficiently large constant $c_1$.
\end{lemma}

The following bounds relating the treewidth of $G$, $\tw(G)$, and the minimum size vertex $({2\over 3})$-separator of $G$, $\sep_{2/3}(G)$, are useful in the analysis of our proposed algorithm for the problem of editing to treewidth \Tw.
\begin{lemma}[{\cite[Lemma 7]{robertson1986graph, robertson1995graph, harvey2017parameters}}]\label{lem:tw-sep}
For any graph $G$, $\sep_{2/3}(G) \leq \tw(G) + 1 \leq 4 \sep_{2/3}(G)$.
\end{lemma}

\begin{lemma}\label{lem:edit-lower-bound}
For any graph $G=(V,E)$, and integer
$w\leq {3\over 4}\cdot \tw(G)$, $\sep_{2/3}(G)\leq 6\cdot\optsol{\bTWV}{G}$.
\end{lemma}
\begin{proof}
It is straightforward to verify that for any $X \subseteq V$, if $\tw(G[V\setminus X]) = w$, then
$|X| \geq \tw(G)-w$. Suppose not. Then, we can add $X$ to all bags in the tree decomposition of $G[V\setminus X]$ and the resulting tree decomposition of $G$ has treewidth less than $\tw(G)$ which is a contradiction.

If we assume $\Tw \leq \frac{3}{4} \cdot \tw(G)$, then $\optsol{\bTWV}{G} \geq \tw(G) - w \geq \tw(G)/4$, which, together with Lemma~\ref{lem:tw-sep}, implies that $\sep_{2/3}(G) \leq {3\over 2}\cdot \tw(G) \leq 6\cdot\optsol{\bTWV}{G}$.
\end{proof}

\subsubsection{Treewidth: $(O(\log^{1.5} n), O(\sqrt{\log w}))$-approximation for vertex deletion}

Our method exploits the general recursive approach of the approximation algorithms for constructing a tree decomposition~\cite{amir2010approximation, bodlaender1995approximating, feige2008improved, reed1992finding}. Our algorithm
iteratively subdivides the graph, considering $G[V_i]$ in iteration $i$. We first apply
the result of~\cite{bodlaender1995approximating,feige2008improved} to determine if $G[V_i]$ has a tree decomposition with ``small'' width; if
yes, the algorithm removes nothing and terminates.
Otherwise, we compute an approximate vertex $(3/4)$-separator $S$ of $G[V_i]$ (applying the algorithm of~\cite{feige2008improved}),
remove it from the graph, and recurse on the connected components of $G[V_i\setminus S]$.
We show that the total number of vertices removed from $G$ in our algorithm is not more than $O(\optsol{\bTWV}{G}\log^{1.5}n)$ in Theorem~\ref{thm:edit-size} and the treewidth of the resulting graph is $O(w\sqrt{\log w})$ in Theorem~\ref{thm:final-width}.

\begin{algorithm}
\caption{Approximation for Vertex Editing to Bounded Treewidth Graphs}\label{alg:treewidth-node-edit}
\begin{algorithmic}[1]
\Procedure{TreeWidthNodeEdit}{$G=(V, E)$, $w$}
\State $t \leftarrow$ compute $\tw(G)$ by invoking the algorithm of~\cite{bodlaender1995approximating} together with~\cite{feige2008improved} \label{alg:compute-tw}
\If{$t \leq 32c_1 \cdot w\sqrt{\log w}$} \label{if:cond-tw}
	\State \Return $\emptyset$ \label{alg:small-width}
\Else
	\State $S \leftarrow$ compute a vertex $({3\over 4})$-separator of $G$ by invoking the algorithm of~\cite{feige2008improved}\label{alg:sep}
	\State {\bf let} $G[V_1], \cdots, G[V_\ell]$ be the connected components of $G[V\setminus S]$.
    \State \Return $\big(\bigcup_{i\leq \ell} \Call{TreeWidthNodeEdit}{G[V_i], w}\big) \cup S$
\EndIf
\EndProcedure
\end{algorithmic}
\end{algorithm}

The key observation in our approach is the following lemma:

\begin{lemma}\label{lem:disjoint-sep-tw}
Suppose that the vertex set of $G=(V,E)$ is partitioned into $V_1, \cdots, V_\ell$.
The minimum edit set of $G$ to treewidth $w$ has size at least $\sum_{i\leq \ell} \max(0,\tw(G[V_i]) - w)$.
\end{lemma}

\begin{proof}
This directly follows from the straightforward observation that if $\tw(G[V \setminus X]) = w$, then $|X| \geq \tw(G) - \Tw$. Since the sets of vertices are disjoint, then the lower bound on the number of vertices that must be deleted is the summation of the lower bound of the number of vertices that must be deleted in each disjoint set.
\end{proof}

In Algorithm~\ref{alg:treewidth-node-edit}, we use the approach of Bodlaender~\etal~\cite{bodlaender1995approximating} together with the $O(\sqrt{\log \tw(G)})$-approximation algorithm of~\cite{feige2008improved} for computing treewidth of $G$ in Line~\ref{alg:compute-tw}.



\begin{theorem}[\cite{bodlaender1995approximating,feige2008improved}]\label{thm:tree decomp-alg}
There exists an algorithm that, given an input graph $G$, in polynomial time returns a tree decomposition of $G$ of width at most $c_2\cdot \tw(G)\sqrt{\log \tw(G)}$ and height $O\left(\log |V(G)|\right)$ for a sufficiently large constant $c_2$.
\end{theorem}
For the sake of completeness, we provide the proof of Theorem~\ref{thm:tree decomp-alg} in Section~\ref{sec:proof-tree-decomposition}.
Next, we analyze the performance of Algorithm~\ref{alg:treewidth-node-edit}.

\begin{theorem}\label{thm:edit-size}
Algorithm~\ref{alg:treewidth-node-edit} removes at most $O(\log^{1.5} n)\optsol{\bTWV}{G}$ vertices from
any $n$-vertex graph $G$.
\end{theorem}
\begin{proof}
The proof is by induction on the number of vertices in the given induced subgraph of $G$: we show that for any subset $V' \subseteq V$, the number of vertices removed by Algorithm~\ref{alg:treewidth-node-edit} is at most $(c \cdot \log_{4/3} |V'| \cdot \sqrt{\log n})\optsol{\bTWV}{G[V']}$ where $c\geq 6$ is a fixed constant. We remark that $c$ must also be greater than the constant $c_2$ in Theorem~\ref{thm:tree decomp-alg} for the width guarantee.

By the condition in Line~\ref{if:cond-tw} of the algorithm, the claim trivially holds for the case $|V| = O(w\sqrt{\log w})$.
Assume that the claim holds for all induced subgraphs of $G$ containing at most $n'-1$ vertices.
Next, we show that the claim holds for any $n'$-vertex induced subgraph of $G$, $G[V']$, too.
If $t \leq 32 c_1 \cdot w\sqrt{\log w}$, then no vertex is deleted and the claim holds.
Otherwise,
\begin{align*}
32c_1 \cdot w\sqrt{\log w}< t
&\leq c_1\cdot\sep_{2/3}(G[V'])\sqrt{\log \sep_{2/3}(G[V'])}  & \text{by Lemma~\ref{lem:approx-node-sep},}\\
&\leq c_1\cdot(2\tw(G[V']))\sqrt{2\log \tw(G[V'])}  & \text{by Lemma~\ref{lem:tw-sep}},
\end{align*}
which implies that $\tw(G[V']) \geq 4w$.
By Lemma~\ref{lem:edit-lower-bound}, $\sep_{2/3}(G[V']) \leq 6\optsol{\bTWV}{G[V']}$ which implies that $|S'| \leq c_1 \sqrt{\log \sep_{2/3}(G[V'])} \cdot \sep_{2/3}(G[V']) \leq 6c_1 \sqrt{\log n}\cdot \optsol{\bTWV}{G[V']}$ where $S'$ is a vertex $({3 \over 4})$-separator of $V'$ computed in Line~\ref{alg:sep} of the algorithm.

Let $V'_1,\cdots, V'_{\ell}$ be the disjoint components in $G[V'\setminus S']$.
Then, $\optsol{\bTWV}{G[V']} \geq \sum_{i\leq \ell} \optsol{\bTWV}{G[V'_i]}$, by Lemma~\ref{lem:disjoint-sep-tw}.
Further, by the induction assumption for each $i\leq \ell$, the number of vertices removed by
$\Call{TreeWidthNodeEdit}{G[V'_i], w}$ is at most $(c\log_{4/3} |V'_i| \cdot \sqrt{\log n})\optsol{\bTWV}{G[V'_i]}$.
Hence, the vertices removed by $\Call{TreeWidthNodeEdit}{G[V'],w}$, satisfy
\begin{align*}
|X| &\leq 6c_1\sqrt{\log n} \cdot \optsol{\bTWV}{G[V']} + \sum_{i\leq \ell} (c\log_{4/3} |V'_i| \cdot \sqrt{\log n})\cdot \optsol{\bTWV}{G[V'_i]}  \\
&\leq 6c_1 \sqrt{\log n} \cdot \optsol{\bTWV}{G[V']} + (c\log_{4/3} {3|V'|\over 4} \cdot \sqrt{\log n})\sum_{i\leq \ell} \optsol{\bTWV}{G[V'_i]} \\
&\leq c\sqrt{\log n} \cdot (1 + \log_{4/3} |V'| -1) \optsol{\bTWV}{G[V']} & \text{with } c>6c_1. \\
&\leq (c\log_{4/3} |V'| \cdot \sqrt{\log n}) \optsol{\bTWV}{G[V']}.
\end{align*}
\end{proof}
\begin{theorem}\label{thm:final-width}
The treewidth of the subgraph of $G$ returned by Algorithm~\ref{alg:treewidth-node-edit} is $O(w\cdot \sqrt{\log w})$.
\end{theorem}
\begin{proof}
This follows immediately from the condition in Line~\ref{if:cond-tw} of the algorithm.
\end{proof}

\subsubsection{Pathwidth: $(O(\log^{1.5} n), O(\sqrt{\log w}\cdot \log n))$-approximation for vertex deletion}\label{sec:pathwidth}
Our algorithm in this section builds on the reduction of Bodleander \etal~\cite{bodlaender1995approximating} from tree decomposition to path decomposition and Algorithm~\ref{alg:treewidth-node-edit} described in Section~\ref{section:positive_treewidth} for finding a minimum size edit set to treewidth $w$. The main component of the reduction approach of~\cite{bodlaender1995approximating} is the following.
\begin{lemma}[\cite{bodlaender1995approximating}]\label{lem:tree-to-path}
Given a tree decomposition of $G$ with width at most $\Tw$ and height at most $h$, we can find a path decomposition of $G$ with width at most $w\cdot h$ efficiently.
\end{lemma}

\begin{corollary}
Given an input graph $G=(V,E)$ and a target pathwidth $w$, Algorithm~\ref{alg:treewidth-node-edit} removes $O(\log^{1.5} n)\cdot \optsol{\bPWV}{G}$ vertices $\editset$ such that $\pw(G[V\setminus \editset]) \leq (\sqrt{\log w} \cdot \log n) \cdot w$.
\end{corollary}
\begin{proof}
By Theorem~\ref{thm:edit-size}, $|\editset| \leq (\log^{1.5} n) \cdot \optsol{\bTWV}{G}$. Since $\tw(G) \leq \pw(G)$ for all $G$, we have $\optsol{\bTWV}{G} \leq \optsol{\bPWV}{G}$. Hence, $|\editset| \leq \log^{1.5} n \cdot \optsol{\bPWV}{G}$.
Further, by Lemma~\ref{lem:tree-to-path} and Theorem~\ref{thm:tree decomp-alg}, $\pw(G[V\setminus X]) \leq (w\sqrt{\log w})\cdot \log n$.
\end{proof}

\stoptocwriting
\subsubsection{Proof of Theorem~\ref{thm:tree decomp-alg}}\label{sec:proof-tree-decomposition}
\resumetocwriting
In this section, we provide the proof of Theorem~\ref{thm:tree decomp-alg} which is essentially via the tree decomposition of Bodleander \etal~\cite{bodlaender1995approximating} by plugging in the $O(\sqrt{\log \opt})$-approximation algorithm of~\cite{feige2008improved} for vertex separators. Algorithm~\ref{alg:tree decomp} is the recursive approach of~\cite{bodlaender1995approximating} for approximating treewidth (and constructing its tree decomposition).

\begin{algorithm}
\caption{Approximation Algorithm for Tree Decomposition (From~\cite{bodlaender1995approximating,feige2008improved})}\label{alg:tree decomp}
\begin{algorithmic}[1]
\Procedure{TreeDecomposition}{$G, Z, W$} \Comment{$Z\cap W = \emptyset$, output contains $W$ in root bag}
\If{$8|Z| \leq |W|$}
	\State \Return a tree decomposition with a single node containing $Z\cap W$ \label{alg:base-case}
\Else
	\State $S \leftarrow$ a vertex $({3\over 4})$-separator of $W$ in $G[Z\cup W]$ by invoking the algorithm of~\cite{feige2008improved}\label{alg:sep-w}
	\State $T \leftarrow$ a vertex $({3\over 4})$-separator of $Z\cup W$ in $G[Z\cup W]$ by invoking the algorithm of~\cite{feige2008improved}\label{alg:sep-all}
	\State {\bf let} $G[V_1], \cdots, G[V_\ell]$ be the connected components of $G[(W\cup Z)\setminus (S\cup T)]$.
\EndIf
\For{$i=1$ to $\ell$}
	\State $Z_i \leftarrow Z \cap V_i$
	\State $W_i \leftarrow W \cap V_i$
	\State $T_i \leftarrow \Call{TreeDecomposition}{G, Z_i, W_i\cup S \cup T}$
\EndFor
\State \Return the tree decomposition with $(W \cup S\cup T)$ as its root and $T_1, \cdots, T_\ell$ as its children
\EndProcedure
\end{algorithmic}
\end{algorithm}

\begin{claim}\label{clm:w-size}
	If $W$ and $Z$ are disjoint sets of vertices of $G$ and $|W| \leq 32c_1\cdot \tw(G)\sqrt{\log\tw(G)}$, then the solution produced by the algorithm is a tree decomposition of $G[W\cup Z]$ of width at most $36c_1\cdot \tw(G)\sqrt{\log\tw(G)}$.
\end{claim}
\begin{proof}
First we show that the output is a valid tree decomposition of $G[W\cup Z]$.
\begin{itemize}
	\item{\bf All edges of $G[Z\cup W]$ are covered in $T$.} The proof is by an induction on the recursive structure of the algorithm. For a leaf bag in the tree decomposition ($|Z| \leq |W|/8$), a single bag contains all vertices of $Z \cup W$; hence, the claim holds. Now, suppose that this property holds for all subtrees rooted at children of the tree decomposition constructed by \Call{TreeDecomposition}{$G, Z, W$}.
	Consider an edge $uv\in E$. If $u,v\in W$, then $u$ and $v$ are both contained in the root bag and it is covered in the tree decomposition. Otherwise, $v$ and $u$ both belong to $V_i\cup S\cup T$ for an $i\in[\ell]$ and by the induction hypothesis, $uv$ is covered in the subtree $T_i$ corresponding to \Call{TreeDecomposition}{$G, Z_i, W_i\cup S \cup T$}. Thus, the property holds for $T$ as well.
	\item{\bf Bags containing each vertex are connected in tree structure $T$.} The proof is by induction on the recursive structure of \Call{TreeDecomposition}{$\cdot$}.
	More precisely, we show that for each pair $(Z', W')$, the bags containing a vertex $v \in Z' \cup W'$ are connected in $\Call{TreeDecomposition}{G, Z', W'}$. The property trivially holds for the leaves of $T$ (the case $|Z'| \leq |W'|/8$). Suppose that this property holds for all subtrees rooted at children of the tree decomposition $T = $ \Call{TreeDecomposition}{$G, Z, W$}.
	Then, we show that the property holds for $T$ as well. If $v\in W\cup S\cup T$, then by the induction hypothesis on children of $T$, the property holds for $T$ as well. Otherwise, $v$ is only contained in one of the children of $T$ (i.e., $v\in Z_i \cup W_i$) and by the induction hypothesis the property holds.
\end{itemize}

Next, we show that the width of the tree decomposition constructed by \Call{TreeDecomposition}{$G, Z, W$} is at most $36c_1\cdot \tw(G)\sqrt{\log\tw(G)}$. By induction on the size of $Z$, it suffices to show that $|W\cup S \cup T| \leq 36c_1\cdot \tw(G)\sqrt{\log\tw(G)}$ and for each $i$, $|W_i\cup S \cup T| \leq 32c_1\cdot \tw(G)\sqrt{\log\tw(G)}$.
Note that, by Line~\ref{alg:base-case}, if $|Z|\leq |W|/8$, then the returned tree decomposition has width at most $36c_1\cdot \tw(G)\sqrt{\log\tw(G)}$.
Suppose that the claim holds for all $(Z', W')$ where $|Z'| < |Z|$. We bound the size of $S$ and $T$ as follows:
\begin{align*}
	|S|, |T|
	&\leq c_1 \sep_{2/3}(G)\sqrt{\log \sep_{2/3}(G)} & \text{by Lemma~\ref{lem:approx-node-sep},}\\
	&< {4c_1} \cdot \tw(G)\sqrt{\log \tw(G)} & \text{by Lemma~\ref{lem:tw-sep}.}
\end{align*}
Since $|W| \leq 32c_1\cdot \tw(G)\sqrt{\log\tw(G)}$, the root bag has size at most $(32c_1 + 2\cdot 4c_1) \cdot \tw(G)\sqrt{\log \tw(G)}$. Moreover, since $S$ and $T$ are respectively a $({3\over 4})$-vertex separator of $W$ and $W\cup Z$ in $G[W\cup Z]$, for each $i\in[\ell]$,
\begin{align*}
	|Z_i| &\leq {3\over 4} |Z| < |Z|, \text{ and} \\
	|W_i \cup S \cup T| &\leq \left({3\over 4} 32c_1+8c_1\right)\cdot \tw(G)\sqrt{\log \tw(G)} \leq 32c_1 \cdot \tw(G)\sqrt{\log \tw(G)}.
\end{align*}
Thus, it follows from the induction hypothesis that the width of each subtree $T_i$ is at most $36c_1 \cdot \tw(G)\sqrt{\log \tw(G)}$.
\end{proof}

\begin{proof}[Proof of Theorem~\ref{thm:tree decomp-alg}]
It follows from Claim~\ref{clm:w-size} that \Call{TreeDecomposition}{$G=(V,E), \emptyset, V$} constructs a tree decomposition of $G$ of width at most $36c_1\cdot \tw(G)\sqrt{\log \tw(G)}$. Moreover, since for each $i\in[\ell]$, $|Z_i| \leq 3|Z|/4$, the returned tree decomposition has height $O(\log |V(G)|)$.
\end{proof}

\fi

\ifdefined\isbde
	\subsection{Bounded Degree: Polynomial Time Algorithm for Edge Editing}
	\label{section:bounded_degree_edge_edit}
 	We show a polynomial time algorithm for \bBDDEfull by a reduction to the problem of finding a minimum size $f$-edge cover of the graph.

\begin{definition}
\label{def:f-edge-cover}
Given a function $f: V \rightarrow \mathbb{Z}^+$, an \emph{$f$-edge cover} of a graph $G = (V, E)$ is a subset $F \subseteq E$ such that $\deg_{G[F]}(v) \geq f(v)$ for all $v \in V$.
\end{definition}

We use the following result to show a polynomial time algorithm for \bBDDE.

\begin{theorem}[\cite{gabow83efficient,huang17approximate}]\label{thm:editing-bounded-degree}
Given a graph $G=(V, E)$, and function $f:f: V \rightarrow \mathbb{Z}^+$, there exists a polynomial time algorithm for finding a minimum size $f$-edge cover of $G$.
\end{theorem}
\begin{corollary}
There exists a polynomial time algorithm for $\bBDDE(G)$.
\end{corollary}
\begin{proof}
It follows by an application of Theorem~\ref{thm:editing-bounded-degree} to ($G, f$) where for each vertex $v$, $f(v) := \deg(v) - \Td$.
\end{proof}

\fi

\ifdefined\isstarforest
    \subsection{Treedepth 2: $O(1)$-approximation for Editing}
    \label{section:positive_starforest}
    Star forests are exactly the family of graphs with treedepth equal to two.
Thus, algorithms for editing to star forests may serve as a foundation for editing to larger values of treedepth. We refer to the vertex with degree greater than one in a star as its \emph{center} (if a star is a single edge, arbitrarily pick a vertex to be the center). All other vertices are referred to as \emph{leaves}.

We give polynomial-time $O(1)$-approximations for \SFVfull and \SFEfull based on \textsc{Hitting Set}.

\begin{lemma}
There exists a polynomial-time 4-approximation for \SFVfull and a polynomial-time 3-approximation for \SFEfull.
\end{lemma}

\begin{proof}
It is straightforward to see that a graph is a star forest if and only if it does not contain $P_4$ or $C_3$ as a subgraph. Therefore for a given instance $G=(V,E), \editnumber$ of \SFV, we must delete at least one vertex from every $P_4$ and $C_3$ in $G$. We can enumerate all size-four and size-three sets of nodes in $O(n^4)$ time, and for each we can check if it is a $P_4$ or $C_3$ in constant time. If we let $S$ be the set of subsets of $V$ which are $P_4$'s or $C_3$'s then we have an instance of \textsc{4-Hitting Set}, as each set in $S$ must have a node from it deleted and the cardinality of each set is at most four. Finding a hitting set equates to finding an edit set \editset~ in our \SFV instance. Hochbaum~\cite{hochbaum1982approximation} showed that there exists a polynomial-time $k$-approximation for $k$-\textsc{Hitting Set}, and thus there exists a polynomial-time $4$-approximation for \SFVfull.

Our approach is similar for \SFE, but here we note there could be multiple $P_4$'s on a set of four nodes and deleting a single edge may not remove all of them (whereas deleting a single node would). Thus we must distinguish each edge-unique $P_4$ in $S$. We check if each of the  permutations of nodes is a $P_4$, and if it is we include the corresponding edges as a set in $S$ (note that checking for only a permutation or its reverse is sufficient). We again use the result from~\cite{hochbaum1982approximation}, and as the cardinality of each set in $S$ is at most three, we obtain a polynomial-time $3$-approximation for \SFE.
\end{proof}

We point out that, because there exists a finite list of forbidden subgraphs that characterizes graphs of treedepth bounded by $\Tp$~\cite{nesetril2012sparsity}, the above hitting-set--based algorithm generalizes for the problem of deleting to any fixed treedepth \Tp.
Such an approximation may take time exponential in the size of the largest forbidden subgraph, but this size is a constant for fixed \Tp, and so the approximation is still poly-time in $|V(G)|$.

\fi

  \section{Editing Hardness Results}
\label{section:negative_results}

\ifdefined\isBDeEhard
    \subsection{Degeneracy: $o(\log (n/\Tr) )$-Inapproximability of Vertex and Edge Editing}
    \label{section:negative_bounded_degeneracy_editing}
    In this section, we prove the following inapproximability results for \bDEV and \bDEE.
\begin{theorem}\label{theorem:hardness-node-degeneracy}
		For graphs $G$ with $n$ nodes and $\Tr$ satisfying $\degener(G) > \Tr \geq 2$, \bDEEfull and \bDEVfull are $o(\log (n/\Tr))$-inapproximable.
\end{theorem}

\mypar{Reduction strategy.}
Our proof relies on a strict reduction from \SCfull.
We prove that the graph constructed has degeneracy $\Tr+1$, and that reducing its degeneracy to $\Tr$ requires deleting edges (vertices) from set gadgets that correspond to a solution for $\SC(\mcU,\mcF)$.
We proceed by introducing the gadgets used in this reduction.

\mypar{Mapping an instance of \SCfull to an instance of \bDEfull.}
Given an instance $\SC(\mcU,\mcF)$, we want to define a function $f$ that maps this instance to a graph $G = f(\mcU,\mcF)$ as an instance of $\bDE$.
This reduction relies on three gadgets which we describe here:
\emph{set gadgets} and \emph{element gadgets} represent the sets and elements of the set cover instance, while \emph{split gadgets} enable us to connect the set and element gadgets to encode the containment relationships of $\SC(\mcU,\mcF)$.
\begin{definition}\label{def:sc-de-split-gadget}
	Given a fixed value $\Tr \geq 2$, the corresponding \emph{split gadget}, $\splitg$, consists of $\Tr+2$ vertices: $\Tr$ vertices connected in a clique, and two vertices, labeled $t$ and $b$ (for ``top'' and ``bottom'') each connected to all $\Tr$ clique vertices.
	See Figure~\ref{fig:de-split-gadget} for a visualization.
\end{definition}
\begin{definition}\label{def:sc-de-set-gadget}
  Given a fixed value $\Tr \geq 2$, an instance $\SC(\mcU,\mcF)$, and a set $S \in \mcF$, we define the \emph{set gadget} $\gDS_S$ as follows.
	The gadget consists of a length-$\Tr$ path $x_1, \cdots, x_{\Tr}$;
	an independent $W$ set of $\Tr$ nodes, $w_1, \cdots, w_{\Tr}$ each of which is connected to every vertex $x_j$;
	a \emph{set vertex} $v_S$ connected to $W$;
	and two instances of \splitg, each of which is connected to $v_S$ by a single edge from its $b$ vertex.
	See Figure~\ref{fig:de-set-gadget} for a visualization.
\end{definition}
\begin{definition}\label{def:sc-de-element-gadget}
	Given a fixed value $\Tr \geq 2$, an instance $\SC(\mcU,\mcF)$, and an element $e \in \mcU$,
	we define the corresponding \emph{element gadget} $\sG_e$ as follows.
	Let $f_e$ be the number of sets in $\mcF$ that contain $e$.
	The gadget consists of $2 \Tr f_e$ total vertices:
	\begin{compactenum}
		\item $\Tr  f_e$ vertices $z_j^i$ connected in a cycle, where $j = 1, \cdots, \Tr$ and $i = 1, \cdots, f_e$;
		\item $(\Tr-1) f_e$ independent vertices $a_j^i$, where $j = 1, \cdots, \Tr-1$ and $i = 1, \cdots, f_e$; and
		\item $f_e$ independent vertices $v_e^1, \cdots, v_e^{f_e}$.
	\end{compactenum}
	For each $i$, the vertices $\{z_j^i\}_{j=1}^{\Tr}$ and the vertex $v_e^i$ are completely connected to the vertices $\{a_j^i\}_{j=1}^{\Tr-1}$.
	See Figure~\ref{fig:de-element-gadget} for a visualization.
\end{definition}
The motivation for these gadgets, as we will prove below, is that, once the gadgets are linked together properly, every vertex will have degree at least $\Tr+1$, and each set gadget $\DS_S$ remains in the $(\Tr+1)$-core until either a vertex in $\DS_S$ is deleted or
all element gadgets attached to it have been deleted.
Before proving these properties, we first describe our map $f$ from $\SC$ to $\bDE$.\looseness-1
\begin{definition}\label{def:sc-de-gadget}
		Given an instance $\SC(\mcU,\mcF)$, we construct an instance $G = f(\mcU,\mcF)$ of $\bDE$ as follows.
		For each element $e \in \mcU$, the graph $G$ contains one element gadget $\sG_e$ (Definition~\ref{def:sc-de-element-gadget}).
		For each set $S \in \mcF$, $G$ contains one set gadget $\DS_S$ (Definition~\ref{def:sc-de-set-gadget});
		for each element $e \in S$, an edge connects one of the vertices $v_e^i \in \sG_e$ to both of the
		vertices $t \in \splitg_1$ and $t \in \splitg_2$ in $\DS_S$.
		If a set $S$ contains more than $\Tr$ elements, then the instances $\splitg_j$ in $\DS_S$ are first attached to additional copies of \splitg, such that no instance of \splitg has its vertex $t$ attached to more than $\Tr$ edges.
		See Figure~\ref{fig:de-arrange-splits}.
\end{definition}

Note that the above definitions require $\Tr \geq 2$ and $f_e \geq 2$ in order for the cycle structure in the element gadget to make sense.
This requirement can easily be satisfied by pre-processing the instance $\SC(\mcU,\mcF)$ to remove trivial elements with $f_e=1$.

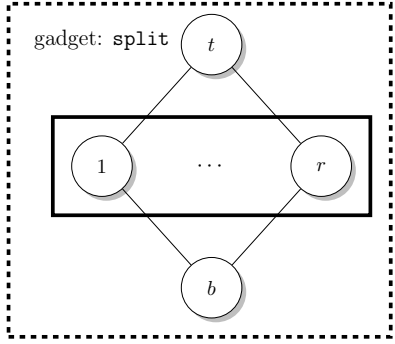
\begin{figure}[!h]
    \centering
	\scalebox{0.7}{
		\begin{tikzpicture}
\tikzstyle{node_style} = [state, fill=white, drop shadow, minimum size=1.15cm]
\tikzstyle{edge_style} = [line width=2pt]

\def\cliquewidth{3}
\def\gheight{1.5}

\node (v1) [node_style] {$1$};
\node (vr) [node_style, right = \cliquewidth cm of v1] {$r$};

\node (v_mid) at ($(v1)!0.5!(vr)$) {\large $\cdots$};

\node (t) [node_style, above = \gheight cm of v_mid] {$t$};
\node (b) [node_style, below = \gheight cm of v_mid] {$b$};

\node (box_clique) [fit=(v1)(vr), draw=black, line width=2pt,  inner sep=10pt] {};

\node (gadget_label) [above = \gheight cm of v1] {\large gadget: \large$\mathtt{split}$};
\node (box) [fit=(gadget_label)(vr)(box_clique)(b), draw=black, line width=2pt,  inner sep=10pt, dashed] {};

\draw (v1) edge (t);
\draw (v1) edge (b);
\draw (vr) edge (t);
\draw (vr) edge (b);

\end{tikzpicture}
	}
    \caption{The gadget \splitg, described in Definition~\ref{def:sc-de-split-gadget}, provides the mechanism by which we connect a set gadget, $\DS_S$, to each of the element gadgets $\sG_e$ for $e \in S$.
    }\label{fig:de-split-gadget}
\end{figure}
\begin{figure}[!h]
    \centering
	\begin{tikzpicture}
\tikzstyle{node_style} = [state, fill=white, drop shadow, minimum size=1.15cm]
\tikzstyle{small_node_style} = [state, fill=white, drop shadow]
\tikzstyle{tiny_node_style} = [state, fill=white, drop shadow, minimum size=0.3cm]
\tikzstyle{square_style} = [minimum width=1.3cm, minimum height=1.3cm,
                            fill=white, draw=black, drop shadow, line width=2pt]
\tikzstyle{edge_style} = [line width=2pt]
\tikzstyle{path_style} = [line width=2pt, decorate, decoration={snake}]
\tikzstyle{arrow_style} = [draw, fill=white, single arrow, single arrow head indent=1ex, minimum size=1cm, drop shadow]

\def\distwtox{0.7}
\def\distwtov{0.9}
\def\distvtosplit{0.3}
\def\ydistvtosplit{0.8}
\def\wtoelt{3*\distwtov + 3*\ydistvtosplit + 5.25}

\node (w1) [node_style] {$w_1$};
\node (wr) [node_style, right = 4.5cm of w1] {$w_r$};
\node (w2) [node_style] at ($(w1)!0.38!(wr)$) {$w_2$};
\node (w_mid) at ($(w2)!0.5!(wr)$) {$\cdots$};

\node (x1) [node_style, below = \distwtox cm of w1] {$x_1$};
\node (xr) [node_style, below = \distwtox cm of wr] {$x_r$};
\node (x2) [node_style] at ($(x1)!0.38!(xr)$) {$x_2$};
\node (x_mid) at ($(x2)!0.5!(xr)$) {$\cdots$};

\node (w_midpoint) at ($(w1)!0.5!(wr)$) {};

\node (vs) [node_style, above = \distwtov cm of w_midpoint] {$v_S$};


\node (split1b) [small_node_style, above left = \ydistvtosplit cm and \distvtosplit cm of vs] {$b$};
\node (split1t) [small_node_style, above = 0.1cm of split1b] {$t$};
\node (splt1_dummy) at ($(split1t)!0.5!(split1b)$) {};
\node (spltsm1) [tiny_node_style, left = 1.7cm of splt1_dummy] {};
\node (spltsmr) [tiny_node_style, left = 0.5cm of splt1_dummy] {};
\node (spltsmellips) at ($(spltsm1)!0.6!(spltsmr)$) {$\cdots$};
\node (split1gadget_label) [left = 0.1cm of split1t] {$\mathtt{split}_1$};
\node (box1) [fit=(split1b)(split1t)(split1gadget_label), draw=black, line width=2pt,  inner sep=10pt] {};

\node (split2b) [small_node_style, above right = \ydistvtosplit cm and \distvtosplit cm of vs] {$b$};
\node (split2t) [small_node_style, above = 0.1cm of split2b] {$t$};
\node (splt2_dummy) at ($(split2t)!0.5!(split2b)$) {};
\node (splt2sm1) [tiny_node_style, right = 1.7cm of splt2_dummy] {};
\node (splt2smr) [tiny_node_style, right = 0.5cm of splt2_dummy] {};
\node (splt2smellips) at ($(splt2sm1)!0.5!(splt2smr)$) {$\cdots$};
\node (split2gadget_label) [right = 0.1cm of split2t] {$\mathtt{split}_{2}$};
\node (box2) [fit=(split2b)(split2t)(split2gadget_label), draw=black, line width=2pt,  inner sep=10pt] {};

\node (gadget_label) [above right = 0.5cm and -1cm of w1] {\large gadget: \LARGE$\mathtt{D}_S$};
\node (dummy_corner) [below right = 0.5cm and 0.1cm of xr] {};
\node (box) [fit=(gadget_label)(box2)(xr), draw=black, line width=2pt,  inner sep=10pt, dashed] {};

\node (elt1) [node_style, above = \wtoelt cm of w1] {$v_{e_{1}}$};
\node (elt2) [node_style, above = \wtoelt cm of wr] {$v_{e_{|S|}}$};

\node (elt_ellipses) at ($(elt1)!0.5!(elt2)$) {\LARGE $\cdots$};

\node (elt1label) [above = 0.1cm of elt1] {\LARGE $\mathtt{G}_{e_1}$};
\node (elt2label) [above = 0.1cm of elt2] {\LARGE $\mathtt{G}_{e_{|S|}}$};

\node (box1) [fit=(elt1)(elt1label), draw=black, line width=2pt,  inner sep=10pt] {};
\node (box2) [fit=(elt2)(elt2label), draw=black, line width=2pt,  inner sep=10pt] {};

\draw (w1) edge (x1);
\draw (w1) edge (x2);
\draw (w1) edge (xr);
\draw (w2) edge (x1);
\draw (w2) edge (x2);
\draw (w2) edge (xr);
\draw (wr) edge (x1);
\draw (wr) edge (x2);
\draw (wr) edge (xr);

\draw (vs) edge (w1);
\draw (vs) edge (w2);
\draw (vs) edge (wr);

\draw (x1) edge (x2);
\draw (x2) edge ($(x2)!0.5!(x_mid)$);
\draw (xr) edge ($(xr)!0.5!(x_mid)$);

\draw (vs) edge (split1b);
\draw (vs) edge (split2b);

\draw (split1t) edge (elt1);
\draw (split1t) edge (elt2);
\draw (split2t) edge (elt1);
\draw (split2t) edge (elt2);

\foreach \verta in {split1t, split1b}
{
    \foreach \vertb in {spltsm1, spltsmr}
    {
        \draw (\verta) edge (\vertb);
    }
}
\foreach \verta in {split2t, split2b}
{
    \foreach \vertb in {splt2sm1, splt2smr}
    {
        \draw (\verta) edge (\vertb);
    }
}

\end{tikzpicture}
    \caption{$\gDS_S$ is the gadget corresponding to set $S \in \mcF$, described in Definition~\ref{def:sc-de-set-gadget}.
    }\label{fig:de-set-gadget}
\end{figure}
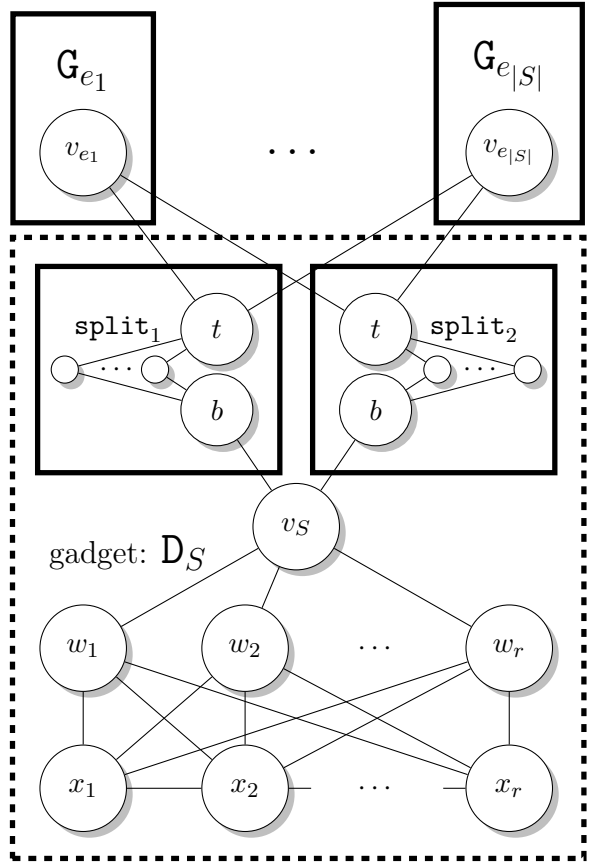
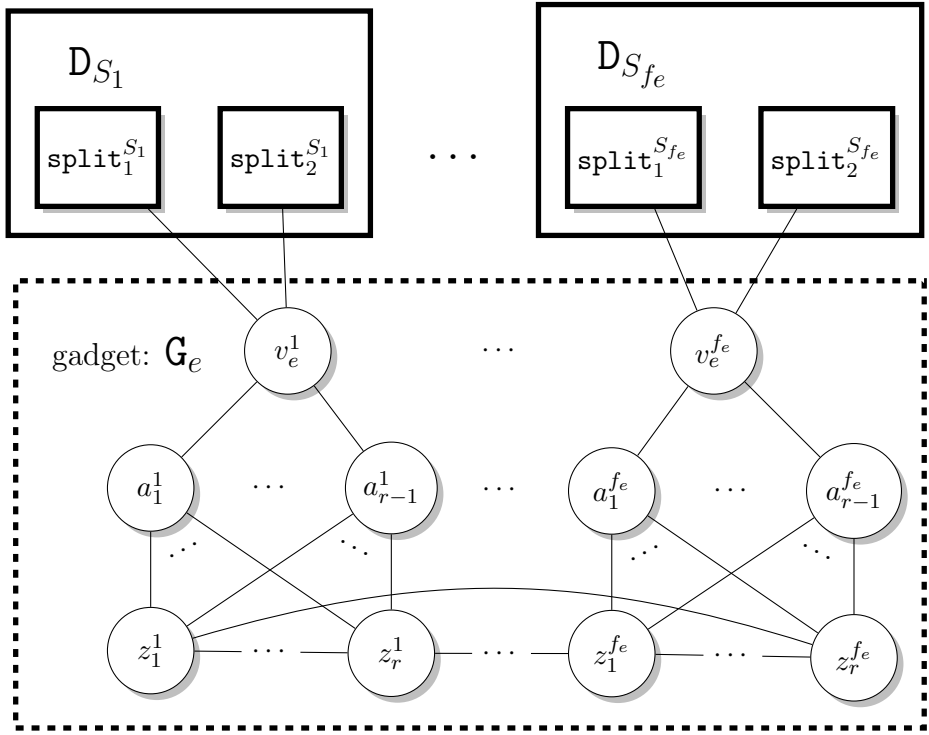
\begin{figure}[!h]
    \centering
  	\begin{tikzpicture}
\tikzstyle{node_style} = [state, fill=white, drop shadow, minimum size=1.15cm]
\tikzstyle{square_style} = [minimum width=1.3cm, minimum height=1.3cm,
                            fill=white, draw=black, drop shadow, line width=2pt]
\tikzstyle{edge_style} = [line width=2pt]
\tikzstyle{path_style} = [line width=2pt, decorate, decoration={snake}]
\tikzstyle{arrow_style} = [draw, fill=white, single arrow, single arrow head indent=1ex, minimum size=1cm, drop shadow]

\def\labelshiftX{-0.6}
\def\labelshiftY{-0.25}
\def\vxspace{1}

\node (v1) [node_style] {$v_{e}^1$};
\node (vfe) [node_style, right = 4.5cm of v1] {$v_e^{f_e}$};

\node (a1) [node_style, below left = \vxspace cm and 1cm of v1] {$a^{1}_{1}$};
\node (a2) [node_style, right = 2cm of a1] {$a^{1}_{r-1}$};
\node (a4) [node_style, below right = \vxspace cm and 1cm of vfe] {$a^{f_e}_{r-1}$};
\node (a3) [node_style, left = 2cm of a4] {$a^{f_e}_{1}$};

\node (z1) [node_style, below = 1cm of a1] {$z^1_1$};
\node (z2) [node_style, below = 1cm of a2] {$z^1_{r}$};
\node (z3) [node_style, below = 1cm of a3] {$z^{f_e}_1$};
\node (z4) [node_style, below = 1cm of a4] {$z^{f_e}_{r}$};

\def\distvtosplit{3.3}
\node (split11) [square_style, above left = \distvtosplit cm and -0.55cm of a1] {$\mathtt{split}^{S_{1}}_1$};
\node (split12) [square_style, right = 0.75 cm of split11] {$\mathtt{split}^{S_1}_2$};

\node (split22) [square_style, above left = \distvtosplit cm and -1cm of a4] {$\mathtt{split}^{S_{f_e}}_2$};
\node (split21) [square_style, left = 0.75 cm of split22] {$\mathtt{split}^{S_{f_e}}_1$};

\node (v_ellipses) at ($(v1)!0.5!(vfe)$) {$\cdots$};

\node (center_a_ellipses) at ($(a2)!0.5!(a3)$) {$\cdots$};
\node (left_a_ellipses) at ($(a1)!0.5!(a2)$) {$\cdots$};
\node (right_a_ellipses) at ($(a3)!0.5!(a4)$) {$\cdots$};

\node (center_z_ellipses) at ($(z2)!0.5!(z3)$) {$\cdots$};
\node (left_z_ellipses) at ($(z1)!0.5!(z2)$) {$\cdots$};
\node (right_z_ellipses) at ($(z3)!0.5!(z4)$) {$\cdots$};

\node (dummy_left) at ($(a1)!0.5!(z1)$) {};
\node (dummy_center_left) at ($(a2)!0.5!(z2)$) {};
\node (dummy_center_right) at ($(a3)!0.5!(z3)$) {};
\node (dummy_right) at ($(a4)!0.5!(z4)$) {};

\def\azleftoffset{0.17}
\def\azrightoffset{1.00-\azleftoffset}
\node(left_edge_ellipses) at ($(dummy_left)!\azleftoffset!(dummy_center_left)$) {};
\node(center_left_edge_ellipses) at ($(dummy_left)!\azrightoffset!(dummy_center_left)$) {};
\node(right_edge_ellipses) at ($(dummy_right)!\azleftoffset!(dummy_center_right)$) {};
\node(center_right_edge_ellipses) at ($(dummy_right)!\azrightoffset!(dummy_center_right)$) {};

\def\azellipsisdist{1.05}
\def\azellipsisangle{28}
\node(lefta1z1_ellipses) at ($(a1)!\azellipsisdist!(left_edge_ellipses)$) [label={[label distance=0.0cm,rotate=\azellipsisangle]:$\cdots$}] {};
\node(lefta2z2_ellipses) at ($(a2)!\azellipsisdist!(center_left_edge_ellipses)$) [label={[label distance=0.0cm,rotate=-\azellipsisangle]:$\cdots$}] {};

\node(righta3z3_ellipses) at ($(a3)!\azellipsisdist!(center_right_edge_ellipses)$) [label={[label distance=0.0cm,rotate=\azellipsisangle]:$\cdots$}] {};
\node(righta4z4_ellipses) at ($(a4)!\azellipsisdist!(right_edge_ellipses)$) [label={[label distance=0.0cm,rotate=-\azellipsisangle]:$\cdots$}] {};

\node (split_ellipses) at ($(split12)!0.5!(split21)$) {\LARGE$\cdots$};

\node (gadget_label) [above left = 1cm and -1.25cm of a1] {\large gadget: \LARGE$\mathtt{G}_e$};
\node (set_gadget_label1) [above = 0.1cm of split11] {\LARGE $\mathtt{D}_{S_1}$};
\node (set_gadget_label2) [above = 0.1cm of split21] {\LARGE $\mathtt{D}_{S_{f_e}}$};

\node (dummy_corner) [below right = 0.5cm and 0.1cm of a4] {};
\node (box) [fit=(gadget_label)(vfe)(z4), draw=black, line width=2pt,  inner sep=10pt, dashed] {};

\node (box_split1) [fit=(split11)(set_gadget_label1)(split12), draw=black, line width=2pt,  inner sep=10pt] {};
\node (box_split2) [fit=(split21)(set_gadget_label2)(split22), draw=black, line width=2pt,  inner sep=10pt] {};

\draw (v1) edge (a1);
\draw (v1) edge (a2);
\draw (vfe) edge (a3);
\draw (vfe) edge (a4);

\def\bendangle{17}

\draw (z1) edge [bend left = \bendangle] (z4);
\draw (z1) edge (left_z_ellipses);
\draw (left_z_ellipses) edge (z2);
\draw (z2) edge (center_z_ellipses);
\draw (center_z_ellipses) edge (z3);
\draw (z3) edge (right_z_ellipses);
\draw (right_z_ellipses) edge (z4);

\draw (z1) edge (a1);
\draw (z1) edge (a2);
\draw (z2) edge (a1);
\draw (z2) edge (a2);
\draw (z3) edge (a3);
\draw (z3) edge (a4);
\draw (z4) edge (a3);
\draw (z4) edge (a4);

\draw (v1) edge (split11);
\draw (v1) edge (split12);
\draw (vfe) edge (split21);
\draw (vfe) edge (split22);

\end{tikzpicture}
    \caption{$\sG_e$ is the gadget corresponding to element $e \in \mcU$, described in Definition~\ref{def:sc-de-gadget}.
    }\label{fig:de-element-gadget}
\end{figure}
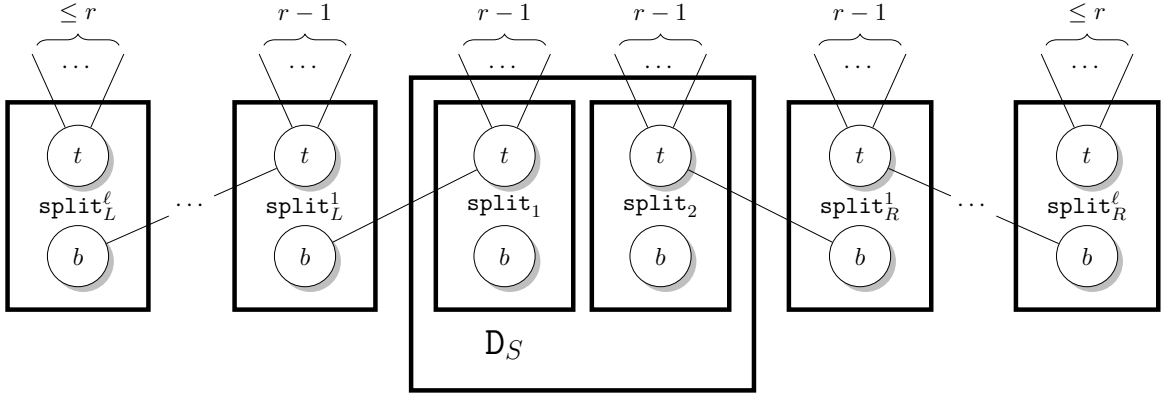
\begin{figure}[!h]
    \centering
	\scalebox{0.85}{
		\begin{tikzpicture}
\tikzstyle{node_style} = [state, fill=white, drop shadow, minimum size=0.85cm]
\tikzstyle{small_node_style} = [state, fill=white, drop shadow]
\tikzstyle{square_style} = [minimum width=1.3cm, minimum height=1.3cm,
                            fill=white, draw=black, drop shadow, line width=2pt]
\tikzstyle{edge_style} = [line width=2pt]
\tikzstyle{path_style} = [line width=2pt, decorate, decoration={snake}]
\tikzstyle{arrow_style} = [draw, fill=white, single arrow, single arrow head indent=1ex, minimum size=1cm, drop shadow]

\def\distvtosplit{0.6}
\def\xsplitdist{2.15}
\def\xsplitdistextra{\xsplitdist*1.2}
\def\spltspace{0.6}
\node (vs) {};

\node (split1b) [small_node_style, left = \distvtosplit cm of vs] {$b$};
\node (split1t) [small_node_style, above = \spltspace cm of split1b] {$t$};
\node (split1gadget_label) at ($(split1b)!0.5!(split1t)$) {$\mathtt{split}_1$};
\node (box1) [fit=(split1b)(split1t)(split1gadget_label), draw=black, line width=2pt,  inner sep=10pt] {};

\node (split2b) [small_node_style, right = \distvtosplit cm of vs] {$b$};
\node (split2t) [small_node_style, above = \spltspace cm of split2b] {$t$};
\node (split2gadget_label) at ($(split2b)!0.5!(split2t)$) {$\mathtt{split}_{2}$};
\node (box2) [fit=(split2b)(split2t)(split2gadget_label), draw=black, line width=2pt,  inner sep=10pt] {};

\node (set_gadget_label) [below = 1.5cm of split1gadget_label] {\LARGE $\mathtt{D}_S$};
\node (box) [fit=(box1)(split1t)(box2)(split2b)(set_gadget_label), draw=black, line width=2pt,  inner sep=10pt] {};

\node (left_splitb) [small_node_style, left = \xsplitdist cm of split1b] {$b$};
\node (left_splitt) [small_node_style, above = \spltspace cm of left_splitb] {$t$};
\node (left_split_label) at ($(left_splitb)!0.5!(left_splitt)$) {$\mathtt{split}_L^1$};
\node (boxL1) [fit=(left_splitb)(left_splitt)(left_split_label), draw=black, line width=2pt,  inner sep=10pt] {};

\node (right_splitb) [small_node_style, right = \xsplitdist cm of split2b] {$b$};
\node (right_splitt) [small_node_style, above = \spltspace cm of right_splitb] {$t$};
\node (right_split_label) at ($(right_splitb)!0.5!(right_splitt)$) {$\mathtt{split}_R^1$};
\node (boxR1) [fit=(right_splitb)(right_splitt)(right_split_label), draw=black, line width=2pt,  inner sep=10pt] {};

\node (left_splitb2) [small_node_style, left = \xsplitdistextra cm of left_splitb] {$b$};
\node (left_splitt2) [small_node_style, above = \spltspace cm of left_splitb2] {$t$};
\node (left_split_label2) at ($(left_splitb2)!0.5!(left_splitt2)$) {$\mathtt{split}_L^{\ell}$};
\node (boxL1) [fit=(left_splitb2)(left_splitt2)(left_split_label2), draw=black, line width=2pt,  inner sep=10pt] {};

\node (right_splitb2) [small_node_style, right = \xsplitdistextra cm of right_splitb] {$b$};
\node (right_splitt2) [small_node_style, above = \spltspace cm of right_splitb2] {$t$};
\node (right_split_label2) at ($(right_splitb2)!0.5!(right_splitt2)$) {$\mathtt{split}_R^{\ell}$};
\node (boxR1) [fit=(right_splitb2)(right_splitt2)(right_split_label2), draw=black, line width=2pt,  inner sep=10pt] {};

\node (left_split_ellipsis) at ($(left_splitb2)!0.5!(left_splitt)$) {$\large \cdots$};
\node (right_split_ellipsis) at ($(right_splitb2)!0.5!(right_splitt)$) {$\large \cdots$};

\draw (split1t) edge (left_splitb);
\draw (split2t) edge (right_splitb);

\draw (left_splitt) edge ($(left_splitt)!0.8!(left_split_ellipsis)$);
\draw (right_splitt) edge ($(right_splitt)!0.8!(right_split_ellipsis)$);
\draw (left_splitb2) edge ($(left_splitb2)!0.8!(left_split_ellipsis)$);
\draw (right_splitb2) edge ($(right_splitb2)!0.8!(right_split_ellipsis)$);

\def\horizwidth{0.3}
\def\edgeheight{1.25}
\def\labelheight{0.35}
\foreach \vert in {left_splitt2, right_splitt2}
{
    \node (1) [above left = \edgeheight cm and \horizwidth cm of \vert] {};
    \node (2) [above right = \edgeheight cm and \horizwidth cm of \vert] {};
    \node (dummy_mid) at ($(1)!0.5!(2)$) {};
    \node (ellipsis) [below = -0.01cm of dummy_mid] {$\cdots$};

    \node (label) [above = \labelheight cm of ellipsis] {$\leq r$};

    \draw [decorate,decoration={brace,amplitude=5pt},rotate=90] (1) -- (2);

    \draw (\vert) edge (1);
    \draw (\vert) edge (2);
}

\foreach \vert in {left_splitt, right_splitt, split1t, split2t}
{
    \node (1) [above left = \edgeheight cm and \horizwidth cm of \vert] {};
    \node (2) [above right = \edgeheight cm and \horizwidth cm of \vert] {};
    \node (dummy_mid) at ($(1)!0.5!(2)$) {};
    \node (ellipsis) [below = -0.01 cm of dummy_mid] {$\cdots$};

    \node (label) [above = \labelheight cm of ellipsis] {$r-1$};

    \draw [decorate,decoration={brace,amplitude=5pt},rotate=90] (1) -- (2);

    \draw (\vert) edge (1);
    \draw (\vert) edge (2);
}

\end{tikzpicture}
	}
    \caption{By linking together instances of \splitg, we can connect the gadget $\DS_S$
	to all $|S|$ element gadgets $\sG_e$ with $e \in S$ while also guaranteeing
	no vertex $t$ in an instance of \splitg has more than $\Tr$ external edges.
	An instance of \splitg is connected to an element gadget $\sG_e$ by a single edge from the vertex $t$ in \splitg to a vertex in $\sG_e$.
	In each instance of \splitg, the vertex $b$ is connected via a single edge to either the node $v_S$ in a set gadget $\DS_S$, or to a vertex $t$ in another instance of \splitg.
	See Definition~\ref{def:sc-de-gadget}.
    }\label{fig:de-arrange-splits}
\end{figure}

Next we establish some properties of these gadgets that will be useful in proving Theorem~\ref{theorem:hardness-node-degeneracy}.

\begin{lemma}\label{lemma:number-splits}
		In the graph $G = f(\mcU,\mcF)$ from Definition~\ref{def:sc-de-gadget}, if $|S| > \Tr$, then the set gadget $\DS_S$ requires exactly $2 \ceil{\tfrac{|S| - \Tr }{\Tr-1}}$ additional copies of \splitg to connect $\splitg_1,\splitg_2 \in \DS_S$ to all of the element gadgets $\sG_e$ for $e \in S$.
\end{lemma}
\begin{proof}
		Each split gadget $\splitg_j \in \DS_S$ must be able to attach to a vertex of the form $v_e^j \in \sG_e$ for each $e \in S$.
		This requires $|S|$ edges; since each instance of \splitg provides (up to) $\Tr$ edges out of its vertex $t$, but each additional copy of \splitg uses one such edge to connect to its vertex $b$, the total number $\ell$ of additional copies of \splitg is the smallest integer satisfying $(\Tr-1)\ell + \Tr \geq |S|$.
		Choosing $\ell = \ceil{\tfrac{|S| - \Tr }{\Tr-1}}$ guarantees this.
		Since $\DS_S$ has two split gadgets and \textit{each} requires the same number of additional edges, this completes the proof.
\end{proof}

We now give a precise bound on the size of the graph $G = f(\mcU,\mcF)$.
\begin{lemma}\label{lemma:de-graph-size}
		Given $\Tr \geq 2$ and an instance $\SC(\mcU,\mcF)$, the graph $G = f(\mcU,\mcF)$ (Definition~\ref{def:sc-de-gadget}) has size $n \leq 10\Tr |\mcU||\mcF|$, and the map $f$ can be contructed in time polynomial in $|\mcU|, |\mcF|,$ and $\Tr$.
\end{lemma}
\begin{proof}
		Each $\DS_S$ contains $4\Tr+5$ vertices, plus $\Tr+2$ vertices per additional copy of \splitg needed to connect $\DS_S$ to the gadgets $\sG_e$.
		By Lemma~\ref{lemma:number-splits}, we use $2 \ceil{\tfrac{|S| - \Tr }{\Tr-1}} \leq 2( (|S| - \Tr)/(\Tr-1) + 1) = 2(|S|-1)/(\Tr-1)$ copies of \splitg.
		Since $\Tr\geq 2$, we can upperbound $(\Tr+2)/(\Tr-1)$ by 4, and so the total number of vertices in extra \splitg copies for a $\DS_S$ is bounded by $8|S|-8$.
		Thus, the gadget $\DS_S$ and associated copies of \splitg have at most $4\Tr - 3 + 8|S|$ vertices.
		Each $\sG_e$ contains $2\Tr f_e$ vertices.
		Since $G$ contains a copy of $\sG_e$ for each $e \in \mcU$, and a gadget $\DS_S$ for each $S \in \mcF$, we have
		\[
				n \quad \leq \quad \sum_{e \in \mcU} 2\Tr f_e  + \sum_{S \in \mcF} (4\Tr - 3 + 8|S|)
				  \quad \leq \quad 2\Tr |\mcU||\mcF|  + 4\Tr |\mcF| + 8 \sum_{S \in \mcF} |S|.
		\]
		Observing that $8\sum_{S \in \mcF} |S| \leq 4\Tr|\mcU||\mcF|$ and $4\Tr|\mcF| \leq 4\Tr|\mcU||\mcF|$, we get
		$n \leq 10\Tr |\mcU| |\mcF|$.
\end{proof}

\begin{lemma}\label{lemma:de-degeneracy}\label{lem:sc-de-gadget-degeneracy}
		Given $\Tr \geq 2$ and an instance $\SC(\mcU,\mcF)$, the graph $G = f(\mcU,\mcF)$ (Definition~\ref{def:sc-de-gadget}) has degeneracy $\Tr+1$.
		Moreover, for each set $S$ the induced subgraph $G[\DS_S]$ has degeneracy $\Tr$, and for each element $e$,
		all vertices in $\sG_e$ are not in the $(\Tr+1)$-core if at least one vertex $v_e^j \in \sG_e$ has both of its external edges removed.
		Finally, in any copy of $\splitg$, if either $t$ or $b$ has no external edges attached to it, then the vertices in $\splitg$ are not in the $(\Tr+1)$-core; otherwise, they are in the $(\Tr+1)$-core.
\end{lemma}
\begin{proof}
		By the construction and the assumptions that no $S \in \mcF$ is empty and each element $e \in \mcU$ is contained in at least two sets,
		we know every vertex has degree at least $\Tr+1$ in $G$.
		To see that $G$ has degeneracy exactly $\Tr+1$, we exhibit an ordering to use Lemma~\ref{lem:degen-properties}.
		Note that all vertices $w_j$ in set gadgets $\DS_S$ have degree $\Tr+1$.
		Removing these leaves the vertices $v_S, x_j \in \DS_S$ (for each $j\leq r$) with degree two or less, so they can be removed as well.
		After this, all vertices remaining in $G$ have degree $\Tr+1$ or less, possibly with the exception of the vertices $t$ in copies of $\splitg$. After removing all vertices other than those labeled $t$ in instances of $\splitg$, then those vertices have degree $\leq \Tr$ as well.
		This proves $G$ has degeneracy $\Tr+1$.

		To see that the vertices of $\splitg$ are not in the $(\Tr+1)$-core if $t$ or $b$ has no external edges, consider the vertex ordering startin with either $t$ or $b$ (whichever vertex has no external edges to it), then the $\Tr$ clique vertices in \splitg, and finally the remaining vertex.

		Suppose $v_e^0 \in \sG_e$ has had its edges external to $\sG_e$ removed;
		then $v_e^0$ has degree $\Tr-1$ and
		the vertices of $\sG_e$ are not in the $(\Tr+1)$-core
		as witnessed by the vertex ordering $v_e^0$, $\{a_i^0\}_{i}$, $\{z_i^0\}_{i}$,
		 $\{z_i^j\}_{i,j}$, $\{a_i^j\}_{i,j}$, $\{v_e^j\}_{j}$.
		The degeneracy $\Tr$ of the induced subgraph $G[\DS_S]$ is exhibited by the vertex ordering $t_1, t_2, \splitg_1, \splitg_2$, $v_S, \{w_j\}_j, \{x_j\}_j$.
\end{proof}

\begin{lemma}\label{lemma:set-sol}
	Given instance $\SC(\mcU,\mcF)$, let $G = f(\mcU,\mcF)$ (Definition~\ref{def:sc-de-gadget}).
	Let $\mathcal{T} \subseteq \mcF$ be any set cover for $\SC(\mcU,\mcF)$.
	Let $\yB \subseteq V(G)$  ($\yB \subseteq E(G))$ be any set of vertices (edges) such that for each $S \in \mathcal{T}$,
	the set $\yB\cap V(\DS_S)$ is non-empty (the set $\yB\cap E(\DS_S)$ contains at least one edge incident to a vertex $w_j$).
	Then $\yB$ is a feasible solution for $\bDEV(G)$ ($\bDEE(G)$).
\end{lemma}
\begin{proof}
		We first prove the lemma for \bDEE.
		Deleting any edge in a gadget $\DS_S$ that is incident to a vertex $w_j$ makes the degree of that vertex $\Tr$---thus, that vertex is removed from the $(\Tr+1)$-core; this decreases the degree of the neighbors $x_i$ of $w_j$ to $\Tr$ as well, and by induction each $x_i$ in $\DS_S$ is removed from the ($\Tr+1$)-core.
		This removes all $w_j$ from the $(\Tr+1)$-core, which then removes $v_S$, and by Lemma~\ref{lem:sc-de-gadget-degeneracy}, both \splitg gadgets.
		Furthermore, this causes all edges leaving $\DS_S$ to be removed, and so all \splitg and $\sG_e$ gadgets attached to those edges then leave the $(\Tr+1)$-core, by Lemma~\ref{lem:sc-de-gadget-degeneracy}.
		By assumption, $\mathcal{T}$ is a set cover, and so the set of gadgets $\DS_S$ for $S\in \mathcal{T}$ necessarily connect to every element gadget $\sG_e$ in $G$.
		This shows that every $\sG_e$ is removed from the $(\Tr+1)$-core of $G\setminus \yB$.

		With every element gadget $\sG_e$ removed from the $(\Tr+1)$-core of $G\setminus \yB$, every edge going from a vertex in $\sG_e$ to a vertex $t$ in a copy of $\splitg$ is removed; this means all copies of \splitg are removed from the $(\Tr+1)$-core by Lemma~\ref{lemma:de-degeneracy}.
		This means that every set gadget $\DS_S$ (which wasn't already removed because $S \in \mathcal{T}$) has no external edges to element gadgets or \splitg gadgets,
		and so by Lemma~\ref{lem:sc-de-gadget-degeneracy} each $\DS_S$ is removed from the $(\Tr+1)$-core.
		Hence, $G\setminus\yB$ has degeneracy $\Tr$.\looseness-1

		To see that the result also holds for vertex deletions, note that if $\yB$ contains any vertex in a set gadget $\DS_S$, then deleting $\yB$ necessarily removes that instance of $\DS_S$ from the $(\Tr+1)$-core of $G\setminus \yB$.
		The rest of the proof follows as above.
\end{proof}

\mypar{Mapping a graph deletion set to a \SCfull solution.}

\begin{definition}\label{def:map-bde-to-sc}
	Given instance $\SC(\mcU,\mcF)$, let $G = f(\mcU,\mcF)$ (Definition~\ref{def:sc-de-gadget}).
	For any feasible solution $\yB \subset V(G)$  ($\yB \subset E(G)$) to $\bDEV(G)$ ($\bDEE(G)$),
	we construct a \emph{map $g$ from \bDE to \SC} that maps $\yB$ to a solution $g(\yB)$ for $\SC(\mcU,\mcF)$ in two steps:
	\begin{enumerate}
		\itemsep0em
		\item Construct a \emph{canonical edit set} $\ycan$ as follows: for each set $S \in \mcF$, let $H_S \subset G$ be the subgraph formed by $\DS_S$, $\sG_e$ for each element $e \in S$, and every \splitg lying on a path from $\DS_S$ to any such $\sG_e$.
		If $H_S \cap \yB$ is non-empty, then add the vertex $v_S$ (or the edge $(v_S,w_1)$) to $\ycan$.\looseness-1
		\item Set $g(\yB) = \{ S \in \mcU \mid \ycan \cap \DS_S \neq \emptyset \}$.
    \end{enumerate}
\end{definition}

\begin{lemma}\label{lemma:canonical-de}
	In the notation of Definition~\ref{def:map-bde-to-sc}, with $G = f(\mcU,\mcF)$, for a feasible solution $\yB$ to $\bDE(G)$,
	the corresponding set $\ycan$ is a feasible solution to $\bDE(G)$, the set $g(\yB)$ is a feasible solution to $\SC(\mcU,\mcF)$,
	and they satisfy $|g(\yB)| \leq |\ycan| \leq |\yB|$.
\end{lemma}
\begin{proof}

		For each element $e \in \mcU$ let $H_e$ be the subgraph formed by $\sG_e$, the set gadgets $\DS_S$ for each set $S \in \mcF$ containing $e$, and every gadget \splitg lying on a path from each such $\DS_S$ to $\sG_e$.

		Let $\yB$ be a feasible solution to $\bDE(G)$ and suppose the corresponding canonical solution $\ycan$ is infeasible.
		Consider the subgraph $H_e^+ = H_e \cup \left( \bigcup_{S: e \in S} H_S \right)$ of $G$.
		If for each $e \in \mcU$ we have that $\ycan \cap H_e^+$ is non-empty, then $\ycan$ would be feasible.
		To see this, observe that in this situation each $\sG_e$ would be removed from the $(\Tr+1)$-core, and if all $\sG_e$ are removed from the $(\Tr+1)$-core, then each gadget $\DS_S$ is also removed, leaving the $(\Tr+1)$-core empty.
		Thus, $\ycan \cap H_e^+$ is empty for some $e$.
		However, since $\yB$ is feasible, $\yB \cap H_e^+ \neq \emptyset$ for every $e$;
		in particular, $\yB \cap H_S$ must be non-empty for some $S$ containing $e$.
		By Definition~\ref{def:map-bde-to-sc}, $\ycan \cap H_S$ must then be non-empty, a contradiction.\looseness-1

		Given a feasible solution $\yB$ to $\bDE(G)$, let $\ycan$ be the corresponding canonical solution and suppose that $g(\yB)$ is not a feasible solution to $\SC$.
		Then there is at least one element $e \in \mcU$ that is not covered by $g(\yB)$.
		Hence, for each set $S \in \mcF$ containing $e$, we have that $\ycan \cap H_S$ is empty.
		Since $e$ is not covered by $g(\yB)$ by assumption, then by construction $\ycan$ contains no edge (vertex) in $H_e^+$.
		We will show this implies that $H_e^+$ has minimum degree $\Tr+1$.
		Since $H_e^+$ has no edge (vertex) in $\ycan$, in particular $\ycan$ contains no edge (vertex) of $\sG_e$, any $\DS_S$ with $e \in S$, or any instance of \splitg lying on a path between $\sG_e$ and such a $\DS_S$.
		Thus, vertices in each $\DS_S$ have degree at least $\Tr+1$, since the vertices $t \in \splitg_j \in \DS_S$ each have at least one external edge on a path to $\sG_e$;
		clearly the instances \splitg lying on paths from $\DS_S$ to $\sG_e$ each have minimum degree $\Tr+1$;
		furthermore, since $H_e^+$ contains $\DS_S$ for each $S$ containing $e$, each vertex $v_e^j \in \sG_e$ has two external edges connected to instances \splitg.
		Since no edges (vertices) internal to $\sG_e$ are missing, this proves $H_e^+$ has minimum degree $\Tr+1$.
		This contradicts feasibility of $\ycan$, so $g(\yB)$ must be feasible.\looseness-1

		The inequalities follow directly from the definitions of the sets $\ycan$ and $g(\yB)$.
\end{proof}

\mypar{Strict reduction from \SCfull.}\label{sec:sc-de-reduction}
In the notation of Definition~\ref{def:sc-de-gadget}, fix an instance $\SC(\mcU,\mcF)$ and corresponding instance $\bDE(G)$.
We can now prove the main theorem of this section.

\begin{proof}[Proof of Theorem~\ref{theorem:hardness-node-degeneracy}]
		By Lemma~\ref{lemma:de-graph-size}, the map $f$ can be computed in time polynomial in $|\mcU|, |\mcF|,$ and $\Tr$.
		We also observe that, given a feasible solution $\yB$ to $\bDE(G)$, the map $g$ gives a polynomial time construction of a solution $g(\yB)$ to $\SC(\mcU,\mF)$.
		Furthermore, by Lemma~\ref{lemma:canonical-de}, we know that $\cost_{\SC}(g(\yB)) \leq \cost_{\bDE}(\yB)$.
		In particular, this implies $\opt_{\SC}(\mcU,\mcF) \leq \opt_{\bDE}(G)$.

		On the other hand, given any solution $\yA$ to $\SC(\mcU,\mcF)$,
		the set $\ycan = \{ v_S \mid S \in \yA \}$ is a valid solution for $\bDEV(f(\mcU,\mcF))$ by Lemma~\ref{lemma:set-sol} and satisfies
		$\cost_{\SC}(\yA) \geq \cost_{\bDEV}(\ycan)$.
		Additionally, the set $\ycan = \{ (v_S,w_1) \mid S \in \yA \}$ is a valid solution for $\bDEE$ and satisfies
		$\cost_{\SC}(\yA) \geq \cost_{\bDEE}(\ycan)$.
		In particular, for optimal $\yA$ this implies
		$\opt_{\SC}(\mcU,\mcF) \geq \opt_{\bDE}(f(\mcU,\mcF))$
		completing a proof that $\opt_{\SC}(\mcU,\mcF) = \opt_{\bDE}(f(\mcU,\mcF))$ for both vertex and edge deletion.

		This proves that maps $f$ and $g$ satisfy the conditions of Definition~\ref{def:strict-reduction}.
		Finally, by Lemma~\ref{lemma:de-graph-size} the map $f$ creates a graph with $n \leq 10 \Tr |\mcU||\mcF| $.
		By Theorem~\ref{lem-set-cover-sparse}, we can assume that $|\mcF| \leq |\mcU|^C$ for some constant $C > 0$,
		so we have $\log(n/\Tr) \leq \log(|\mcU|^{C+1} )$.
		Thus, there exists a constant $C_1>0$ such that a $C_1 \cdot \ln(n / \Tr)$-approximation for $\bDE(G)$ would give an $(1-\epsilon)(\ln |\mcU| )$-approximation for all \SC instances with $|\mcF| \leq |\mcU|^C$.
\end{proof}

\fi

\ifdefined\iswcolhard
    \subsection{Weak Coloring Numbers: $o(\Tt)$-Inapproximability of Editing}
    \label{section:negative_weak_coloring}
    
In this section, we explore the hardness of editing a graph $G$ so that it has weak $\Tc$-coloring number at most $\Tt$ (see Definition~\ref{definition:wcol}).
Because it is NP-hard to decide whether a graph has weak $\Tc$-coloring number at most $\Tt$, it is trivially hard to guarantee a finite approximation ratio for editing graphs which already have weak $\Tc$-coloring number at most $\Tt$ into the desired class.
By contrast, our hardness-of-approximation results apply when editing graphs that are far from the target class, which is the most favorable scenario for approximation.
Throughout this section, we say a vertex $u$ can {\em reach} another vertex $v$ with respect to an ordering $L$ if $v \in \wreach{\Tc}{G}{L}{u}$.

We examine three variants of the problem---namely vertex-deletion, edge-deletion, and edge-contraction---and prove the following results.

\begin{theorem} \label{theorem:bwcne-hard}
  For each fixed $\Tc > 2$ and constant $\Tt \geq \max\{12, \floor{\Tc/2}+4\}$, \bWCNfull (\bWCN) is $o(\Tt)$-inapproximable.
  Furthermore, there exists a constant $b > 0$ such that
  \bWCN is $o(\Tt)$-inapproximable for $\Tt \in  O( \log^{1/b} n )$,
  unless NP $\subseteq$ DTIME$(n^{O(\log\log n)})$.
\end{theorem}

The case where $\Tc = 1$ is equivalent to editing a graph to have \textit{degeneracy} $\Tt$, and it is handled in Section~\ref{section:negative_bounded_degeneracy_editing}.
The case where $\Tc = 2$ uses a similar construction to when $\Tc > 2$ and is explained at the end of this section.

\mypar{Reduction Strategy.}
We proceed via a strict reduction from a variant of \SCfull called \uSCfull (see Section~\ref{sec:prelims-reductions}).

We describe a function $f$ that maps an instance $(\mcU, \mcF)$ of \SCfull to a graph $G = f(\mcU, \mcF)$.
The graph $G$ consists of an element gadget $\weg{x}$ for each element $x \in \mcU$ and a set gadget $\wsg{S}$ for each set $S \in \mcF$ such that the gadget $\weg{x}$ contains $\wsg{S}$ when $x \in S$.
Thus, two gadgets $\weg{x}$ and $\weg{x'}$ overlap on $\wsg{S}$ for every set $S \in \mcF$ with $x, x' \in S$.
For an edit set $\yB$, we show that $\wcol{\Tc}{G \setminus \yB} \leq \Tt$ if and only if $\wcol{\Tc}{\weg{x} \setminus \yB} \leq \Tt$ for all $x \in \mcU$.
Additionally, we show that $\wcol{\Tc}{\weg{x} \setminus \yB} \leq \Tt$ if every $\weg{x}$ contains a $\wsg{S}$ that was edited.
We define a canonical solution to \bWCN so that all of the edits appear in set gadgets and then use a bijective mapping between problems to obtain the necessary inequalities.
We begin by introducing the gadgets in this reduction.

\paragraph{Gadgets.}
Let $f_x$ be the frequency of $x$ in $\mcF$, and let $f_{max} = \max_{x \in \mcU} \left(f_x\right)$.
We note that $f_x \geq 1$ trivially for all $x \in \mcU$ in all hard instances $(\mcU, \mcF)$.
Let $\ell = \floor{\frac{c}{2}}$ and let $\Tt$ be the smallest integer that satisfies $\Tt \geq 4 \cdot f_{max}$ and $\Tt \geq \ell + 3 \cdot f_{max} - 2$.
Note that $\ell \geq 1$ since $\Tc > 2$.

\begin{definition} \label{definition:wcn-gs}
The \emph{set gadget} $\wsg{S}$, depicted in Figure~\ref{fig:wcn-sg}, consists of two vertex-disjoint cliques $D_S^1$ and $D_S^2$ of sizes $\Tt$ and $\Tt + 1$, respectively.
We distinguish one vertex $v_S^1 \in D_S^1$ and another vertex $v_S^2 \in D_S^2$ and connect them by an edge.
\end{definition}

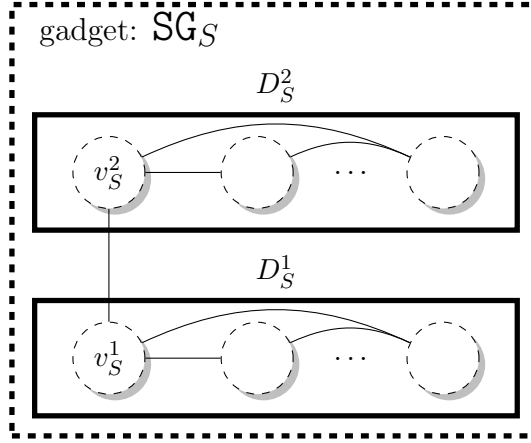
\begin{figure}
  \centering

\begin{tikzpicture}
\tikzstyle{node_style} = [state, fill=white, drop shadow]
\tikzstyle{edge_style} = [line width=2pt]
\tikzstyle{path_style} = [line width=2pt, decorate, decoration={snake}]
\tikzstyle{arrow_style} = [draw, fill=white, single arrow, single arrow head indent=1ex, minimum size=1cm, drop shadow]

\def\edgebend{25}

\node (vs1) [node_style, dashed] {$v_S^1$};
\node (vs2) [node_style, dashed, above = 1.5cm of vs1] {$v_S^2$};

\node (ds11) [node_style, dashed, right = 1cm of vs1] {};
\node (ds1t) [node_style, dashed, right = 1.5cm of ds11] {};
\node (ds1d) [dashed, right = 0.4cm of ds11] {$\dots$};

\node (ds21) [node_style, dashed, right = 1cm of vs2] {};
\node (ds2t) [node_style, dashed, right = 1.5cm of ds21] {};
\node (ds2d) [dashed, right = 0.4cm of ds21] {$\dots$};

\draw (vs1) edge (vs2);
\draw (vs1) edge (ds11);
\draw (vs2) edge (ds21);
\draw (vs1) edge [bend left = \edgebend] (ds1t);
\draw (vs2) edge [bend left = \edgebend] (ds2t);
\draw (ds11) edge [bend left = \edgebend] (ds1t);
\draw (ds21) edge [bend left = \edgebend] (ds2t);

\node (dummy1l) [left = -0.4cm of vs1] {};
\node (dummy1r) [right = -0.4cm of ds1t] {};
\node (ds1) [fit=(dummy1l)(dummy1r),
          draw=black,
          line width=2pt,
          label={above:$D_S^1$},
          inner sep=18pt] {};

\node (dummy2l) [left = -0.4cm of vs2] {};
\node (dummy2r) [right = -0.4cm of ds2t] {};
\node (ds2) [fit=(dummy2l)(dummy2r),
        draw=black,
        line width=2pt,
        label={above:$D_S^2$},
        inner sep=18pt] {};

\node (dummyr) [below right = 0cm and 3.88cm of vs1] {};
\node (dummyl) [above left = 1.2cm and 3.9cm of ds2t] {};

\node (label) [above right = 1.2cm and -1.4cm of vs2] {\large gadget: \LARGE $\texttt{SG}_S$};

\node (sgs) [fit=(dummyl)(dummyr),
			draw=black,
			dashed,
			line width=2pt,
			minimum width=7cm,
			inner sep=12pt] {};

\end{tikzpicture}
  \caption{\label{fig:wcn-sg}
    The set gadget $\wsg{S}$ (Definition~\ref{definition:wcn-gs}).  $\wsg{S}$ contains cliques $D_S^1$ of size $\Tt$ and $D_S^2$ of size $\Tt + 1$ with distinguished vertices $v_S^1$ and $v_S^2$, respectively, connected by an edge.
  }
\end{figure}

\begin{definition} \label{definition:wcn-gx}
Each \emph{element gadget} $\weg{x}$, depicted in Figure~\ref{fig:wcn-eg}, contains a clique $D_x$ of size $\Tt - 3 \cdot f_x + 2$.
Note that, by our choice of $\Tt$, $D_x$ always contains at least one vertex.
For each $S \in \mcF$ which contains $x$, $\weg{x}$ contains $\wsg{S}$ and an additional clique $D_{S,x}$ of size $f_x$.
For each set $S \in \mcF$ for which $x \in S$, there is a vertex $p_{S,x}^1$ connected to all vertices in $D_x$ and $D_{S,x}$.
Additionally, the vertex $p_{S,x}^1$ is in a path of length $\ell$, $\set{p_{S,x}^1 \dots p_{S,x}^\ell}$, and the vertex $p_{S,x}^\ell$ is connected to $v_S^1$ in an instance of $\wsg{S}$.
\end{definition}

\paragraph{Reduction.}
We now define reduction functions $f$ and $g$
(Definition~\ref{def:strict-reduction}).

\begin{definition} \label{definition:wcn-f}
Given a \SCfull instance $(\mcU, \mcF)$, the function \emph{$f(\mcU, \mcF)$} produces the corresponding \bWCN instance $G$.
Specifically, $G$ contains a set gadget $\wsg{S}$ for every $S \in \mcF$ as described in Definition~\ref{definition:wcn-gs}.
Using the set gadgets, $G$ contains an element gadget $\weg{x}$ for every $x \in \mcU$ as described in Definition~\ref{definition:wcn-gx}.
The graph $G$ is the union of these element gadgets.
\end{definition}

\begin{figure}
  \centering
\begin{tikzpicture}
\tikzstyle{node_style} = [state, fill=white, drop shadow]
\tikzstyle{edge_style} = [line width=2pt]
\tikzstyle{path_style} = [line width=2pt, decorate, decoration={snake}]
\tikzstyle{arrow_style} = [draw, fill=white, single arrow, single arrow head indent=1ex, minimum size=1cm, drop shadow]

\def\edgebend{25}

\node (dx1) [node_style] {};
\node (dx2) [node_style, below = 1cm of dx1] {};
\node (dxt) [node_style, below = 1.5cm of dx2] {};
\node (dxd) [below = 0.4cm of dx2] {\rotatebox{90}{$\dots$}};

\node (dummydxa) [above = -0.4cm of dx1] {};
\node (dummydxb) [below = -0.4cm of dxt] {};
\node (dx) [fit=(dummydxa)(dummydxb),
			draw=black,
			line width=2pt,
			label={left:$D_x$},
			inner sep=18pt] {};

\draw (dx1) edge (dx2);
\draw (dx1) edge [bend right = \edgebend] (dxt);
\draw (dx2) edge [bend right = \edgebend] (dxt);

\node (psx1) [node_style, above right = -0.5cm and 1cm of dx1] {$p_{S_1,x}^1$};
\node (psxl) [node_style, above right = 1cm and 1cm of psx1] {$p_{S_1,x}^\ell$};

\draw (dx1) edge (psx1);
\draw (dx2) edge (psx1);
\draw (dxt) edge (psx1);

\draw (psx1) edge [path_style] (psxl);

\node (ptx1) [node_style, below right = -0.5cm and 1cm of dxt] {$p_{S_{f_x},x}^1$};
\node (ptxl) [node_style, above right = 1cm and 1cm of ptx1] {$p_{S_{f_x},x}^\ell$};

\draw (dx1) edge (ptx1);
\draw (dx2) edge (ptx1);
\draw (dxt) edge (ptx1);

\draw (ptx1) edge [path_style] (ptxl);

\node (dots) [rotate = 90, above right = 1.8cm and -0.3cm of ptx1] {\Large{\textbf{$\dots$}}};

\node (dsx1) [node_style, right = 2cm of psx1] {};
\node (dsx2) [node_style, right = 1cm of dsx1] {};
\node (dsxt) [node_style, right = 1.5cm of dsx2] {};
\node (dsxd) [right = 0.4cm of dsx2] {$\dots$};

\node (dummydsxl) [left = -0.4cm of dsx1] {};
\node (dummydsxr) [right = -0.4cm of dsxt] {};
\node (dsx) [fit=(dummydsxl)(dummydsxr),
			draw=black,
			line width=2pt,
			label={below:$D_{S_1,x}$},
			inner sep=18pt] {};

\draw (dsx1) edge (dsx2);
\draw (dsx1) edge [bend right = \edgebend] (dsxt);
\draw (dsx2) edge [bend right = \edgebend] (dsxt);

\draw (psx1) edge (dsx1);
\draw (psx1) edge [bend left = \edgebend] (dsx2);
\draw (psx1) edge [bend left = \edgebend] (dsxt);

\node (dtx1) [node_style, right = 2cm of ptx1] {};
\node (dtx2) [node_style, right = 1cm of dtx1] {};
\node (dtxt) [node_style, right = 1.5cm of dtx2] {};
\node (dtxd) [right = 0.4cm of dtx2] {$\dots$};

\node (dummydtxl) [left = -0.4cm of dtx1] {};
\node (dummydtxr) [right = -0.4cm of dtxt] {};
\node (dtx) [fit=(dummydtxl)(dummydtxr),
			draw=black,
			line width=2pt,
			label={below:$D_{S_{f_x},x}$},
			inner sep=18pt] {};

\draw (dtx1) edge (dtx2);
\draw (dtx1) edge [bend right = \edgebend] (dtxt);
\draw (dtx2) edge [bend right = \edgebend] (dtxt);

\draw (ptx1) edge (dtx1);
\draw (ptx1) edge [bend left = \edgebend] (dtx2);
\draw (ptx1) edge [bend left = \edgebend] (dtxt);

\node (vs1) [node_style, dashed, right = 2cm of psxl] {$v_{S_1}^1$};

\node (dummysgsl) [left = -0.6cm of vs1] {};
\node (dummysgsr) [right = -0.6cm of vs1] {};
\node (sgs) [fit=(dummysgsl)(dummysgsr),
			draw=black,
			line width=2pt,
			label={right:$\texttt{SG}_{S_1}$},
			inner sep=18pt] {};

\draw (psxl) edge (vs1);

\node (vt1) [node_style, dashed, right = 2cm of ptxl] {$v_{S_{f_x}}^1$};

\node (dummysgtl) [right = -0.6cm of vt1] {};
\node (dummysgtr) [left = -0.6cm of vt1] {};
\node (sgt) [fit=(dummysgtl)(dummysgtr),
			draw=black,
			line width=2pt,
			label={right:$\texttt{SG}_{S_{f_x}}$},
			inner sep=18pt] {};

\draw (ptxl) edge (vt1);

\node (dummyl) [below left = 0cm and 2cm of ptx1] {};
\node (dummyr) [above right = 1.8cm and 0cm of dsxt] {};
\node (egx) [fit=(dummyl)(dummyr),
			draw=black,
			dashed,
			line width=2pt,
			inner sep=18pt] {};

\node (label) [above right = 1cm and -2.1cm of dx] {\large gadget: \LARGE$\texttt{EG}_x$};

\end{tikzpicture}
  \caption{\label{fig:wcn-eg}
    The element gadget $\weg{x}$ (Definition~\ref{definition:wcn-gx}).  It contains a clique $D_x$ of size $\Tt - 3 \cdot f_x + 2$ and a clique $D_{S,x}$ of size $f_x$ for each $S \in \mcF$ that contains $x$.
    $\weg{x}$ also contains the set gadget $\wsg{S}$ for each $S$ that contains $x$, and each vertex in $D_x$ and $D_{S,x}$ is connected to $\wsg{S}$ by a path of length $\ell$ (indicated by the squiggly line).
  }
\end{figure}
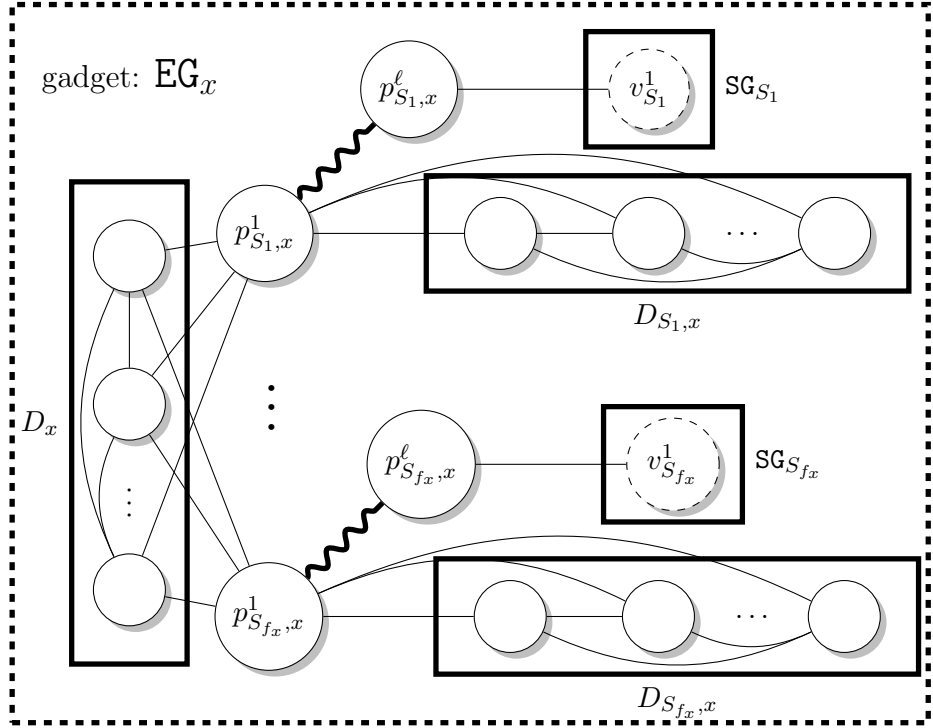

Before we demonstrate the mapping between solutions for \SCfull and \bWCN, we first show that $f$ is polynomial time computable and produces non-trivial instances of \bWCN.

\begin{lemma} \label{lemma:wcn-f-poly}
The function $f$ defined in Definition~\ref{definition:wcn-f} which produces the \bWCN instance $G$ from the \SCfull instance $(\mcU, \mcF)$ is polynomial time computable.
\end{lemma}

\begin{proof}
First, we note that $G$ includes one set gadget, $\wsg{S}$, for each set $S \in \mcF$, and each such gadget has exactly $2\Tt+1$ vertices.
Next, we observe that for each element $x \in \mcU$, $G$ contains an element gadget $\weg{x}$ that has $(f_x^2 + (\ell-3)f_x + \Tt+2)$ vertices:
$f_x$ vertices in each of $f_x$ cliques labelled $D_{S,x^j}$ for $j=1,\cdots, f_x$; $\Tt-3f_x+2$ vertices for the clique $D_x$; and $f_x$ paths of length $\ell$, $p_{S,x}^j$.
Then, $n = |V(G)|$ satisfies
\begin{align*}
  n &= |\mcF|(2\Tt+1) + \sum_{x \in \mcU} \left( f_x^2 + (\ell-3)f_x + \Tt+2 \right) \\
  &\leq |\mcF|(2\Tt+1) + |\mcU|(\Tt+2) + |\mcU|\Tt^2,
\end{align*}
since $\Tt \geq f_x$ and $\Tt \geq \ell+f_x - 3$ by construction.
Thus the number of vertices (and therefore edges) in $G$ is polynomial in $|\mcU|$, $|\mcF|$, $\Tt$, and $\Tc$.
\end{proof}

In order to show that an instance $G = f(\mcU, \mcF)$ of \bWCN is non-trivial, we show that the weak $\Tc$-coloring number of $G$ is greater than $\Tt$.
To prove this, we will use the fact that weak $\Tc$-coloring number is a hereditary property.

\begin{lemma} \label{lemma:wcol-hereditary}
For any graph $H$, if $\wcol{\Tc}{H} \leq \Tt$, then for each induced subgraph $H'$ of $H$, $\wcol{\Tc}{H'} \leq \Tt$.
\end{lemma}

\begin{proof}
Suppose that $\wcol{\Tc}{H} \leq \Tt$ for a graph $H$.
By Definition~\ref{definition:wcol}, there exists an ordering $L$ such that $|\wreach{\Tc}{H}{L}{v}| \leq \Tt$ for all $v \in V(H)$.
Then, $|\wreach{\Tc}{H'}{L}{v}| \leq |\wreach{\Tc}{H}{L}{v}|$ for any $v \in V(H')$ since $H'$ is an induced subgraph.
Thus, $\wscore{\Tc}{H'}{L} \leq \Tt$ and $\wcol{\Tc}{H'} \leq \Tt$ as desired.
\end{proof}

In particular, the contrapositive of Lemma~\ref{lemma:wcol-hereditary} will be useful.

\begin{corollary} \label{corollary:wcol-hereditary}
If $\wcol{\Tc}{H'} > \Tt$ for any induced subgraph $H'$ of $H$, then $\wcol{\Tc}{H} > \Tt$.
\end{corollary}

Now, we need only show that $\wcol{\Tc}{\weg{x}} > \Tt$ for some $x \in \mcU$.
To show that $\wcol{\Tc}{\weg{x}} > \Tt$, we establish structural properties that hold for any ordering $L$ that would satisfy $\wscore{\Tc}{\weg{x}}{L} \leq \Tt$ and then show that no such ordering can exist.
We first prove that the vertices $v_S^1$, $v_S^2$, $p_{S,x}^1$ can be assumed to come after every vertex in $D_S^1$, $D_S^2$, and $D_{S,x}$ respectively in the ordering $L$.

\begin{lemma} \label{lemma:wcn-cliques-first}
Given an instance $\SC(\mcU, \mcF)$, let $G$ be the graph $f(\mcU, \mcF)$ (Definition~\ref{definition:wcn-f}), and let $\weg{x}$ be an induced subgraph of $G$ as in Definition~\ref{definition:wcn-gx} for some $x \in \mcU$.
If $\exists L$ such that $\wscore{\Tc}{\weg{x}}{L} \leq \Tt$, then there exists $L'$ with $\wscore{\Tc}{\weg{x}}{L'} \leq \wscore{\Tc}{\weg{x}}{L}$ such that

\begin{enumerate}\itemsep0em
    \item $L'(d_1) < L'(v_S^1) \; \forall d_1 \in D_S^1$,
    \item $L'(d_2) < L'(v_S^2) \; \forall d_2 \in D_S^2$, and
    \item $L'(d) < L'(p_{S,x}^1) \; \forall d \in D_{S,x}$.
\end{enumerate}

\end{lemma}

\begin{proof}
Let $D$ be the relevant clique ($D_S^1$, $D_S^2$, or $D_{S,x}$) and $v$ be the relevant vertex ($v_S^1$, $v_S^2$, or $p_{S,x}^1$).
In any of the three cases, the vertices in $D$ only have edges to $v$.
The first vertex from $\{v\} \cup D$ in $L$ will be able to reach all other vertices in the clique, and so we can choose any of these vertices first \WLOG.
Assuming that the ordering of the other elements is fixed relative to $v$, then the first vertex will also be able to reach the same subset of vertices unless $v$ is the first vertex.
However, since all other vertices in $D$ require an additional step traversing to $v$ first, $v$ can reach at least as many other vertices outside of the clique.
Furthermore, if $v$ is last, then no vertex from $D$ is reachable from another vertex not in $D$.
While neither of these guarantee that $\wscore{\Tc}{\weg{x}}{L}$ decreases, both ensure that it will not increase meaning that it is safe to put $v$ last.
\end{proof}

This property allows us to assume that $D_S^1$, $D_S^2$, and $D_{S,x}$ appear at the beginning of an optimal ordering.
This is because the set of vertices weakly $\Tc$-reachable from a vertex in the clique $D$ only depends on the location of $v$ in the ordering.
Furthermore, since $v$ is last in the ordering, no other vertex can reach a vertex in the clique.
Thus, the placement of the clique does not affect the set of weakly $\Tc$-reachable vertices from a vertex outside a clique.

Given this assumption, next we argue that $v_S^1$ must come before $v_S^2$ in the ordering for all $S$ that contain $x$ due to the sizes of $D_S^1$ and $D_S^2$.

\begin{lemma} \label{lemma:vs2-last}
Given an instance $\SC(\mcU, \mcF)$, let $G$ be the graph $f(\mcU, \mcF)$ (Definition~\ref{definition:wcn-f}), and let $\weg{x}$ be an induced subgraph of $G$ as in Definition~\ref{definition:wcn-gx} for some $x \in \mcU$.
If $L$ is an ordering on the vertices of $\weg{x}$ such that $\wscore{\Tc}{\weg{x}}{L} \leq \Tt$, then $L(v_S^1) < L(v_S^2)$ for all $S$ that contain $x$.
\end{lemma}

\begin{proof}
Assume $L(v_S^1) > L(v_S^2)$ for some $S$ that contains $x$.
Consider $d$, the first vertex from $D_S^2$ in $L$.
Then, $|\wreach{\Tc}{\weg{x}}{L}{d}| \geq \Tt + 1$ since all $\Tt$ other vertices in $D_S^2$ and $v_S^1$ are weakly $\Tc$-reachable from $d$, a contradiction.
\end{proof}

We apply similar reasoning for the path vertices $p_{S,x}^1 \dots p_{S,x}^\ell$ and the vertices in $D_x$ to $v_S^1$.

\begin{lemma} \label{lemma:vs1-next}
Given an instance $\SC(\mcU, \mcF)$, let $G$ be the graph $f(\mcU, \mcF)$ (Definition~\ref{definition:wcn-f}), and let $\weg{x}$ be an induced subgraph of $G$ as in Definition~\ref{definition:wcn-gx} for some $x \in \mcU$.
If $L$ is an ordering on the vertices of $\weg{x}$ satisfying $\wscore{\Tc}{\weg{x}}{L} \leq \Tt$, then for each $v_S^1$ in $\weg{x}$, $L(u) < L(v_S^1)$ for all $u \in \{D_x \cup p_{S,x}^1 \cup \dots \cup p_{S,x}^\ell\}$.
\end{lemma}

\begin{proof}
Suppose $u \in \{D_x \cup p_{S,x}^1 \cup \dots \cup p_{S,x}^\ell\}$ and $L(u) > L(v_S^1)$.
Furthermore, assume \WLOG that $u$ is the last such vertex in $L$.
The vertex $u$ must be weakly $\Tc$-reachable from $v_S^1$ since $u$ is distance at most $\ell + 1 < c$ away from $v_S^1$ and all other vertices on the path must be ordered before $u$ by our selection of $u$.
Let $w$ be the first vertex in $L$ from $D_S^1$.
Then $|\wreach{\Tc}{\weg{x}}{L}{w}| \geq \Tt + 1$ since it can reach $\Tt - 2$ vertices in $D_S^1$ as well as $v_S^1$ and $v_S^2$ by Lemma~\ref{lemma:vs2-last}, and also $u$.
Therefore, no such $u$ can exist.
\end{proof}

It now suffices to consider whether $D_x$ and $p_{S,x}^1 \dots p_{S,x}^\ell$ can be ordered without exceeding $\Tt$ weakly $\Tc$-reachable vertices from any vertex.
The following lemma establishes that there is no way to do this.

\begin{lemma} \label{lemma:wcol-gx}
Given an instance $\SC(\mcU, \mcF)$, let $G$ be the graph $f(\mcU, \mcF)$ (Definition~\ref{definition:wcn-f}), and let $\weg{x}$ be an induced subgraph of $G$ as in Definition~\ref{definition:wcn-gx} for some $x \in \mcU$.
Then $\wcol{\Tc}{\weg{x}} > \Tt$.
\end{lemma}

\begin{proof}
Assume $\wcol{\Tc}{\weg{x}} \leq \Tt$, and let the ordering $L$ be a witness.
Since $L$ must satisfy $\wscore{\Tc}{\weg{x}}{L} \leq \Tt$, we use Lemmas~\ref{lemma:wcn-cliques-first},~\ref{lemma:vs2-last}, and~\ref{lemma:vs1-next}.
Consider $u$, the first vertex in $L$ from $D_x$.
If $u$ comes before $p_{S,x}^1$ for all $S$, then $|\wreach{\Tc}{\weg{x}}{L}{u}| \geq \Tt + 1$ since it can reach $\Tt - 3 \cdot f_x + 1$ vertices in $D_x$, $f_x$ in each of $p_{S,x}^1$, $f_x$ in each of $v_S^1$, and $f_x$ in each of $v_S^2$.
Thus, $w$, the first $p_{S,x}^1$ in $L$, must come before $u$ in $L$.
Let $d$ be the first vertex from $D_w$.
Then, $|\wreach{\Tc}{\weg{x}}{L}{d}| \geq \Tt + 1$ since $d$ can reach at least $f_x - 1$ vertices in $D_w$, $w$ itself, $\Tt - 3 \cdot f_x + 1$ vertices in $D_x$, $f_x$ in each of $v_S^1$, and $f_x$ in each of $v_S^2$.
Thus, $L$ cannot exist and $\wcol{\Tc}{\weg{x}} > \Tt$.
\end{proof}

\begin{lemma} \label{lemma:wcol-g}
Suppose $G$ is a graph constructed using the function $f$ defined in Definition~\ref{definition:wcn-f} applied to the \SCfull instance $(\mcU, \mcF)$.
Then $\wcol{\Tc}{G} > \Tt$.
\end{lemma}

\begin{proof}
Lemma~\ref{lemma:wcol-g} follows directly from Lemma~\ref{lemma:wcol-gx} and Corollary~\ref{corollary:wcol-hereditary}.
\end{proof}

Now, we show that a solution to the \SCfull instance $(\mcU, \mcF)$ can be recovered from an edit set to the \bWCN instance $G = f(\mcU, \mcF)$ using the function $g$.
We do this by first constructing a canonical solution and then converting to a \SCfull solution using a bijective mapping.
The definition of a canonical solution depends on the edit operation.

\begin{definition} \label{definition:wcn-canonical}
A solution $\yB$ to \bWCNV is \emph{canonical} if $\yB \subseteq \{v_S^1:S \in \mcF\}$.
A solution $\yB$ to \bWCNE or \bWCNC is canonical if $\yB \subseteq \{e = (v_S^1, v_S^2):S \in \mcF\}$.
\end{definition}

Canonical solutions all use the same ordering $L$ to exhibit bounded weak $\Tc$-coloring number.
We define that ordering here for convenience in later proofs, starting with the order on a single element gadget.

\begin{lemma} \label{lemma:wcn-canonical-gx}
Given an instance $\SC(\mcU, \mcF)$, let $G$ be the graph $f(\mcU, \mcF)$ (Definition~\ref{definition:wcn-f}), and let $\weg{x}$ be an induced subgraph of $G$ as in Definition~\ref{definition:wcn-gx} for some $x \in \mcU$.
Let $\yCanon$ be a canonical solution to \bWCN (Definition~\ref{definition:wcn-canonical}).
Suppose that $x$ is contained in the sets $S_i \dots S_j$.
Let $L$ be the following ordering:
$$D_{S_i}^1 \dots D_{S_j}^1, D_{S_i}^2 \dots D_{S_j}^2, D_{S_i, x} \dots D_{S_j, x}, D_x, p_{S_i, x}^\ell \dots p_{S_j, x}^\ell, \dots, p_{S_i, x}^1 \dots p_{S_j, x}^1, v_{S_i}^1 \dots v_{S_j}^1, v_{S_i}^2 \dots v_{S_j}^2.$$
Then $\wscore{\Tc}{\weg{x} \setminus \yCanon}{L} \leq \Tt$.
\end{lemma}

\begin{proof}
Note that $\weg{x} \cap \yB$ must be non-empty for any solution $\yB$.
Otherwise, by Lemma~\ref{lemma:wcol-gx}, $\wcol{\Tc}{\weg{x} \setminus \yB} > \Tt$, and so by Corollary~\ref{corollary:wcol-hereditary}, $\wcol{\Tc}{G \setminus \yB} > \Tt$, a contradiction.

From left to right, we show that $|\wreach{\Tc}{\weg{x} \setminus \yCanon}{L}{u}| \leq \Tt$ for each vertex $u$.
There is some $u \in D_{S_i}^1$ that can reach $\Tt - 2$ other vertices in $D_{S_i}^1$.
However, since the path to all other vertices in $\weg{x}$ goes through $v_{S_i}^1$, only $v_{S_i}^1$ and $v_{S_i}^2$ will be reachable from $u$.
Note here that no $v_{S_a}^1$ can reach any other $v_{S_b}^1$ since the shortest path between them has length $2 \cdot \ell + 2 \geq \Tc + 1$.
This also precludes $u$ from reaching any vertex $w$ in some other $D_S^1$.
Thus, $u$ can only reach $\Tt$ vertices.

Similar logic applies to $D_{S_i}^2$.
There is some $u \in D_{S_i}^2$ that can reach $\Tt - 1$ other vertices in $D_{S_i}^2$.
Additionally, $u$ can reach $v_{S_i}^2$ to bring the total to $\Tt$.
Only other $v_S^2$ vertices are reachable now due to the position of $v_{S_i}^2$.
However, again we observe that the shortest path between any two $v_S^2$ vertices is longer than $\Tc$.

There is also some $u \in D_{S_i, x}$ which can reach $f_x - 1$ other vertices.
Again, every path exiting $D_{S_i, x}$ uses $p_{S_i, x}^1$ and so only vertices after $p_{S_i, x}^1$ will be reachable from $u$.
In fact, $u$ can reach every vertex after $p_{S_i, x}^1$.
However, there are only $3 \cdot f_x$ such vertices.
Thus, $u$ can reach a total of $4 \cdot f_x - 1 < \Tt$ vertices.

Let $u$ be the first vertex from $D_x$ in the ordering $L$.
Then $u$ can reach $\Tt - 3 \cdot f_x + 1$ other vertices in $D_x$.
It can also reach the first vertex on each of the $f_x$ paths.
Because $p_{S,x}^1$ comes after $p_{S,x}^2 \dots p_{S,x}^\ell$, $u$ cannot reach any other path vertex.
Additionally, $u$ will be able to reach every $v_S^1$ and every $v_S^2$ assuming that $\Tc \geq 3$.
Superficially, this seems to result in $|\wreach{\Tc}{\weg{x} \setminus \yB}{L}{u}| = \Tt + 1$.
However, since $\weg{x} \cap \yCanon$ must be non-empty, this cannot be the case.
In the case of vertex deletion, at least one $v_S^1$ must be in $\yCanon$, and so there are only $(f_x - 1)$ $v_S^1$ vertices that are weakly $\Tc$-reachable.
Likewise, in the case of edge deletion and contraction, at least one edge connecting $v_S^1$ to $v_S^2$ must be in $\yCanon$ and so there are only $(f_x - 1)$ $v_S^2$ vertices that are weakly $\Tc$-reachable.
Note that $D_S^2$ contains $v_S^1$ after the edge is contracted, but this does not affect the number of vertices weakly $\Tc$-reachable from inside $D_S^2$.

Consider $u = p_{S_i, x}^\ell$ next.
The vertex $u$ can reach at most every $v_S^1$, $v_S^2$, and $p_{S,x}^1$ for a total of $3 \cdot f_x$.
Additionally, it can reach every vertex on its own path for another $\ell - 2$ vertices (since it cannot count itself and we already counted $p_{S_i, x}^1$).
The vertex $u$ cannot reach any vertex on another path since the unique path between them in $\weg{x}$ goes through both $p_{S,x}^1$ vertices.
Thus, $u$ can reach $3 \cdot f_x + \ell - 2 \leq \Tt$ vertices.

Lastly, we note again that there are at most $\Tt$ vertices after $p_{S_i, x}^1$.
Therefore, no vertex appearing after $p_{S_i, x}^1$ can reach more than $\Tt$ other vertices by definition.
\end{proof}

Since the same analysis applies to every element gadget, this shows that $\wcol{\Tc}{\weg{x} \setminus \yCanon} \leq \Tt$ for all $x \in \mcU$.
We use this result to produce an ordering for the vertices of the entire graph.

\begin{lemma} \label{lemma:wcn-canonical-g}
Given an instance $\SC(\mcU, \mcF)$, let $G$ be the graph $f(\mcU, \mcF)$ (Definition~\ref{definition:wcn-f}).
Let $\yCanon$ be a canonical solution to \bWCN (Definition~\ref{definition:wcn-canonical}).
Let $L_x$ be the ordering for $\weg{x}$ given by Lemma~\ref{lemma:wcn-canonical-gx} for each $x \in \mcU$, and let $L'_x$ be the ordering of the vertices produced from $L_x$ by excluding the vertices that are not unique to a single $\weg{x}$ subgraph (specifically $v_S^1$, $v_S^2$, $D_S^1$, and $D_S^2$).
Let $L$ be the following ordering:
$ D_{S_1}^1 \dots D_{S_m}^1, D_{S_1}^2 \dots D_{S_m}^2, L'_{x_1} \dots L'_{x_n}, v_{S_1}^1 \dots v_{S_m}^1, v_{S_1}^2 \dots v_{S_m}^2 $.
Then $\wscore{\Tc}{G \setminus \yCanon}{L} \leq \Tt$.
\end{lemma}

\begin{proof}
We refer to Lemma~\ref{lemma:wcn-canonical-gx} for the analysis of vertices in $D_S^1$ or $D_S^2$.
The same analysis applies since again it is true that the shortest path from $v_{S_a}^1$ to $v_{S_b}^2$ is longer than $\Tc$.
Thus, we need only show that $\wreach{\Tc}{G \setminus \yCanon}{L}{u} = \wreach{\Tc}{\weg{x} \setminus \yCanon}{L_x}{u}$ for every $u \in L'_x$ for every $x \in \mcU$.
By Lemma~\ref{lemma:wcn-canonical-gx}, we know that $|\wreach{\Tc}{\weg{x} \setminus \yCanon}{L_x}{u}| \leq \Tt$, and so this shows that $\wscore{\Tc}{G \setminus \yCanon}{L} \leq \Tt$.

Take $u$ from $L'_i$ and $w$ from $L'_j$.
All paths from $u$ to $w$ must pass through at least one $v_{S_i}^1$.
Note that in order for the path to have length less than $\Tc$, it can pass through at most one.
However, since $v_{S_i}^1$ comes after both $u$ and $w$ in $L$, neither can reach the other.
Thus, the reachability of $u$ in $\weg{x}$ with respect to $L_x$ is identical to its reachability in $G$ with respect to $L$.\looseness=-1
\end{proof}

In order to define the function $g$, we must first show that any solution $\yB$ to \bWCN can be converted to a canonical solution $\yCanon$.

\begin{lemma} \label{lemma:wcn-convert-canon}
Given an instance $\SC(\mcU, \mcF)$, let $G$ be the graph $f(\mcU, \mcF)$ (Definition~\ref{definition:wcn-f}), and let $\yB$ be an edit set consisting of either vertex deletions, edge deletions, or edge contractions such that $\wcol{\Tc}{G \setminus \yB} \leq \Tt$.
There exists a canonical solution $\yCanon$ such that $\cost(\yCanon) \leq \cost(\yB)$.
\end{lemma}

\begin{proof}
Construct $\yCanon$ in the following manner.
For \bWCNV, set $\yCanon = \yB \cap \{v_S^1 : S \in \mcF\}$.
Additionally, if $\yB$ contains $v_S^2$ but not $v_S^1$ for some $S \in \mcF$, then add $v_S^1$ to $\yCanon$ as well.
For \bWCNE and \bWCNC, set $\yCanon = \yB \cap \{e = (v_S^1, v_S^2) : S \in \mcF\}$.
Now, we consider $\weg{x}$ for each $x \in \mcU$.
By Lemma~\ref{lemma:wcol-hereditary} and the fact that $\yB$ is a solution, $\wcol{\Tc}{\weg{x} \setminus \yB} \leq \Tt$.
In conjunction with Lemma~\ref{lemma:wcol-gx}, this shows that $\yB \cap \weg{x}$ must be non-empty.
If $\yCanon \cap \weg{x}$ is empty, then we add any vertex from $\{v_S^1 : x \in S\}$ or any edge from $\{e = (v_S^1, v_S^2) : x \in S\}$ to $\yCanon$ depending on the problem.
The intersection of $\yCanon$ and $\weg{x}$ can only be empty when $\yB$ does not contain any edits in $\{v_S^1 : x \in S\}$ and $\{v_S^2 : x \in S\}$ or $\{e = (v_S^1, v_S^2) : x \in S\}$ depending on the edit operation.
Thus, $\yB \cap \weg{x}$ must contain an edit that is not added to $\yCanon$, and so $\yCanon$ must be smaller than $\yB$.
\end{proof}

\mypar{Mapping a graph edit set to a \SCfull solution.}

\begin{definition} \label{definition:wcn-g}
Given an instance $\SC(\mcU, \mcF)$, let $G$ be the graph $f(\mcU, \mcF)$ (Definition~\ref{definition:wcn-f}).
Suppose $\yB$ is a solution to the \bWCN instance $G$.
The function \emph{$g$} converts $\yB$ to a canonical solution $\yCanon$ and then produces either the set $\{S : v_S^1 \in \yCanon\}$ for an instance of \bWCNV or the set $\{S : e = (v_S^1, v_S^2) \in \yCanon\}$ for an instance of \bWCNE or \bWCNC using a bijective mapping between canonical edit sets and set covers.
\end{definition}

Clearly, these sets can be computed in polynomial time from a canonical solution.
Furthermore, the canonical solution can be computed in polynomial time as described in Lemma~\ref{lemma:wcn-convert-canon}, and so $g$ is polynomial time computable.
In order to prove that the functions $f$ and $g$ satisfy the necessary requirements for a strict reduction, we must show that $g$ produces valid set covers.

\begin{lemma} \label{lemma:wcn-g-valid}
Given an instance $\SC(\mcU, \mcF)$, let $G$ be the graph $f(\mcU, \mcF)$ (Definition~\ref{definition:wcn-f}), and let $\yB$ be a solution to the \bWCN instance $G$.
Then $g(\yB)$ is a solution to $\SC(\mcU, \mcF)$.
\end{lemma}

\begin{proof}
Suppose $\yCanon$ is a canonical edit set found by applying Lemma~\ref{lemma:wcn-convert-canon} to $\yB$.
Since each edit in $\yCanon$ corresponds to a chosen set in $g(\yCanon)$, the edit in $\yCanon \cap \weg{x}$ corresponds to a set $S \in \mcF$ that covers $x$.
By Lemma~\ref{lemma:wcol-hereditary} and the fact that $\yCanon$ is a solution for \bWCN on $G$, we know that $\yCanon \cap \weg{x}$ is non-empty for all $x \in \mcU$.
Thus, for all $x \in \mcU, \exists S \in \mcF: S \in g(\yCanon)$.
\end{proof}

\mypar{Strict reduction from \SCfull.}
\begin{lemma} \label{lemma:wcn-cost}
Given an instance $\SC(\mcU, \mcF)$, let $G=f(\mcU, \mcF)$ (Definition~\ref{definition:wcn-f}).
Then for any solution $\yB$ to \bWCN on $G$, $\cost(\yB) \geq \cost(\yCanon) = \cost(g(\yCanon))$, where $\yCanon$ is the canonical solution corresponding to $\yB$ (Definition~\ref{definition:wcn-canonical}).
\end{lemma}

\begin{proof}
Lemma~\ref{lemma:wcn-convert-canon} shows that $\cost(\yB) \geq \cost(\yCanon)$.  By Definition~\ref{definition:wcn-g}, $\cost(\yCanon) = \cost(g(\yCanon))$.
\end{proof}

\begin{lemma} \label{lemma:wcn-opt}
Given an instance $\SC(\mcU, \mcF)$, let $G$ be the graph $f(\mcU, \mcF)$ (Definition~\ref{definition:wcn-f}), $\opt_\bWCN$ be the weight of the optimal solution to the $\bWCN$ instance $G$, and let $\opt_\SC$ be the weight of the optimal solution to $\SC(\mcU, \mcF)$.
Then $\opt_\bWCN \leq \opt_\SC$.
\end{lemma}

\begin{proof}
Let $\yOpt$ be an optimal solution to the $\SC(\mcU, \mcF)$.
Using the bijective mapping $g$ described in Definition~\ref{definition:wcn-g}, we can find a canonical edit set $\yCanon$ from $\yOpt$.
By Lemmas~\ref{lemma:wcol-hereditary} and~\ref{lemma:wcn-canonical-g}, we know that $\yCanon$ is a solution to the \bWCN instance $G$ if and only if $\yCanon \cap \weg{x}$ is non-empty for all $x \in \mcU$,
and this is implied by $\yOpt$ being a solution to the $\SC(\mcU, \mcF)$.
By Definition~\ref{definition:wcn-g}, $\cost(\yCanon) = \opt_\SC$, and furthermore, $OPT_\bWCN \leq \cost(\yB)$ for any potential solution $\yB$ including $\yCanon$.
\end{proof}

\begin{corollary} \label{corollary:wcn-opt-equal}
Given an instance $\SC(\mcU, \mcF)$, let $G$ be the graph $f(\mcU, \mcF)$ (Definition~\ref{definition:wcn-f}), $\opt_\bWCN$ be the weight of the optimal solution to the $\bWCN$ instance $G$, and let $\opt_\SC$ be the weight of the optimal solution to $\SC(\mcU, \mcF)$.
Then $\opt_\bWCN = \opt_\SC$.
\end{corollary}

\begin{proof}
Using the notation of the previous proof, assume that there is another solution $\yB$ for the \bWCN instance $G$ so that $\cost(\yB) < \cost(\yCanon)$.
Then by Lemmas~\ref{lemma:wcn-cost} and~\ref{lemma:wcn-opt}, $\cost(g(\yB)) \leq \cost(\yB) < \cost(\yCanon) \leq \cost(\yOpt)$.
This is a contradiction to the assumption that $\yOpt$ was an optimal solution to the \SCfull instance $(\mcU, \mcF)$, and so $\yCanon$ must be optimal.
\end{proof}

\begin{proof}[Proof of Theorem~\ref{theorem:bwcne-hard}]
  Adapting the proof of Lemma~\ref{lemma:wcn-f-poly} to the case that the instance $\SC(\mcU,\mcF)$ is actually an instance of $\xSC{k}$, i.e. every element $x \in \mcU$ has the same frequency $f_x = k$, we get that the size of the graph in our reduction is
  \begin{align}
    n &= |\mcF|(2\Tt+1) + \sum_{x \in \mcU} \left( f_x^2 + (\ell-3)f_x + \Tt+2 \right)  \nonumber \\
      &= (2\Tt+1)|\mcF| + |\mcU|(k^2 + (\ell-3)k + \Tt + 2) \label{eqn:wcn-graph-size-usc} \\
      &= \Theta(|\mcU| \cdot (k^2 + \Tt)) \label{eqn:simplified-wcn-graph-size},
  \end{align}
  where the last inequality follows from Theorem~\ref{lem-set-cover-frequency-superconstant} on instances where $|\mcF| \leq |\mcU|$.
  Hence, the functions $f$ and $g$ run in polynomial time, and so by Lemma~\ref{lemma:wcn-cost} and Corollary~\ref{corollary:wcn-opt-equal}, we have demonstrated a strict reduction from $\xSC{k}$.

  Note that in the graphs we construct, we require $\Tt \geq \max\{4\cdot f_{\max}, \ell + 3\cdot f_{\max}-2\}$, and in particular for a fixed $\Tc$ (and therefore fixed $\ell = \floor{\Tc/2}$) and an instance of \xSC{k} we can set $\Tt = 4\cdot k$.
  Thus, an $o(\Tt)$ approximation for $\bWCN$ would yield an $o(k)$ approximation for $\SC$.
  Hence, by Theorem~\ref{lem-set-cover-frequency} it is NP-hard to approximate $\bWCN$ within a factor of $o(\Tt)$.
  Note this specifically applies for $k$ constant with respect to $|\mcU|$ and therefore $\Tt = 4k$ constant with respect to $n$, by Equation~\eqref{eqn:wcn-graph-size-usc}.
  Furthermore, we note that because Theorem~\ref{lem-set-cover-frequency} requires $k \geq 3$ and our reduction requires $\Tt \geq 4\cdot f_{\max}$, $\Tt \geq \ell + 3\cdot f_{\max}-2$, our reduction applies only when $\Tt \geq \max\{12, \ell + 7\}$.

  Finally, using Theorem~\ref{lem-set-cover-frequency-superconstant} we can obtain a similar result for super-constant values of $\Tt$
  assuming NP $\nsubseteq$ DTIME($n^{O(\log \log n)})$).
  Since $\Tt = 4k$ in our reduction from $\xSC{k}$ and $|\mcU| \leq n$ by Equation~\eqref{eqn:wcn-graph-size-usc},
  by Theorem~\ref{lem-set-cover-frequency-superconstant} there is a constant $b>0$ such that
  there is no $o(\Tt)$-approximation for any $\Tt/4  = k \leq (\log |\mcU|)^{1/b} \leq (\log n)^{1/b}$.
\end{proof}

\mypar{Handling the case $\Tc = 2$.}

We now show inapproximability results for \bWCNfull when $\Tc = 2$ by using a modification of the reduction described for $\Tc > 2$.

\begin{theorem} \label{theorem:bwcne-hard-v1}
  For $\Tc = 2$ and constant $\Tt \geq 9$, \bWCNVfull is $o(\Tt)$-inapproximable.
  Furthermore, there exists a constant $b > 0$ such that
  \bWCNV is $o(\Tt)$-inapproximable for $\Tt \in  O( \log^{1/b} n )$,
  unless NP $\subseteq$ DTIME$(n^{O(\log\log n)})$.
\end{theorem}
\begin{proof}
    For $\Problem{\Tt-BWE-V-2}$, we change the graph construction as follows.
    We replace each set gadget with a simpler construction in which $D_S^2$ is removed from $\wsg{S}$ for all $S \in \mcF$.
    We set $|D_S^1| = \Tt + 1$, $|D_x| = \Tt - 2 \cdot f_x + 2$, and let $\Tt \geq 3 \cdot f_{\max}$.
    Note that since $\Tc=2$ we have $\ell = 1$, and so $\Tt$ satisfies $\Tt \geq 3 \cdot f_{max}$ and $\Tt \geq \ell + 2 \cdot f_{max} - 2$.

    The result of these modifications is that the graph $G$ produced from $\xSC{k}(\mcU, \mcF)$ has size equal to
    \begin{align}
        n &= |\mcF|(\Tt+1) + \sum_{x \in \mcU} \left( f_x^2 - f_x + \Tt + 2 \right)  \nonumber \\
          &= (\Tt+1)|\mcF| + |\mcU|(k^2 - k + \Tt + 2) \label{eqn:wcnv-graph-size-usc} \\
          &= \Theta(|\mcU| \cdot k^2) \label{eqn:simplified-wcnv-graph-size},
    \end{align}
    where the last equality follows from the assumption that $|\mcF| \leq |\mcU|$ given in Theorem~\ref{lem-set-cover-frequency-superconstant}.
    Hence, the functions $f$ and $g$ can be computed in polynomial time, and so by Lemma~\ref{lemma:wcn-cost} and Corollary~\ref{corollary:wcn-opt-equal}, we have again demonstrated a strict reduction from $\xSC{k}$.

    Note that in the graphs we construct, we require $\Tt \geq 3\cdot f_{\max}$, so for a fixed instance of $\xSC{k}$ we can set $\Tt = 3\cdot k$.
    Thus, an $o(\Tt)$ approximation for $\bWCNV$ when $\Tc = 2$ would yield an $o(k)$ approximation for $\xSC{k}$.
    Hence, by Theorem~\ref{lem-set-cover-frequency} it is NP-hard to approximate $\bWCNV$ within a factor of $o(\Tt)$ when $\Tc = 2$.
    Note this specifically applies for $k$ constant with respect to $|\mcU|$ and therefore $\Tt = 3k$ constant with respect to $n$, by Equation~\eqref{eqn:wcn-graph-size-usc}.
    We note that because Theorem~\ref{lem-set-cover-frequency} requires $k \geq 3$ and our reduction requires $\Tt \geq 3\cdot f_{\max}$, this reduction applies only when $\Tt \geq 9$.

    Finally, using Theorem~\ref{lem-set-cover-frequency-superconstant} we can obtain a similar result for super-constant values of $\Tt$
    assuming NP $\nsubseteq$ DTIME($n^{O(\log \log n)})$).
    Since $\Tt = 3k$ in our reduction from $\xSC{k}$ and $|\mcU| \leq n$ by Equation~\eqref{eqn:wcnv-graph-size-usc},
    by Theorem~\ref{lem-set-cover-frequency-superconstant} there is a constant $b>0$ such that
    there is no $o(\Tt)$-approximation for any $\Tt/3  = k \leq (\log |\mcU|)^{1/b} \leq (\log n)^{1/b}$.

\end{proof}

\begin{theorem} \label{theorem:bwcne-hard-v2}
  \bWCNE and \bWCNC are $o(\log n)$-inapproximable when $\Tc = 2$ and $\Tt \in \Omega(n^{1/2})$.
\end{theorem}

\begin{proof}
    For edge deletion and contraction with $\Tc=2$, we change the graph construction as follows.
    From each element gadget $\weg{x}$, we remove both the path $p_{S,x}^j$ and the clique $D_{S,x}$ for every set $S$ that contains element $x$.
    To reconnect the graph, we add edges from each vertex in $D_x$ directly to $v_S^1$ in each set gadget that $\weg{x}$ contains.
    Additionally, we set $|D_x| = \Tt - 2 \cdot f_x + 2$, then set $\Tt$ to any integer $\geq \max\{ 2 \cdot f_{max}, |\mcF|\}$.
    The constraint $\Tt \geq 2 \cdot f_{max} $ guarantees that there is a vertex in each $D_x$, since we set $|D_x| = \Tt - 2 \cdot f_x + 2$.

    The bound $\Tt \geq |\mcF|$ is more complicated to justify.
    Note that, due to the removal of $p_{S, x}^1$, the length of the path between $v_{S_i}^1$ and $v_{S_j}^1$ for each pair $S_i, S_j \in \mcF$ is no longer greater than $\Tc$.
    To prevent any $v_S^1$ from being able to reach too many other nodes, we pick an ordering and use it to determine the size of $D_S^1$ and $D_S^2$ for each $S \in \mcF$.
    For each set $S$, let $N_S \subseteq \mcF$ be the collection of sets that share an element with $S$ and appear after $S$ in the ordering.
    Note that $v_T^1$ is reachable from $v_S^1$ for all $T \in N_S$.
    Thus, we set $|D_S^1| = \Tt - |N_S|$ and $|D_S^2| = \Tt - |N_S| + 1$.
    Since $|N_S|$ could be as large as $|\mcF| - 1$, we require $\Tt \geq |\mcF|$ so that $|D_S^1| \geq 1$ for all $S \in \mcF$.

    Next we derive the size $n$ of the graph.
    For each element $x$, the graph contains a clique $D_x$ of size $\Tt - 2 \cdot f_x + 2$.
    For each set $S$, the set gadget $\wsg{S}$ has $|D_S^1| + |D_S^2| \leq 1 + 2(\Tt - |N_S|)$ nodes, which is at most $2\Tt$.
    Thus, we have
    \begin{align*}
        n &= \sum_{x\in \mcU} |D_x| + \sum_{S \in \mcF} |\wsg{S}| \\
        &\leq \sum_{x\in \mcU} (\Tt - 2 \cdot f_x + 2) + \sum_{S \in \mcF} 2\Tt,
    \end{align*}
    which is bounded above by $\Tt|\mcU| + 2\Tt|\mcF|$, since $(\Tt - 2 \cdot f_x + 2) \geq \Tt$.
    Finally, since we can choose any $\Tt \geq \max\{ 2 \cdot f_{max}, |\mcF|\}$, we set $\Tt = 2 |\mcF|$ so that $n = O( 2|\mcF||\mcU| + 4|\mcF|^2)$.

    By Theorem~\ref{lem:set-cover-general}, it is NP-hard to approximate $\SC(\mcU,\mcF)$ within a factor of $(1-\epsilon)(\ln |\mcU|)$ for any $\epsilon > 0$, even for instances where $|\mcF|$ is bounded by a polynomial in $|\mcU|$, so we can assume $|\mcF| \leq |\mcU|^s$ for some constant $s \in \mathbb{N}^+$.
    Thus, by Lemma~\ref{lemma:wcn-f-poly}, the graph $G$ has $O(|\mcU|^{2s})$ vertices, and so by the $(1-\epsilon)(\ln |\mcU|)$-inapproximability of \SC,
    $(\log n)$-inapproximability of \bWCN when $\Tc = 2$ follows from these alternative constructions and the analysis above.
    However, we remark that in our reduction all instances satisfy $n \geq \Tt \geq 2|\mcF| = \Theta(|\mcU|^{s}) = \Theta(n^{1/2})$, and so we can only conclude inapproximability of $\bWCNE$ and $\bWCNC$ on instances where $\Tt \in \Omega(n^{1/2})$.

\end{proof}

\fi

\ifdefined\istwhard
    \subsection{Treewidth and Clique Number: $o(\log n)$-Inapproximability of Vertex Editing}
    \label{section:negative_treewidth_and_clique}
    In this section, we address the hardness of \bTWVfull (\bTWV) and \bCNVfull (\bCNV).
In the positive direction, recent work provided an approximation scheme for \bTWVfull and related problems where the approximation ratio depends only on the target treewidth, $\Tw$, and not on $\tw(G)$~\cite{fiorini2010hitting,fomin2012planar}.
Our results here provide the sharpest-known \emph{lower} bound on these approximation ratios.

\begin{theorem}\label{thm:treewidth-hardness}
    \bTWVfull (\bCNVfull) is $o(\log n)$-inapproximable when $\Tw = \Omega(n^\delta)$ ($\Tcn = \Omega(n^\delta)$) for $\delta \geq 1/2$.
\end{theorem}

Because it is NP-hard to decide whether a graph has treewidth $\leq \Tw$ (or clique number $\leq \Tcn$), it is trivially hard to guarantee a finite approximation ratio for editing to the class.
By contrast, our hardness-of-approximation results apply when editing graphs that are far from the desired class, which is the most favorable scenario for approximation.

\mypar{Reduction strategy.}
Our proof relies on a strict reduction from \SCfull.
We construct a graph $G$ that encodes an instance $\SC(\mcU,\mcF)$ such that a solution to \bTWV (or \bCNV) on $G$ can be mapped to a solution for $\SC(\mcU,\mcF)$.
Intuitively, $G$ consists of a set of outer cliques that overlap with parts of a single ``central'' clique.
For a given instance $\SC(\mcU,\mcF)$, we map each element in $\mcU$ to an outer clique and each set in $\mcF$ to a vertex in the central clique.
We prove that reducing either treewidth or clique number in $G$ requires deleting vertices in the central clique, and that the deleted vertices correspond to sets in $\mcF$ that form an optimal cover of $\mcU$.
We proceed by introducing the gadgets used in this reduction.

\mypar{Mapping an instance of \SCfull to an instance of vertex deletion problems.}

\begin{definition}\label{def:sc-tw-gadget}
    A \emph{set cover gadget} (\scg) for $\SC(\mcU, \mcF)$ is a graph $G = (V,E)$ that consists of overlapping cliques constructed as follows.
    For each set $S \in \mcF$ there is a corresponding vertex $v_S \in V(G)$.
    We call these vertices \emph{central vertices}, and they are all connected together to form the \emph{central clique}.
    For each element $x \in \mcU$ $G$ contains an \emph{outer clique} consisting of $|\mcF|-f_x$ \emph{dummy vertices} as well as the $f_x$ central vertices corresponding to the sets that contain the element $x$.
    We assume that no element is in every set $S \in \mcF$ and that each set contains at least one element.
\end{definition}
For a visualization of Definition~\ref{def:sc-tw-gadget}, see Figure~\ref{fig:tw-sc-graph}.
This construction gives us a function $f$ that maps an instance $(\mcU,\mcF)$ of \SC to an instance $f(\mcU,\mcF)$ of either \bTWV or \bCNV.
In the instance $f(\mcU,\mcF)$ of $\bTWV$ we set the target treewidth $\Tw = |\mcF|-2$, and for $\bCNV$ we set $\Tcn = |\mcF|-1$. As we will see, this corresponds to setting the target parameter so that solving $\SC$ corresponds to reducing treewdith (clique number) by 1.

\begin{figure}[!ht]
    \centering
    \includegraphics[width=\textwidth]{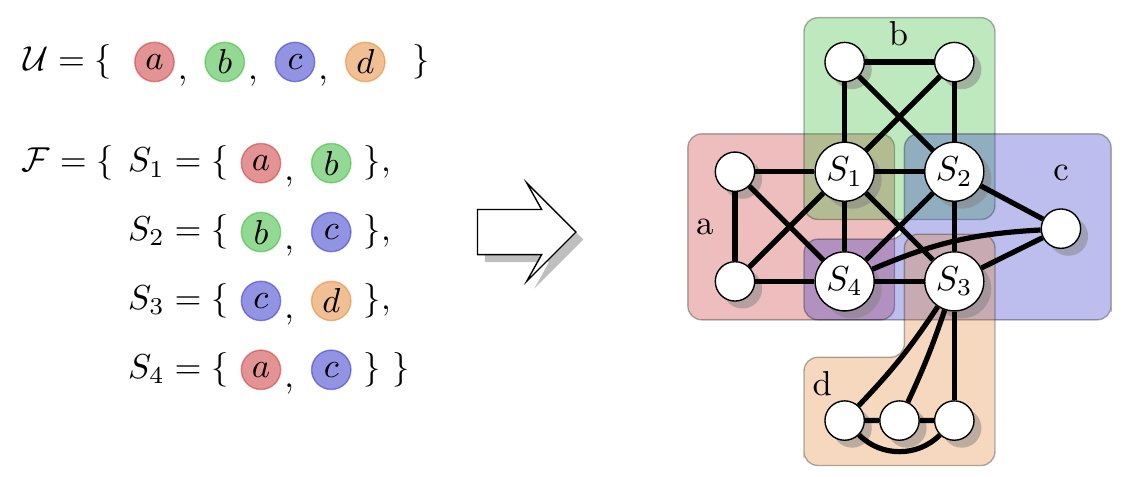}
    \caption{\label{fig:tw-sc-graph}
        (\emph{Left}.) An instance of \SCfull.
        (\emph{Right}.) Our encoding of the \SC instance $(\mcU, \mcF)$ as a graph.
        Each set $S_j \in \mcF$ is mapped to a single vertex in the ``central clique''.
        Each element is mapped to an ``outer clique'' (green, blue, orange, red) that contains each central vertex $v_{S}$ corresponding to a set $S$ that contains $x$ plus
        $|\mcF|-f_x$ ``outer vertices'', where $f_x$ is the number of sets containing $x$.
        The unique optimal solution to the \SC instance above is $\{S_1, S_3\}$.
        Note that the corresponding vertex deletion set consisting of the $S_1$ and $S_3$ central clique vertices is also the unique optimal solution to both 2-\TWV and 3-\CNV.
    }
\end{figure}

\begin{lemma}\label{lemma:set-cover-graph-tw}
    Let $(\mcU, \mcF)$ be an instance of \SCfull
    such that no set $S \in \mcF$ is empty or contains all of $\mcU$,
    and no element $x \in \mcU$ is contained in every set.
    Then the corresponding \texttt{scg} $G$ from Definition~\ref{def:sc-tw-gadget} has $\tw(G) = |\mcF|-1$ and $\omega(G) = |\mcF|$.
\end{lemma}
\begin{proof}
    First, note that the central clique in the construction of the set cover gadget $G$ contains exactly $|\mcF|$ vertices,
    and so $\tw(G) \geq |\mcF|-1$ and $\omega(G) \geq |\mcF|$.
    To prove the corresponding upper bound for $\omega(G)$, we show no node is contained in a clique of size greater than $|\mcF|$. By assumption no set contains every element, and so the central clique is maximal---no outer vertex is connected to every vertex in the inner clique.
    On the other hand, all outer vertices have degree exactly $|\mcF|-1$ because they are connected to only vertices in their outer clique, constructed to be size $|\mcF|$.

    To prove $\tw(G) \leq |\mcF|-1$ we describe a tree decomposition $(\mathcal{Y}, T = (I,F) )$ of width $|\mcF|-1$.
    For each outer clique, $C_i$, there is a leaf node in the tree decomposition, $v(B_i) \in I$, such that the bag $B_i \in \mathcal{Y}$ contains all the vertices of $C_i$.
    Additionally, there is one node $v(B) \in I$ such that its bag $B \in \mathcal{Y}$ contains all central vertices of $G$.
    By construction, each outer clique has size $|\mcF|$, as does the central clique, proving that the width of this decomposition is in fact $|\mcF|-1$.
    To see the tree decomposition is valid, first note that, by construction, every vertex and edge of $G$ is contained in at least one bag in $\mathcal{Y}$.
    Furthermore, the tree $T = (I,F)$ is a star graph with the bag $B$ corresponding to the center of the star, $v(B)$.
    For any two pendant bags $B_i, B_j$, $B_i \cap B_j$ consists of central vertices.
    Since all central vertices are in $B$,
    we have that for $v \in G$, the nodes in $T$ containing $v$ form a connected subtree.
\end{proof}

\mypar{Mapping a graph deletion set to a \SCfull solution.}
Once a solution $\yB$ is found in \scg for \bTWV (\bCNV), we want to map $\yB$ to a solution $g(\yB)$ for \SC.
Here, we specify such a function $g$ by first describing how to convert an arbitrary solution to \bTWV (\bCNV) on $G$ to a canonical solution.

\begin{definition}\label{def:map-g-to-sc}
  Given an instance $(\mcU, \mcF)$ of \SC,
  let $G = \scg(\mcU,\mcF)$.
  Given a deletion set $\yB$ in $G$ for \bTWV or \bCNV, we map it to a solution $g(\yB)$ for \SC in two steps:
  \begin{enumerate}
    \itemsep0em
    \item Construct a \emph{canonical edit set} $\ycan$ as follows: for each $v_i \in \yB$ if $v_i$ is in the central clique of $G$, add $v_i$ to $\ycan$;
    otherwise, choose any central vertex $v(S_j)$ in the same outer clique as $v_i$, and add $v(S_j)$ to $\ycan$.
    \item Set $g(\yB) = \{ S_j \in \mcU \mid v(S_j) \in \ycan \}$.
  \end{enumerate}
\end{definition}
Note that such a central vertex $v(S_j)$ in Definition~\ref{def:map-g-to-sc} always exists because every element $x \in \mcU$ is contained in at least one set.
For our strict reduction to work we have to show that the canonical edit set $\ycan$  described above is still valid for \bTWV (\bCNV) and that the set $g(\yB)$ is a valid solution to \SC.
\begin{lemma}\label{lemma:canonical-tw-cn-valid}
    Let $(\mcU, \mcF)$ be an instance of \SCfull, $f(\mcU,\mcF)$ the corresponding instance of \bTWV (\bCNV), $\yB$ any solution to \bTWV (\bCNV) on $f(\mcU,\mcF)$, and $\ycan$ the canonical set described in Definition~\ref{def:map-g-to-sc}.
    Then $\ycan$ is a valid solution to \bTWV (\bCNV) and $|\ycan| \leq |\yB|$.
\end{lemma}
\begin{proof}
    To show that $\ycan$ is a valid edit set, i.e., that $\tw(G[V\setminus \ycan]) < \tw(G)$, we will show that the tree decomposition of $G$ presented in the proof of Lemma~\ref{lemma:set-cover-graph-tw} is a valid tree decomposition of $G[V \setminus \ycan]$ with smaller width.

    Since each bag in the tree decomposition induces a clique in $G$, we know that the initial solution $\yB$ necessarily contains at least one vertex from each bag---otherwise $G[V \setminus \yB]$ would contain a clique of size $|\mcF|$.
    Moreover, we know that for $\ycan$ to be a valid deletion set, it suffices to show that it contains at least one vertex from each bag because then the tree decomposition bags will all have size strictly less than $|\mcF|$.
    For any bag, $B_i$, we know $\yB$ contains at least one vertex from $B_i$, say $v_i$.
    If this vertex is also in the central bag $B$, then $v_i \in \ycan$.
    Otherwise, the construction of $\ycan$ selected a node $\hat{v}_i \in B_i \cap B$ to put in $\ycan$.
    Thus, $\ycan$ contains at least one vertex from each bag, and so is a valid solution for \bTWV.
    Note that this proof also shows that the solution $\ycan$ is a valid solution for \bCNV.

    Finally, to show that $|\ycan| \leq |\yB|$ observe that each element in $\yB$ corresponds to at least one element of $\ycan$.
\end{proof}
\begin{lemma}\label{lemma:canonical-sc-valid}
  In the notation of Lemma~\ref{lemma:canonical-tw-cn-valid}, $g(\ycan)$ is a valid solution to \SC, and $|g(\ycan)| \leq |\ycan|$.
\end{lemma}
\begin{proof}
    Given an instance $(\mcU, \mcF)$ of \SC, for any $z \in \mcU$ we will show that there is some set $S_j \in g(\yB)$ such that $z \in S_j$.
    By the construction in Definition~\ref{def:sc-tw-gadget}, there is an outer clique $B_z$ in the graph $G$ that corresponds to the element $z$.
    Any solution $\yB$ to \bTWV (\bCNV) necessarily must contain at least one vertex $v_i$ in the outer clique $B_z$, or else $G[V \setminus \yB]$ would still contain a clique of size $|\mcF|$.
    From the construction of the set $\ycan$, we know that $\ycan$ contains a vertex in $B_z$ that is also a central vertex, call it $v(S_j)$.
    Hence, by the set cover gadget construction, the set $S_j$ must contain the element $z$.
    Then $g(\yB)$ contains the set $S_j$, and so $g(\yB)$ is a valid set cover.
    Finally, since $g(\yB)$ has no more than one set for each vertex in $\ycan$, we have $|g(\yB)| \leq |\ycan|$.
\end{proof}

\mypar{Strict reduction from \SCfull.}\label{sec:tw-l-reduction}
To finish our proof of a strict reduction from \SC to \bTWV (\bCNV), we want to show that the functions $f$ and $g$ introduced earlier satisfy the sufficient conditions given in Section~\ref{sec:prelims-reductions}.
From Lemma~\ref{lemma:canonical-tw-cn-valid} and Lemma~\ref{lemma:canonical-sc-valid} we have the following.
\begin{corollary}\label{lemma:sc-tw-costs}
    In the notation of Lemma~\ref{lemma:canonical-tw-cn-valid},
    $\cost_{SC}(g(\yB)) \leq \cost_{TW,CN}(\yB)$.
\end{corollary}
\begin{lemma}\label{lemma:sc-tw-sizes}
    In the notation of Lemma~\ref{lemma:canonical-tw-cn-valid},
    $\opt_{TW,CN}(f(\mcU,\mcF)) = \opt_{SC}(\mcU,\mcF)$.
\end{lemma}
\begin{proof}
    Given a set cover gadget $G$ corresponding to instance $(\mcU,\mcF)$ of \SC, let $\yopt$ be any optimal solution on $G$ to \bTWV (\bCNV).
    Then by Lemma~\ref{lemma:canonical-tw-cn-valid},
    $\opt_{TW,CN}(f(\mcU,\mcF)) \geq |\yopt| \geq |\ycan| = |g(\ycan)| \geq \opt_{SC}(\mcU,\mcF))$.

    On the other hand, given any solution $\yA$ to $\SC(\mcU,\mcF)$, the set $\ycan = \{ v(S_j) \mid S_j \in \yA \}$ is a valid solution for \bTWV (\bCNV), and so we have $\opt_{TW,CN}(f(\mcU,\mcF)) \leq \opt_{SC}(\mcU,\mcF))$.
\end{proof}

\begin{proof}[Proof of Theorem~\ref{thm:treewidth-hardness}]
    The graph \scg consists of $|\mcF|$ central vertices, along with $\sum_{x \in \mcU} (|\mcF| - f_x)$ outer vertices.
    Hence,
    \[
      |V(G)| \leq |\mcF| + |\mcF||\mcU| - \sum_{x \in \mcU} f_x \leq |\mcF| + |\mcF||\mcU| - |\mcU|.
    \]
    By Theorem~\ref{lem:set-cover-general}, we may assume $|\mcF| \leq |\mcU|^s$ for some constant $s \in \mathbb{N}^+$.
    Thus, we have $n \in \Theta( |\mcU|^{s+1})$, and so the maps $f$ and $g$ in our reduction are polynomial time in the size of $\SC(\mcU,\mcF)$.
    Finally, the inequality in Corollary~\ref{lemma:sc-tw-costs} and equation in Lemma~\ref{lemma:sc-tw-sizes} provide the sufficient condition for a strict reduction.

    As noted above, in our reduction we have $n \in \Theta( |\mcU|^{s+1})$ for some constant $s>0$, and so $\log n \in \Theta( \ln |\mcU|)$.
    Hence, an $o(\log n)$-approximation for $\bTWV$ would give an $o(\ln |\mcU|)$-approximation for \SC.
    Finally, we note that in our reduction, $\Tw = |\mcF|-2$, which is $\Tw \in \Theta(|\mcU|^s) = \Theta( n^{s/(s+1)} )$, so we can only conclude inapproximability of $\bTWV$ on instances where $\Tw \in \Omega( n^{\delta})$ for some constant $\delta \geq 1/2$ (since  $s \geq 1$).
\end{proof}

\fi

\ifdefined\isBDVhard
    \subsection{Bounded Degree: $(\ln \Td - C\cdot\ln \ln \Td)$-Inapproximability of Vertex Editing}
    \label{section:negative_bounded_degree}
    
In this section we prove $(\ln \Td - C\cdot \ln \ln \Td)$-inapproximability of \bBDDVfull (\bBDDV), which
matches the best known upper bound for this problem by Ebenlendr~et~al.~\cite{ebenlendrapproximation}.
The result follows from a strict reduction from \SCfull (\SC).

\begin{theorem}
\label{theorem:hardness-degree}
    For each integer $\Td \geq 1$, there exists a constant $C > 0$ such that $\bBDDV$ is $(\ln \Td - C\cdot \ln\ln \Td))$-inapproximable.
\end{theorem}

\mypar{Reduction strategy.}
We encode an instance $\SC(\mcU,\mcF)$ as a bipartite graph $G$ in which one partition represents elements in $\mcU$, the other partition represents sets in $\mcF$, and an edge between a set vertex and an element vertex indicates containment.
Then we exhibit a mapping between solutions of the two problems that we show preserves solution sizes.
We proceed by introducing the gadgets used in this reduction.

\mypar{Mapping an instance of \SCfull to an instance of \bBDDVfull.}
\begin{definition}\label{def:sc-bdd-gadget}
    Given an instance $(\mcU,\mcF)$ of \SCfull, we construct a \emph{set cover graph}, \scdg, as follows.
    Let $\Delta = \max\{|S| : S \in \mcF \}$, let $f_e$ be the number of sets that element $e$ appears in, and let $f_{\max} = \max\{f_e : e \in \mcU\}$.
    Set $\Td=\max\{\Delta,f_{\max}\}$.
    Consider the bipartite representation of $\SC(\mcU, \mcF)$ in which
    each set in $\mcF$ and each element in $\mcU$ is represented by a vertex, and there is an edge between the corresponding vertices of an element $e$ and a set $S$ if $e\in S$.
    Then, we augment the degree of each element vertex (e.g.
    $v_e$ in Figure~\ref{figure:bounded-degree}) to $\Td+1$ by adding dummy vertices $x_{e,1},\cdots, x_{e,d+1-f_e}$.
    For each element $e$, we define $\scte$ to be the set of vertices containing $v_e$, the vertices corresponding to the sets that contain $e$, and the dummy vertices of $e$.
\end{definition}
This construction provides a map $f$ from an instance $\SC(\mcU,\mcF)$ to instance $G = f(\mcU,\mcF)$ of \bBDDV.
The component $\scte$ of the gadget is visualized in Figure~\ref{figure:bounded-degree}.

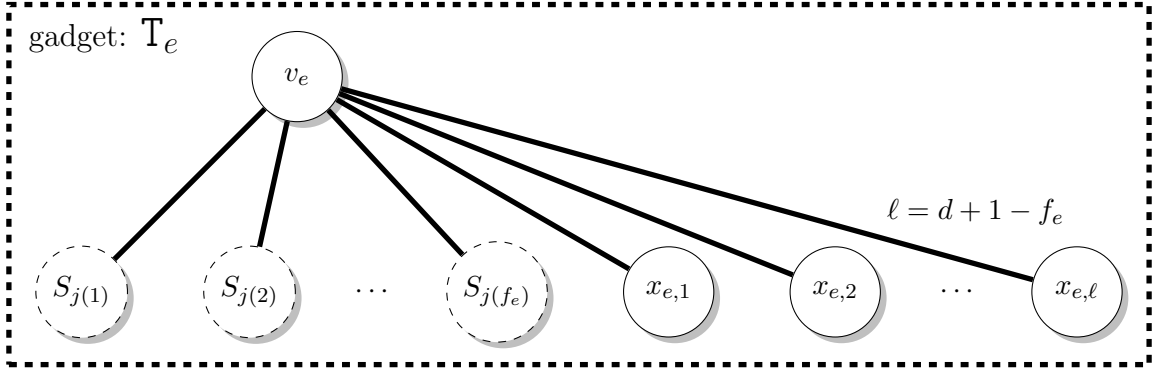
\begin{figure}[!h]
\centering
\begin{tikzpicture}
\tikzstyle{node_style} = [state, fill=white, drop shadow, minimum width=1.2cm]
\tikzstyle{outside_node_style} = [state, fill=white, drop shadow, minimum width=1cm, dashed]

\tikzstyle{edge_style} = [line width=2pt]
\tikzstyle{path_style} = [line width=2pt, decorate, decoration={snake}]
\tikzstyle{arrow_style} = [draw, fill=white, single arrow, single arrow head indent=1ex, minimum size=1cm, drop shadow]

\node (ve) [node_style] {$v_e$};
\node (s1) [outside_node_style, below left = 2cm and 2cm of ve] {$S_{j(1)}$};
\node (s2) [outside_node_style, right = 1cm of s1] {$S_{j(2)}$};
\node (se) [outside_node_style, right = 2cm of s2] {$S_{j(f_e)}$};
\node (x1) [node_style, right = 1cm of se] {$x_{e,1}$};
\node (x2) [node_style, right = 1cm of x1] {$x_{e,2}$};
\node (x_last) [node_style, right = 2cm of x2, label={[label distance=0.15cm]93:$\ell = d+1-f_e$}] {$x_{e,\ell}$};

label={[label distance=1cm]30:label}

\node (s_ellipses) [right = 0.65cm of s2] {$\cdots$};
\node (x_ellipses) [right = 0.65cm of x2] {$\cdots$};

\draw (ve) edge[edge_style] (s1);
\draw (ve) edge[edge_style] (s2);
\draw (ve) edge[edge_style] (se);

\draw (ve) edge[edge_style] (x1);
\draw (ve) edge[edge_style] (x2);
\draw (ve) edge[edge_style] (x_last);

\node (boxall) [fit=(s1)(ve)(x_last), draw=black, line width=2pt, inner sep= 10pt, dashed] {};

\node (gadget_label) [above left = -0.25cm and 1cm of ve] {\large gadget: \LARGE $\texttt{T}_e$};

\end{tikzpicture}
\caption{\label{figure:bounded-degree}
The gadget $\scte$ contains element vertex $v_e$, the dummy vertices
($x_{e,i}$) of element $e$, and the set vertex of each set $S_j$ that contains $e$. The
indices $j(i)$ for $i = 1, \cdots, f_e$ are the indices of the sets
$S_{j(i)}$ that contain $e$. Note that for different elements $e$ and $e'$, the only vertices contained in $\scte \cap \mathtt{T}_{e'}$ are necessarily set vertices.}
\end{figure}

\mypar{Mapping a graph deletion set to a \SCfull instance.}
Next we construct a function $g$ that maps a solution for instance $\yB$ of $\bBDDV$ to a solution $g(\yB)$ to $\SC$.

\begin{definition}\label{def:map-bdd-to-sc}
  Given an instance $(\mcU, \mcF)$ of \SC,
  let $G = (V,E)$ be the corresponding set cover graph constructed as in Definition~\ref{def:sc-bdd-gadget}, and again let $\Td=\max\{\Delta,f_{\max}\}$.
  For any fixed, arbitrary, feasible solution $\yB \subset V(G)$ to $\bBDDV(G)$, we construct a \bBDDV-to-\SC map, $g$, that maps a feasible solution $\yB$ for $\bBDDV(G)$ to a solution $g(\yB)$ for $\SC(\mcU,\mcF)$ in two steps:
  \begin{enumerate}
    \itemsep0em
    \item Construct a \emph{canonical edit set} $\ycan$ as follows: for each $v_i \in \yB$ if $v_i$ is a set vertex for set $S_j$, then add $v_i$ to $\ycan$.
    Otherwise it is necessarily an element vertex $v_e$ or a dummy vertex $x_{e,j}$ for some element $e$; the corresponding gadget $\scte$ must contain at least one set vertex, $v(S_j)$. Choose one such set vertex $v(S_j) \in \scte$ arbitrarily and add $v(S_j)$ to $\ycan$.
    \item Set $g(\ycan) = \{ S_j \in \mcU \mid v(S_j) \in \ycan \}$.
  \end{enumerate}
\end{definition}

\begin{lemma}\label{lemma:element-cover-deg}
  Given an instance $(\mcU,\mcF)$ of \SC, let $G = (V,E)$ be the graph $\scdg$ constructed in Definition~\ref{def:sc-bdd-gadget}.
  If $\yB$ is an arbitrary feasible solution of $\bBDDV(G)$, then for each $e\in \mcU$, $\scte \cap \yB$ is non-empty.
\end{lemma}
\begin{proof}
If $\scte\cap \yB = \emptyset$ for an element $e$, then $\deg(v_e)$ in $G[V\setminus \yB]$ is
$\Td+1$ which contradicts the fact that $\yB$ is a feasible solution of
$\bBDDV(G)$.
\end{proof}
\begin{lemma}\label{lemma:edge-set-sol-deg}
  Let $G = (V,E)$ be the graph $\scdg$ for a \SC instance $(\mcU,\mcF)$, and $Z \subseteq V(G)$.
  If, for each element $e \in \mcU$,
  $\scte\cap Z$ contains at least one set vertex, $S_j$, then $Z$ is a feasible solution of $\bBDDV(G)$.
\end{lemma}
\begin{proof}
    Note that, in the construction in Definition~\ref{def:sc-bdd-gadget}, the degree of each set vertex is at most
    $\Delta\leq \Td$ and each dummy vertex has degree exactly one.
    Moreover, the degree of each element vertex in $G$ is exactly $\Td +1$.
    Hence, if for each element $e$, $\scte\cap Z$ is non-empty (and in
    particular it contains a set vertex), then the degree of each element vertex in
    $G[V\setminus Z]$ is at most $\Td$, and so $Z$ is a feasible solution.
\end{proof}

\begin{lemma}\label{lemma:canonical-feasible-sc-bdd}
    In the notation of Definition~\ref{def:map-bdd-to-sc}, the canonical solution $\ycan$ corresponding to a solution $\yB$ is a feasible solution to $\bBDDV(G)$ and satisfies $|\ycan| \leq |\yB|$.
\end{lemma}
\begin{proof}
    Because $\yB$ is feasible, by Lemma~\ref{lemma:element-cover-deg} we know that for each $e \in \mcU$, $\yB$ contains at least one vertex $v \in \scte$.
    Suppose $\ycan$ is not feasible. Then for some $e\in \mcU$, $\scte \cap \ycan$ is empty, by the contrapositive of Lemma~\ref{lemma:edge-set-sol-deg}.
    Since $\yB$ is feasible, there exists some $v \in \scte \cap \yB$.
    Then by Definition~\ref{def:map-bdd-to-sc}, $\ycan$ must contain $v$ if $v$ is a set vertex, or else $\ycan$ must contain a set vertex in $\scte \cap \yB$. This is a contradiction, so $\ycan$ is feasible.

    Finally, to prove that $|\ycan| \leq |\yB|$, note that each vertex $v \in \yB$ is mapped to no more than one vertex in $\ycan$ in the construction in Definition~\ref{def:map-bdd-to-sc}.
\end{proof}

\begin{lemma}
    In the notation of Definition~\ref{def:map-bdd-to-sc}, $g(\ycan)$ is a feasible solution to $\SC(\mcU,\mcF)$ and satisfies $|g(\ycan)| \leq |\yB|$.
\end{lemma}
\begin{proof}
    Take any $e \in \mcU$.
    By Lemma~\ref{lemma:canonical-feasible-sc-bdd} we know $\ycan$ is a feasible solution to $\bBDDV(G)$ and so by Lemma~\ref{lemma:element-cover-deg} the set $\ycan \cap \scte$ is non-empty.
    Next, we note that by construction $\ycan$ contains only set vertices, and so there is some set $S_j$ such that $v(S_j) \in \ycan \cap \scte$.
    Then by Definition~\ref{def:map-bdd-to-sc}, we know $S_j \in g(\ycan)$.
    This shows every element $e \in \mcU$ is covered by at least one set $S_j$ in $g(\ycan)$, i.e., $g(\ycan)$ is a feasible solution to $\SC(\mcU,\mcF)$.

    Finally, $g(\ycan)$ contains no more than one set $S_j$ for each set vertex $v(S_j)$ in $\ycan$.
\end{proof}
\begin{corollary}\label{corollary:canonical-bdd-to-sc}
    In the notation of Definition~\ref{def:map-bdd-to-sc},
    $g(\ycan)$ satisfies $|g(\ycan)| \leq |\ycan| \leq |\yB|$.
\end{corollary}

\mypar{Strict reduction from \SCfull.}
To finish our proof of a strict reduction from \SC to \bBDDV, we want to show that the functions $f$ and $g$ introduced earlier in this subsection satisfy the sufficient conditions given in Definition~\ref{def:strict-reduction}.

Given a \SC instance $(\mcU,\mcF)$ and any feasible solution $\yB$ to $f(\mcU,\mcF) = \bBDDV(G)$, by Corollary~\ref{corollary:canonical-bdd-to-sc} we have that $g(\ycan)$ is a solution to $(\mcU,\mcF)$ with $|g(\ycan)| \leq |\yB|$.
\begin{corollary}\label{cor:bddv-cost-ineq}
  In the notation of Definition~\ref{def:map-bdd-to-sc}, $\cost_{\SC}(g(\ycan)) \leq \cost_{\bBDDV}(\yB)$.
\end{corollary}
\begin{lemma}\label{lemma:bddv-opt-eq}
  In the notation of Definition~\ref{def:sc-bdd-gadget}, 
  $\opt_{\SC}(\mcU,\mcF) = \opt_{\bBDDV}(f(\mcU,\mcF))$.
\end{lemma}
\begin{proof}
  From Corollary~\ref{cor:bddv-cost-ineq}, we know that if $\yB$ is optimal, then $\opt_{\bBDDV}(f(\mcU,\mcF)) = |\yB|$.
  Since $\opt_{\SC}(\mcU,\mcF) \leq |g(\ycan)|$ by optimality,  we have
  $\opt_{\SC}(\mcU,\mcF) \leq \opt_{\bBDDV}(f(\mcU,\mcF))$.

  Conversely, given any solution $\mathcal{S}$ to $\SC(\mcU,\mcF)$, let $Z \subset V(G)$ consist of each set vertex whose set is in $\mathcal{S}$. Then $|Z| \leq |\mathcal{S}|$, and by Lemma~\ref{lemma:edge-set-sol-deg} we know $Z$ is a feasible solution to $\bBDDV(G)$.
  Thus, we have $\opt_{\SC}(\mcU,\mcF) \geq \opt_{\bBDDV}(f(\mcU,\mcF))$.
\end{proof}
We can now prove the main theorem.
\begin{proof}[Proof of Theorem~\ref{theorem:hardness-degree}]
    Corollary~\ref{cor:bddv-cost-ineq} and Lemma~\ref{lemma:bddv-opt-eq} provide the approximation ratio requirements for a strict reduction given in Definition~\ref{def:strict-reduction}.
    Moreover, the mappings $f$ and $g$ are polynomial in $|\mcU| + |\mcF|$, and so we have demonstrated a strict reduction.

    Finally, we derive a lower bound on the approximation ratio for \bBDDV.
    Recall that $\Delta = \max_{S\in \mcF} |S|$, $f_{\max} = \max\{f_e : e \in \mcU\}$,
    and $\Td=\max\{\Delta,f_{\max}\}$.
    By Theorem~\ref{lem-set-cover-sparse}, there exists a constant $C >0$ such that it is NP-hard to approximate \SC within a factor of $(\ln \Delta - C\cdot \ln\ln \Delta)$.
    Moreover, we can assume $\Delta \geq f_{\max}$
    and so we can express $\Td = \Delta$.
    Thus, by our strict reduction above, it is NP-hard to approximate $\bBDDV$ within a factor of $(\ln \Td - C \cdot \ln\ln \Td)$.
\end{proof}

\fi

\ifdefined\isstarforesthard
    \subsection{Treedepth 2: $(2-\epsilon)$-Inapproximability of Vertex Editing }
    \label{section:negative_starforest_v}
    In this section, we prove the following hardness result on \SFVfull.
\begin{theorem}\label{thm:star-forest-inapprox}
\SFV is $(2-\epsilon)$-inapproximable assuming the Unique Games Conjecture.
\end{theorem}

\mypar{Reduction strategy.}
Given an instance of \VC on a graph $G$ we modify $G$ to create an instance of \SFV $f(G)$.
Then we exhibit a mapping between solutions of the two problems that we show preserves solution sizes, from which a strict reduction follows.
We proceed by introducing the aforementioned mapping.

\mypar{Mapping an instance of \VCfull to an instance of \SFVfull.}
\begin{definition}\label{def:vc-sf-gadget}
		Given an instance $G = (V,E)$ of \VCfull, we construct a \emph{SF graph gadget}
		$G' = (V',E')$ by subdividing the edges of $G$ once.
		We call the vertices in $G'$ created by the subdivision of edges in $G$ the \emph{subdivision vertices}, denoted $S \subset V'$,
		and refer to members of $V$ as \emph{inherited vertices} of $V'$.
		Let $n = |V|$.
		To each inherited vertex in $V'$, we attach $(2n + 1)$ new, pendant vertices in $G'$, which we call the \emph{auxiliary vertices} and denote $A \subset V'$.
\end{definition}

This defines a mapping $f(G) = G' = (V', E')$ (see Figure~\ref{fig:sf-reduction}) where $|V'| = O(|V|^2)$ and $|E'| = O(|E| + |V|^2)$.

\begin{figure}[!ht]
	\centering
	\includegraphics[width=\linewidth]{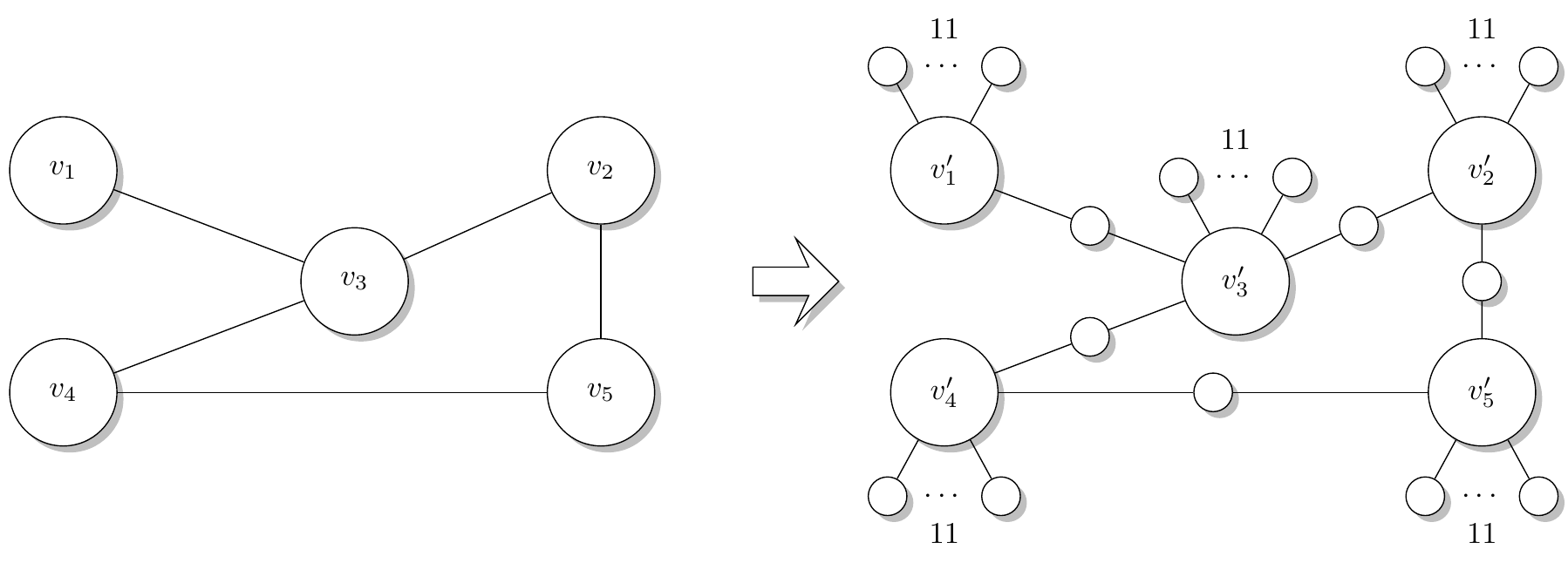}
	\caption{Given a graph $G=(V,E)$ as an instance of \VCfull (\emph{left}), we create a corresponding instance of \SFVfull, the graph $G'=(V', E')$ (\emph{right}), by subdividing every edge in $E$ and attaching $2|V|+1$ pendants to each vertex in $V$.
	}\label{fig:sf-reduction}
\end{figure}

\begin{lemma}\label{lemma:vc-to-sfv-sol}
	Given an instance $G$ of \VC, let $f(G)$ be the SF graph gadget (Definition~\ref{def:vc-sf-gadget}).
	Let $\yA$ be any feasible solution to $\VC(G)$.
	Then $\yA$ is a feasible solution to $\SFV(f(G))$.
\end{lemma}
\begin{proof}
	We first prove that, given a solution $\yA$ to $\VC(G)$, the vertex set $\yB = \yA$ is a solution for instance $\SFV(f(G))$:
	that is, $f(G) \setminus \yB$ is a star forest.

	Using the notation of Definition~\ref{def:vc-sf-gadget},
	for each vertex $w \in S$, since $\yA$ is a vertex cover of $G$, deleting $\yA$ from $g'$ means at least one of the adjacent vertices to $w$ is removed from $G'$.
	Thus, $w$ becomes a leaf of the star graph with center given by the other vertex adjacent to $w$ (or, if both adjacent vertices are removed, $w$ becomes a center).
	For all vertices that are removed, their corresponding neighbors in $A$ become centers.
	For all vertices that are not removed, the corresponding vertices in $A$ that are attached to them are leaves. Thus, $\yB = \yA \subset V \subset V'$ is a feasible solution to $\SFV(G')$.
\end{proof}
Since this applies to all solutions of $\VC(G)$, in particular it applies to optimal solutions.
\begin{corollary}\label{cor:vc-to-sfv-sol}
	In the notation of Lemma~\ref{lemma:vc-to-sfv-sol}, $\opt_{\VC}(G) \geq \opt_{\SFV}(f(G))$.
\end{corollary}

\mypar{Mapping a graph deletion set to a \VCfull instance.}

\begin{lemma}\label{lem:opt-equiv}
		Given an instance $G$ of \VC, let $G' = f(G)$ from Definition~\ref{def:vc-sf-gadget}.
		Given an optimal solution $\yopt$ to $\SFV(G')$, there exists a solution $g(\yopt)$ to $\VC(G)$ such that $|g(\yopt)| = |\yopt|$.
\end{lemma}
\begin{proof}
	Let $\yB$ be any feasible solution to \SFV on $G'$.
	Then for each vertex $v \in \yB$, each vertex adjacent to $v$ either becomes a center of a star or a leaf of a star.
	We will show that no optimal solution $\yopt$ for $\SFV(G')$ contains any vertices in $S \cup A$.

	To start, we show that, given a valid optimal solution $\yopt$ for $\SFV(G')$, we can obtain an optimal solution $C$ for $\SFV(G')$ where $|C| = |\opt_{\SFV}(G')|$.
	We note that $\opt_{\SFV}(G) \leq n$ for all $G$.

	Suppose $\yopt \cap S$ is non-empty and let $v \in \yopt \cap S$.
	Let $e_1$ and $e_2$ be the two neighbors of $v$ in $G'$.
	Then, both $e_1$ and $e_2$ become leaves of stars in $G'[V'-\yopt]$.
	Assume not; then one or both of $e_1$ and $e_2$ are centers.
	If one of $e_1$ and $e_2$ is a leaf then we can delete it instead of $v$, or if they are both centers we can delete either one instead of $v$, resulting in no vertices from $S$ being deleted in some optimal solution, a contradiction.
	Then $e_1$ and $e_2$ have degree 1 in $G'[V'-\yopt]$, they must each have degree $1$.
	By construction of $G'$, $e_1$ and $e_2$ are inherited vertices; thus, each is connected to $2n + 1$ auxiliary vertices in $G'$.
	In order for $e_1$ and $e_2$ to become leaf vertices in $G'[V'-\yopt]$, $4n + 2$ vertices must be deleted.
	Thus, there exists an optimal solution $\yopt$ to $\SFV(G')$ such that $S \cap \yopt$ is empty.

	Next we show that for any optimal solution $\yopt$ to $\SFV(G')$, there exists another optimal solution $\yopt_2$ such that $\yopt_2 \cap A$ is empty.
	If $v \in \yopt$ for some auxiliary vertex $v \in A$, then let $u$ be the inherited vertex to which $v$ is attached, and set $\yopt_2 = \yopt \cup \{u\} \setminus \{v\}$.
	Note that $\yopt_2$ is necessarily still a feasible solution to $\SFV(G')$, and $|\yopt_2| \leq |\yopt|$.

	This shows there exists an optimal solution $\yopt$ to $\SFV(G')$ that contains only inherited vertices, i.e., $\yopt \subset V$.
	Next we show that $C = \yopt$ is a valid vertex cover of $G$.
	In order for $G'[V'-\yopt]$ to be a star forest, all vertices in $S$ must be leaves of $G'[V'-\yopt]$.
	This means that at least one vertex adjacent to each vertex in $S$ is in $\yopt$.
	Equivalently, $\yopt$ touches every edge in $G$ and is a solution to $\VC(G)$ with size $\opt_{\SFV}(f(G))$, as desired.
\end{proof}
Note that Lemma~\ref{lem:opt-equiv} implies $\opt_{\VC}(G) \leq \opt_{\SFV}(f(G))$, so by Corollary~\ref{cor:vc-to-sfv-sol} we have the following.
\begin{corollary}
		In the notation of Lemma~\ref{lemma:vc-to-sfv-sol}, $\opt_{\VC}(G) = \opt_{\SFV}(f(G))$.
\end{corollary}

\mypar{Strict reduction from \VCfull.}
To finish our proof of a strict reduction from \VCfull to \SFVfull, we want to prove our reduction satisfies the sufficient conditions given in Definition~\ref{def:strict-reduction}, then use our reduction and Lemma~\ref{lem:opt-equiv} to prove the main theorem of this section.

\begin{proof}[Proof of Theorem~\ref{thm:star-forest-inapprox}]
Let $G' = (V',E') = f(G)$ and let $\yB$ be any feasible deletion set for $\SFV(G')$.
From $\yB$ we will construct a canonical solution $\ycan$ to $\SFV(G')$ consisting only of vertices from $V$, and show that $\ycan$ is then a solution to $\VC(G)$.

Recall from the construction of $G' = f(G)$ that $V'$ is partitioned into $V, A,$ and $S$. Thus, the vertices in $\yB$ must fall into these three sets.
For each $v \in \yB$,
\begin{enumerate}
    \item if $v \in V$ add $v$ to $\ycan$,
    \item if $v \in A$ then $v \in N(u)$ for some $u \in V$; add $u$ to $\ycan$.
    \item if $v \in S$, then $v$'s two neighbors are in $V$; choose one at random and add it to $\ycan$.
\end{enumerate}

Note that this construction guarantees that $\ycan$ contains at most one vertex for each vertex in $\yB$, and so $|\ycan| \leq |\yB|$.
Furthermore, $\ycan \subseteq V$.

It remains to show that $\ycan$ is a solution to $\VC(G)$.
Suppose not; then $G[V\setminus \ycan]$ contains an edge. Let the edge have endpoints $u,v \in V\setminus \ycan$.
Then in $G'$, the nodes $u,v$ are each connected to a subdivided node, $s \in S$. Since neither $u$ nor $v$ is in $\ycan$, by construction $u,s,v$ are not in $\yB$. Moreover, the construction of $\ycan$ also guarantees that, since $u,v \notin \ycan$, none of the auxiliary vertices attached to $u,v$ in $G'$ are contained in $\yB$.
This implies that $G'[V' \setminus \yB$] contains vertices $u,v,s,$ and all auxiliary vertices attached to $u$ and $v$, proving that $\yB$ is not a solution to $\SFV(G')$.

This proves $\cost_{\VC}(g(\yB)) \leq \cost_{\SFV}(\yB)$, completing a strict reduction from $\VC$ to $\SFV$.
\end{proof}

    \subsection{Treedepth 2: APX-hardness of Edge Editing}
    \label{section:negative_starforest_e}
    
Here we prove that deleting \emph{edges} to make a graph a star forest is APX-hard.
We proceed via a reduction from \textsc{Minimum Dominating Set-$B$}.

\begin{theorem}\label{thm:apx-hardness-sfe}
	\SFEfull is APX-complete, i.e.\ there exists a constant $\epsilon > 0$ such that it is NP-hard to approximate \SFE to within a factor of $1 + \epsilon$.
\end{theorem}

\mypar{Reduction strategy.}
We exhibit a mapping between a solution for \DSB and \SFE on the same graph, and prove the mapping meets the requirements of an L-reduction.

\mypar{Mapping an instance of \DSBfull to an instance of \SFEfull.}
Given an instance $G=(V,E)$ for \DSBfull, we define $f(G)=G$. To clarify when we consider $G$ as an instance of $\SFE$, we refer to it as $G' = (V',E')$.

\mypar{Mapping a graph deletion set to a \DSBfull instance.}
To define a map $g$ that converts a solution $\yB$ for $\SFE(f(G))$ into a solution $g(\yB)$ for $\DSB(G)$, first recall that for a feasible solution $\yB$ to $\SFE(G')$, $G'[E' \setminus \yB]$ is a star forest.
We define $g(\yB)$ to be the set of centers of the stars in $G'[E' \setminus \yB]$, including isolates.

We now prove that $g$ does in fact give a feasible solution to $\DSB(G)$.

\begin{lemma}\label{lem:star-dom-set-equiv}
		Given an instance $G = (V, E)$ of \DSBfull, let $n = |V|$,  $m = |E|$, and $G' = f(G)$ denote the corresponding instance of \SFEfull.
		Given any feasible solution $\yB$ for $\SFE(G')$, $g(\yB)$ is a feasible solution to $\DSB(G)$ and satisfies $|g(\yB)| = n - m +|\yB|$.
		In addition, given a valid dominating set $D \subseteq V$, we obtain a valid edge edit set $Z$ for the graph $G$ to a star forest where $|Z| = m - n + |D|$.
\end{lemma}
\begin{proof}
		Given an edge deletion set $\yB$ for $\SFE(G')$, $g(\yB)$ consists of the centers of the star graphs in $G'[E'\setminus \yB]$.
		Thus, every vertex $u \in G'[E'\setminus \yB]$ is in a star. If $u$ is the center of a star, then $u$ is in $g(\yB)$; otherwise, $u$ is adjacent to a star center, $v$, which is necessarily in $g(\yB)$. Thus, the vertex set $g(\yB)$ dominates all vertices of $G'[E'\setminus\yB]$, and therefore of $G$.

		Next, observe that the number of edges in $G'[E'\setminus\yB]$ is $m - |\yB|$.
		For every vertex $u$ in $G'[E'\setminus\yB]$ that is not a center, we use one edge to connect $u$ to its center.
		Hence, the number of edges in $G'[E'\setminus\yB]$ is equal to the number of vertices which are not centers.
		Since $|g(\yB)|$ equals the number of centers exactly, we have $n - |g(\yB)| = m - |\yB|$.

		Given a dominating set $D$ for $G$, every vertex in $V$ is either in $D$ or adjacent to a vertex in $D$. Thus, we can remove all edges that do not connect a vertex in $V \setminus D$ to a vertex in $D$. If a vertex in $V \setminus D$ is connected to more than one vertex in $D$, remove all but one of the edges to a vertex in $D$. Also remove all edges between vertices $v_1, v_2 \in V \setminus D$ and between vertices $d_1, d_2 \in D$.
		The remaining graph is a star forest since each vertex $v \in V \setminus D$ is connected to at most one vertex and all vertices in $D$ are centers and are not connected to any other centers.
		Finally, let $Z$ be the set of deleted edges. To see that $|D| = n - m + |Z|$, observe that $g(Z) = D$.
\end{proof}

\mypar{L-reduction from \DSBfull.}
To show an L-reduction from \DSB to \SFE we must meet the conditions from Definition~\ref{definition:lreduction}.

\begin{lemma}\label{lem:alpha-ratio}
		For any fixed $B > 0$, there exists a constant $\alpha$ such that for all degree-$B$ bounded graphs $G$ we have $\opt_{\SFE}(f(G)) \leq \alpha \cdot \opt_{\DSB}(G)$.
\end{lemma}
\begin{proof}
		Let $G = (V,E)$ be a degree-$B$ bounded graph and set $G' = f(G)$. Then $\opt_{\SFE}(G') \leq |E| \leq Bn/2$ and $\opt_{\DSB}(G) \geq n/(B+1)$ (since each node in the dominating set can cover at most $B$ other nodes).
		Therefore, any constant $\alpha$ satisfying the condition $\alpha \geq \tfrac{Bn/2}{n/(B+1)} = B(B+1)/2$ also satisfies $\opt_{\SFE}(G') \leq \alpha \cdot \opt_{\DSB}(G)$.
\end{proof}

\begin{lemma}\label{lem:beta-gap}
	Given an instance $\DSB(G)$, and any feasible solution $\yB$ to instance $G'=f(G)$ of $\SFV$,  we have $\cost_{\DSB}(g(\yB)) - \opt_{\DSB}(G) = \cost_{\SFE}(\yB) - \opt_{\SFE}(G') $.
\end{lemma}

\begin{proof}
	For any solution $\yB$ to $\SFE(G')$,
	Lemma~\ref{lem:star-dom-set-equiv} yields the equations $\opt_{\DSB}(G) = n - m + \opt_{\SFE}(G')$ and $\cost_{\DSB}(g(\yB)) = n-m+\cost_{\SFE}(\yB)$.
	Combining these, we have
	$\cost_{\DSB}(g(\yB)) - \opt_{\DSB}(G) = \cost_{\SFE}(\yB) - \opt_{\SFE}(G')$, as desired.
\end{proof}

Finally, we use Theorem~\ref{thm:ds-hard-approx} to complete our proof of Theorem~\ref{thm:apx-hardness-sfe}.
\begin{proof}[Proof of Theorem~\ref{thm:apx-hardness-sfe}]
First note that the functions $f$ and $g$ defined earlier in this subsection give a polynomial time reduction from \DSB to \SFE and satisfy the feasibility conditions of an L-reduction.
Next we show $f$ and $g$ satisfy the approximation inequalities.

Given Lemmas~\ref{lem:alpha-ratio} and~\ref{lem:beta-gap}, if \SFE can be approximated to a ratio of $1+\epsilon$ for some constant $\epsilon > 0$, then we show that \DSB can be approximated to a ratio of $(1 + \frac{\epsilon B(B+1)}{2})$.

The approximation ratio for \DSB is $\frac{\cost_{\DSB}(g(\yB))}{\opt_{\DSB}(G)}$. Substituting from Lemmas~ \ref{lem:alpha-ratio} and \ref{lem:beta-gap}, we obtain:
\begin{align*}
\frac{\cost_{\DSB}(g(\yB))}{\opt_{\DSB}(G)} &\leq \frac{\opt_{\DSB}(G) + (\cost_{\SFE}(\yB) - \opt_{\SFE}(G'))}{\opt_{\DSB}(G)}\\
 &\leq 1 + \alpha \left(\frac{\cost_{\SFE}(\yB) - \opt_{\SFE}(G')}{\opt_{\SFE}(G')}\right),
\end{align*}
which is bounded above by $1 + \epsilon B(B+1)/2$, proving that \SFE is APX-hard.

\end{proof}

\fi

  \section{Open Problems}

We hope that our framework for extending approximation algorithms from
structural graph classes to graphs near those classes, by editing to the class
and lifting the resulting solution, can be applied to many more contexts.
Specific challenges raised by this work include the following:

\begin{enumerate}
\item Editing via edge contractions.  Approximation algorithms for this type of
  editing would enable the framework to apply to the many optimization problems
  closed under just contraction, such as \textsc{TSP Tour} and \textsc{Connected Vertex Cover}.
\item Editing to $H$-minor-free graphs.
  Existing results apply only when $H$ is planar \cite{fomin2012planar}.
  According to Graph Minor Theory, the natural next steps are when
  $H$ can be drawn with a single crossing, when $H$ is an apex graph
  (removal of one vertex leaves a planar graph), and when $H$ is an
  arbitrary graph (say, a clique).  $H$-minor-free graphs have many PTASs
  (e.g., \cite{demaine2005bidimensionality,demaine2011contraction})
  that would be exciting to extend via structural rounding.
\item Editing to bounded clique number and bounded weak $c$-coloring number.
  While we have lower bounds on approximability,
  we lack good approximation algorithms.
\end{enumerate}

\else
\fi

\ifdefined\isall
    \section*{Acknowledgments}
    We thank Abida Haque and Adam Hesterberg for helpful initial discussions,
    Nicole Wein for providing helpful references on bounded degeneracy problems,
    and Michael O'Brien for helpful comments on the manuscript.

    This research was supported in part by the Army Research Office
    under Grant Number W911NF-17-1-0271 to Blair D. Sullivan,
    the Gordon \& Betty Moore Foundation's Data-Driven Discovery Initiative
    under Grant GBMF4560 to Blair D. Sullivan, as well as NSF grants CCF-1161626 and IIS-1546290 to Erik D. Demaine.
    Timothy D. Goodrich is partially supported by a National Defense
    Science \& Engineering Graduate Program fellowship.
    Quanquan Liu is partially supported by a National Science Foundation Graduate Research Fellowship
  	under grant 1122374.

    The views and conclusions contained in this
    document are those of the authors and should not be interpreted as representing the official policies, either
    expressed or implied, of the Army Research Office or the U.S. Government. The U.S. Government is
    authorized to reproduce and distribute reprints for Government purposes notwithstanding any copyright
    notation herein.

\fi

\bibliography{ref}

\begin{thebibliography}{10}

\bibitem{alber01improved}
J.~Alber and R.~Niedermeier.
\newblock Improved tree decomposition based algorithms for domination-like
  problems.
\newblock In {\em Proceedings of the 5th Latin American Symposium on
  Theoretical Informatics}, pages 613--627. Springer, 2001.

\bibitem{amir2010approximation}
E.~Amir.
\newblock Approximation algorithms for treewidth.
\newblock {\em Algorithmica}, 56(4):448--479, 2010.

\bibitem{bansal2017lp}
N.~Bansal, D.~Reichman, and S.~W. Umboh.
\newblock {LP}-based robust algorithms for noisy minor-free and bounded
  treewidth graphs.
\newblock In {\em Proceedings of the 28th Annual ACM-SIAM Symposium on Discrete
  Algorithms}, pages 1964--1979. Society for Industrial and Applied
  Mathematics, 2017.

\bibitem{bansal2017tight}
N.~Bansal and S.~W. Umboh.
\newblock Tight approximation bounds for dominating set on graphs of bounded
  arboricity.
\newblock {\em Information Processing Letters}, 122:21--24, 2017.

\bibitem{bar2004local}
R.~Bar-Yehuda, K.~Bendel, A.~Freund, and D.~Rawitz.
\newblock Local ratio: A unified framework for approximation algorithms.
\newblock {\em ACM Computing Surveys}, 36(4):422--463, 2004.

\bibitem{biedl2016crossing}
T.~Biedl, M.~Chimani, M.~Derka, and P.~Mutzel.
\newblock Crossing number for graphs with bounded pathwidth.
\newblock In {\em 28th International Symposium on Algorithms and Computation,
  (ISAAC)}, pages 13:1--13:13, Phuket, Thailand, December 2017.

\bibitem{bodlaender1988dynamic}
H.~L. Bodlaender.
\newblock Dynamic programming on graphs with bounded treewidth.
\newblock In T.~Lepist{\"o} and A.~Salomaa, editors, {\em Automata, Languages
  and Programming}, pages 105--118, Berlin, Heidelberg, 1988. Springer Berlin
  Heidelberg.

\bibitem{bodlaender1995approximating}
H.~L. Bodlaender, J.~R. Gilbert, H.~Hafsteinsson, and T.~Kloks.
\newblock Approximating treewidth, pathwidth, frontsize, and shortest
  elimination tree.
\newblock {\em Journal of Algorithms}, 18(2):238--255, 1995.

\bibitem{borradaile_et_al:LIPIcs:2017:6919}
G.~Borradaile and H.~Le.
\newblock {Optimal Dynamic Program for $r$-Domination Problems over Tree
  Decompositions}.
\newblock In J.~Guo and D.~Hermelin, editors, {\em Proceedings of the 11th
  International Symposium on Parameterized and Exact Computation}, volume~63 of
  {\em Leibniz International Proceedings in Informatics (LIPIcs)}, pages
  8:1--8:23, Dagstuhl, Germany, 2017. Schloss Dagstuhl--Leibniz-Zentrum fuer
  Informatik.

\bibitem{cai2003parameterized}
L.~Cai.
\newblock Parameterized complexity of vertex colouring.
\newblock {\em Discrete Applied Mathematics}, 127(3):415--429, 2003.

\bibitem{chan2012approximation}
T.~M. Chan and S.~Har-Peled.
\newblock Approximation algorithms for maximum independent set of pseudo-disks.
\newblock {\em Discrete \& Computational Geometry}, 48(2):373--392, 2012.

\bibitem{Chekuri-Chuzhoy-2016}
C.~Chekuri and J.~Chuzhoy.
\newblock Polynomial bounds for the grid-minor theorem.
\newblock {\em Journal of the ACM}, 63(5):40:1--40:65, Dec. 2016.

\bibitem{chekuri2013approximation}
C.~Chekuri and A.~Sidiropoulos.
\newblock Approximation algorithms for {E}uler genus and related problems.
\newblock In {\em Proceedings of the IEEE 54th Annual Symposium on Foundations
  of Computer Science}, pages 167--176. IEEE, 2013.

\bibitem{chlebik2008approximation}
M.~Chleb{\'{\i}}k and J.~Chleb{\'{\i}}kov{\'{a}}.
\newblock Approximation hardness of dominating set problems in bounded degree
  graphs.
\newblock {\em Information and Computation}, 206(11):1264--1275, 2008.

\bibitem{chrobak1991planar}
M.~Chrobak and D.~Eppstein.
\newblock Planar orientations with low out-degree and compaction of adjacency
  matrices.
\newblock {\em Theoretical Computer Science}, 86(2):243--266, 1991.

\bibitem{chuzhoy2011algorithm}
J.~Chuzhoy.
\newblock An algorithm for the graph crossing number problem.
\newblock In {\em Proceedings of the 43rd Annual ACM Symposium on Theory of
  Computing}, pages 303--312, 2011.

\bibitem{chuzhoy2011graph}
J.~Chuzhoy, Y.~Makarychev, and A.~Sidiropoulos.
\newblock On graph crossing number and edge planarization.
\newblock In {\em Proceedings of the 22nd Annual ACM-SIAM Symposium on Discrete
  Algorithms}, pages 1050--1069, 2011.

\bibitem{cygan2011solving}
M.~Cygan, J.~Nederlof, M.~Pilipczuk, M.~Pilipczuk, J.~M. van Rooij, and J.~O.
  Wojtaszczyk.
\newblock Solving connectivity problems parameterized by treewidth in single
  exponential time.
\newblock In {\em Foundations of computer science (focs), 2011 ieee 52nd annual
  symposium on}, pages 150--159. IEEE, 2011.

\bibitem{dabrowski2015editing}
K.~K. Dabrowski, P.~A. Golovach, P.~van't Hof, D.~Paulusma, and D.~M. Thilikos.
\newblock Editing to a planar graph of given degrees.
\newblock In {\em Computer Science--Theory and Applications}, pages 143--156.
  Springer, 2015.

\bibitem{demaine2005bidimensionality}
E.~D. Demaine and M.~Hajiaghayi.
\newblock Bidimensionality: New connections between {FPT} algorithms and
  {PTASs}.
\newblock In {\em Proceedings of the 16th Annual ACM-SIAM Symposium on Discrete
  Algorithms}, pages 590--601, Philadelphia, PA, USA, 2005. Society for
  Industrial and Applied Mathematics.

\bibitem{demaine2011contraction}
E.~D. Demaine, M.~Hajiaghayi, and K.~Kawarabayashi.
\newblock Contraction decomposition in {$H$}-minor-free graphs and algorithmic
  applications.
\newblock In {\em Proceedings of the 43rd Annual ACM Symposium on Theory of
  Computing}, pages 441--450, 2011.

\bibitem{dinur2005new}
I.~Dinur, V.~Guruswami, S.~Khot, and O.~Regev.
\newblock A new multilayered {PCP} and the hardness of hypergraph vertex cover.
\newblock {\em SIAM Journal on Computing}, 34(5):1129--1146, 2005.

\bibitem{dinur2005hardness}
I.~Dinur and S.~Safra.
\newblock On the hardness of approximating minimum vertex cover.
\newblock {\em Annals of {M}athematics}, pages 439--485, 2005.

\bibitem{dinur2014analytical}
I.~Dinur and D.~Steurer.
\newblock Analytical approach to parallel repetition.
\newblock In {\em Proceedings of the 46th Annual ACM Symposium on Theory of
  Computing}, pages 624--633, 2014.

\bibitem{downey2013fundamentals}
R.~G. Downey and M.~R. Fellows.
\newblock {\em Fundamentals of parameterized complexity}, volume~4.
\newblock Springer, 2013.

\bibitem{drange2015parameterized}
P.~G. Drange.
\newblock {\em Parameterized Graph Modification Algorithms}.
\newblock PhD thesis, The University of Bergen, 2015.

\bibitem{drange2015threshold}
P.~G. Drange, M.~S. Dregi, D.~Lokshtanov, and B.~D. Sullivan.
\newblock On the threshold of intractability.
\newblock In {\em Proceedings of the 23rd Annual European Symposium on
  Algorithms}, pages 411--423, Patras, Greece, September 2015.

\bibitem{drange2016compressing}
P.~G. Drange, M.~S. Dregi, and R.~B. Sandeep.
\newblock Compressing bounded degree graphs.
\newblock In {\em {LATIN}}, volume 9644 of {\em Lecture Notes in Computer
  Science}, pages 362--375. Springer, 2016.

\bibitem{ebenlendrapproximation}
T.~Ebenlendr, P.~Kolman, and J.~Sgall.
\newblock An approximation algorithm for bounded degree deletion.
\newblock \url{http://kam.mff.cuni.cz/~kolman/papers/star.pdf}.

\bibitem{feige1998threshold}
U.~Feige.
\newblock A threshold of {$\ln n$} for approximating set cover.
\newblock {\em Journal of the ACM}, 45(4):634--652, July 1998.

\bibitem{feige2008improved}
U.~Feige, M.~Hajiaghayi, and J.~R. Lee.
\newblock Improved approximation algorithms for minimum weight vertex
  separators.
\newblock {\em SIAM Journal on Computing}, 38(2):629--657, 2008.

\bibitem{fiorini2010hitting}
S.~Fiorini, G.~Joret, and U.~Pietropaoli.
\newblock Hitting diamonds and growing cacti.
\newblock In {\em Proceedings of the International Conference on Integer
  Programming and Combinatorial Optimization}, pages 191--204, 2010.

\bibitem{fomin14tight}
F.~V. Fomin, S.~Kratsch, M.~Pilipczuk, M.~Pilipczuk, and Y.~Villanger.
\newblock Tight bounds for parameterized complexity of cluster editing with a
  small number of clusters.
\newblock {\em Journal of Computer and System Sciences}, 80(7):1430--1447,
  2014.

\bibitem{fomin2012planar}
F.~V. Fomin, D.~Lokshtanov, N.~Misra, and S.~Saurabh.
\newblock Planar {F}-deletion: Approximation, kernelization and optimal {FPT}
  algorithms.
\newblock In {\em Proceedings of the IEEE 53rd Annual Symposium on Foundations
  of Computer Science}, pages 470--479, 2012.

\bibitem{fujito1998unified}
T.~Fujito.
\newblock A unified approximation algorithm for node-deletion problems.
\newblock {\em Discrete Applied Mathematics}, 86(2-3):213--231, 1998.

\bibitem{gabow83efficient}
H.~N. Gabow.
\newblock An efficient reduction technique for degree-constrained subgraph and
  bidirected network flow problems.
\newblock In {\em Proceedings of the 15th Annual ACM Symposium on Theory of
  Computing}, pages 448--456, 1983.

\bibitem{gao2010survey}
X.~Gao, B.~Xiao, D.~Tao, and X.~Li.
\newblock A survey of graph edit distance.
\newblock {\em Pattern Analysis and Applications}, 13(1):113--129, 2010.

\bibitem{grohe2016colouring}
M.~Grohe, S.~Kreutzer, R.~Rabinovich, S.~Siebertz, and K.~Stavropoulos.
\newblock Colouring and covering nowhere dense graphs.
\newblock In E.~W. Mayr, editor, {\em Graph-Theoretic Concepts in Computer
  Science}, pages 325--338. Springer Berlin Heidelberg, 2016.

\bibitem{guha1996approximation}
S.~Guha and S.~Khuller.
\newblock Approximation algorithms for connected dominating sets.
\newblock In {\em European symposium on algorithms}, pages 179--193. Springer,
  1996.

\bibitem{guo2004structural}
J.~Guo, F.~H{\"u}ffner, and R.~Niedermeier.
\newblock A structural view on parameterizing problems: Distance from
  triviality.
\newblock In {\em International Workshop on Parameterized and Exact
  Computation}, pages 162--173. Springer, 2004.

\bibitem{gupta2018losing}
A.~Gupta, E.~Lee, J.~Li, P.~Manurangsi, and M.~W\l{}odarczyk.
\newblock Losing treewidth by separating subsets.
\newblock {\em CoRR}, abs/1804.01366, 2018.

\bibitem{guruswami2017inapproximability}
V.~Guruswami and E.~Lee.
\newblock Inapproximability of {H}-transversal/packing.
\newblock {\em SIAM Journal on Discrete Mathematics}, 31(3):1552--1571, 2017.

\bibitem{har2017approximation}
S.~Har{-}Peled and K.~Quanrud.
\newblock Approximation algorithms for polynomial-expansion and low-density
  graphs.
\newblock {\em SIAM Journal on Computing}, 46(6):1712--1744, 2017.

\bibitem{harvey2017parameters}
D.~J. Harvey and D.~R. Wood.
\newblock Parameters tied to treewidth.
\newblock {\em Journal of Graph Theory}, 84(4):364--385, 2017.

\bibitem{hochbaum1982approximation}
D.~S. Hochbaum.
\newblock Approximation algorithms for the set covering and vertex cover
  problems.
\newblock {\em SIAM Journal on computing}, 11(3):555--556, 1982.

\bibitem{huang17approximate}
D.~Huang and S.~Pettie.
\newblock Approximate generalized matching: $f$-factors and $f$-edge covers.
\newblock {\em CoRR}, abs/1706.05761, 2017.

\bibitem{huffner2015editing}
F.~H{\"u}ffner, C.~Komusiewicz, and A.~Nichterlein.
\newblock Editing graphs into few cliques: Complexity, approximation, and
  kernelization schemes.
\newblock In {\em Proceedings of the 14th International Symposium on Algorithms
  and Data Structures}, pages 410--421. Springer, 2015.

\bibitem{jansen2014near}
B.~M.~P. Jansen, D.~Lokshtanov, and S.~Saurabh.
\newblock A near-optimal planarization algorithm.
\newblock In {\em Proceedings of the 25th Annual ACM-SIAM Symposium on Discrete
  Algorithms}, pages 1802--1811, 2014.

\bibitem{johnson1974approximation}
D.~S. Johnson.
\newblock Approximation algorithms for combinatorial problems.
\newblock {\em Journal of Computer and System Sciences}, 9(3):256--278, 1974.

\bibitem{kawarabayashi2009planarity}
K.~Kawarabayashi.
\newblock Planarity allowing few error vertices in linear time.
\newblock In {\em Proceedings of the 50th Annual IEEE Symposium on Foundations
  of Computer Science}, pages 639--648. IEEE, 2009.

\bibitem{KawarabayashiS17}
K.~Kawarabayashi and A.~Sidiropoulos.
\newblock Polylogarithmic approximation for minimum planarization (almost).
\newblock In {\em Proceedings of the 58th {IEEE} Annual Symposium on
  Foundations of Computer Science}, pages 779--788, Berkeley, CA, October 2017.

\bibitem{khot2008vertex}
S.~Khot and O.~Regev.
\newblock Vertex cover might be hard to approximate to within $2-\varepsilon$.
\newblock {\em Journal of Computer and System Sciences}, 74(3):335--349, 2008.

\bibitem{kierstead2003orderings}
H.~A. Kierstead and D.~Yang.
\newblock Orderings on graphs and game coloring number.
\newblock {\em Order}, 20(3):255--264, 2003.

\bibitem{ko1982computational}
K.-I. Ko and H.~Friedman.
\newblock Computational complexity of real functions.
\newblock {\em Theoretical Computer Science}, 20(3):323--352, 1982.

\bibitem{komusiewicz2015}
C.~Komusiewicz.
\newblock Tight running time lower bounds for vertex deletion problems.
\newblock {\em {TOCT}}, 10(2):6:1--6:18, 2018.

\bibitem{kotrbvcik2016edge}
M.~Kotrb{\v{c}}{\'\i}k, R.~Kr{\'a}lovi{\v{c}}, and S.~Ordyniak.
\newblock Edge-editing to a dense and a sparse graph class.
\newblock In {\em Proceedings of the 12th Latin American Symposium on
  Theoretical Informatics}, pages 562--575. Springer, 2016.

\bibitem{krishnamoorthy1979node}
M.~S. Krishnamoorthy and N.~Deo.
\newblock Node-deletion {NP}-complete problems.
\newblock {\em SIAM Journal on Computing}, 8(4):619--625, 1979.

\bibitem{lee2017partitioning}
E.~Lee.
\newblock Partitioning a graph into small pieces with applications to path
  transversal.
\newblock In {\em Proceedings of the 28th Annual ACM-SIAM Symposium on Discrete
  Algorithms}, pages 1546--1558. Society for Industrial and Applied
  Mathematics, 2017.

\bibitem{lenzen2010minimum}
C.~Lenzen and R.~Wattenhofer.
\newblock Minimum dominating set approximation in graphs of bounded arboricity.
\newblock In {\em Proceedings of the International Symposium on Distributed
  Computing}, pages 510--524. Springer, 2010.

\bibitem{lewis1978complexity}
J.~M. Lewis.
\newblock On the complexity of the maximum subgraph problem.
\newblock In {\em Proceedings of the 10th Annual ACM Symposium on Theory of
  Computing}, pages 265--274. ACM, 1978.

\bibitem{lewis1980node}
J.~M. Lewis and M.~Yannakakis.
\newblock The node-deletion problem for hereditary properties is {NP}-complete.
\newblock {\em Journal of Computer and System Sciences}, 20(2):219--230, 1980.

\bibitem{lick1970degenerate}
D.~R. Lick and A.~T. White.
\newblock {$k$}-degenerate graphs.
\newblock {\em Canadian Journal of Mathematics}, 22:1082--1096, 1970.

\bibitem{lund1993approximation}
C.~Lund and M.~Yannakakis.
\newblock The approximation of maximum subgraph problems.
\newblock In {\em Proceedings of the International Colloquium on Automata,
  Languages, and Programming}, pages 40--51. Springer, 1993.

\bibitem{lund1994hardness}
C.~Lund and M.~Yannakakis.
\newblock On the hardness of approximating minimization problems.
\newblock {\em Journal of the ACM}, 41(5):960--981, 1994.

\bibitem{magen2009robust}
A.~Magen and M.~Moharrami.
\newblock Robust algorithms for on minor-free graphs based on the
  {S}herali-{A}dams hierarchy.
\newblock In {\em Approximation, Randomization, and Combinatorial Optimization.
  Algorithms and Techniques}, pages 258--271. Springer, 2009.

\bibitem{marx2006parameterized}
D.~Marx.
\newblock Parameterized coloring problems on chordal graphs.
\newblock {\em Theoretical Computer Science}, 351(3):407--424, 2006.

\bibitem{marx2008survey}
D.~Marx.
\newblock Parameterized complexity and approximation algorithms.
\newblock {\em The Computer Journal}, 51(1):60--78, 2008.

\bibitem{marx2012obtaining}
D.~Marx and I.~Schlotter.
\newblock Obtaining a planar graph by vertex deletion.
\newblock {\em Algorithmica}, 62(3-4):807--822, 2012.

\bibitem{mathieson2010parameterized}
L.~Mathieson.
\newblock The parameterized complexity of editing graphs for bounded
  degeneracy.
\newblock {\em Theoretical Computer Science}, 411(34):3181--3187, 2010.

\bibitem{mathieson2008parameterized}
L.~Mathieson and S.~Szeider.
\newblock Parameterized graph editing with chosen vertex degrees.
\newblock In {\em Combinatorial Optimization and Applications}, pages 13--22.
  Springer, 2008.

\bibitem{matula1983smallest}
D.~W. Matula and L.~L. Beck.
\newblock Smallest-last ordering and clustering and graph coloring algorithms.
\newblock {\em Journal of the ACM}, 30(3):417--427, 1983.

\bibitem{moshkovitz2012projection}
D.~Moshkovitz.
\newblock The projection games conjecture and the {NP}-hardness of {$\ln
  n$}-approximating set-cover.
\newblock In {\em Approximation, Randomization, and Combinatorial Optimization.
  Algorithms and Techniques}, pages 276--287. Springer, 2012.

\bibitem{moshkovitz2010two}
D.~Moshkovitz and R.~Raz.
\newblock Two-query pcp with subconstant error.
\newblock {\em Journal of the ACM}, 57(5):29, 2010.

\bibitem{nesetril2006treedepth}
J.~Ne{\v{s}}et{\v{r}}il and P.~O. de~Mendez.
\newblock Tree-depth, subgraph coloring and homomorphism bounds.
\newblock {\em European Journal of Combinatorics}, 27(6):1022--1041, 2006.

\bibitem{nesetril2011nowhere}
J.~Ne{\v{s}}et{\v{r}}il and P.~O. de~Mendez.
\newblock On nowhere dense graphs.
\newblock {\em European Journal of Combinatorics}, 32(4):600--617, 2011.

\bibitem{nesetril2012sparsity}
J.~Ne{\v{s}}et{\v{r}}il and P.~O. de~Mendez.
\newblock {\em Sparsity: Graphs, Structures, and Algorithms}.
\newblock Algorithms and Combinatorics. Springer Berlin Heidelberg, 2012.

\bibitem{okun2003new}
M.~Okun and A.~Barak.
\newblock A new approach for approximating node deletion problems.
\newblock {\em Information {P}rocessing {L}etters}, 88(5):231--236, 2003.

\bibitem{orponen1987approximation}
P.~Orponen and H.~Mannila.
\newblock On approximation preserving reductions: complete problems and robust
  measures.
\newblock Technical Report C-1987-28, Dept. of Computer Science, University of
  Helsinki, Finland, 1987.

\bibitem{papadimitriou1991optimization}
C.~H. Papadimitriou and M.~Yannakakis.
\newblock Optimization, approximation, and complexity classes.
\newblock {\em Journal of Computer and System Sciences}, 43(3):425--440, 1991.

\bibitem{reed1992finding}
B.~A. Reed.
\newblock Finding approximate separators and computing tree width quickly.
\newblock In {\em Proceedings of the 24th Annual ACM Symposium on Theory of
  Computing}, pages 221--228, 1992.

\bibitem{robertson1986graph}
N.~Robertson and P.~D. Seymour.
\newblock Graph minors. {II.} {A}lgorithmic aspects of tree-width.
\newblock {\em Journal of Algorithms}, 7(3):309--322, 1986.

\bibitem{robertson1995graph}
N.~Robertson and P.~D. Seymour.
\newblock Graph minors. {XIII.} {T}he disjoint paths problem.
\newblock {\em Journal of Combinatorial Theory, Series {B}}, 63(1):65--110,
  1995.

\bibitem{schaefer2013graph}
M.~Schaefer.
\newblock The graph crossing number and its variants: A survey.
\newblock {\em The Electronic Journal of Combinatorics}, 1000:21--22, 2013.

\bibitem{trevisan2001nonapproximability}
L.~Trevisan.
\newblock Non-approximability results for optimization problems on bounded
  degree instances.
\newblock In {\em Proceedings of the 33rd Annual ACM Symposium on Theory of
  Computing}, pages 453--461, 2001.

\bibitem{xiao2016degreeedit}
M.~Xiao.
\newblock A parameterized algorithm for bounded-degree vertex deletion.
\newblock In {\em Proceedings of the 22nd International Conference on Computing
  and Combinatorics}, pages 79--91, 2016.

\bibitem{yannakakis1978node}
M.~Yannakakis.
\newblock Node-and edge-deletion {NP}-complete problems.
\newblock In {\em Proceedings of the 10th Annual ACM Symposium on Theory of
  Computing}, pages 253--264, 1978.

\bibitem{zhu2009colouring}
X.~Zhu.
\newblock Colouring graphs with bounded generalized colouring number.
\newblock {\em Discrete Mathematics}, 309(18):5562--5568, 2009.

\end{thebibliography}
\bibliographystyle{abbrv}

  \appendix

\end{document}